%% file: quasifactorization.tex
\documentclass[aip,jmp,graphicx,11pt]{revtex4-2}
\usepackage{mathrsfs,amsmath,amsxtra,amssymb,amsthm,amsfonts,subcaption,mathtools}
\usepackage[english]{babel}

\newcommand{\mel}{\end{eqnarray*}}
\def\fr{\begin{align*}}

\newcommand{\supp}{\operatorname{supp}}

\newcommand{\p}{\hspace{.05cm}}
\renewcommand{\pl}{\hspace{.1cm}}

\newcommand{\pla}{\hspace{1.5cm}}

\renewcommand{\L}{{\mathcal L}}

\newcommand{\E}{{\mathcal E}}

\newcommand{\A}{{\mathcal A}}
\newcommand{\B}{{\mathcal B}}

\newcommand{\M}{{\mathcal M}}

\renewcommand{\S}{{\mathcal S}}
\newcommand{\T}{{\mathcal T}}

\newcommand{\N}{{\mathcal N}}
\newcommand{\G}{\Gamma}

\newtheorem{lemma}{Lemma}[section]
\newtheorem{prop}[lemma]{Proposition}
\newtheorem{theorem}[lemma]{Theorem}

\newtheorem{cor}[lemma]{Corollary}

\newtheorem{rem}[lemma]{Remark}
\newtheorem{definition}[lemma]{Definition}

\newcommand{\re}{\begin{rem}\rm}
\newcommand{\mar}{\end{rem}}
\newtheorem{exam}[lemma]{Example}
\newcommand{\bra}[1]{\langle{#1}|}
\newcommand{\ket}[1]{|{#1}\rangle}
\newcommand{\ketbra}[1]{|{#1}\rangle\langle{#1}|}
\newcommand{\qd}{\end{proof}\vspace{0.5ex}}
\newcommand{\prf}{\begin{proof}[\bf Proof:]}

\newcommand{\xspace}{\hbox{\kern-2.5pt}}

\newtheorem*{theorem*}{Theorem}

\allowdisplaybreaks

\usepackage{verbatim}
\usepackage{listings}
\usepackage{graphicx}
\usepackage{float}
\usepackage{amsfonts}
\usepackage{latexsym}
\usepackage{color}
\usepackage{url}
\usepackage{hyperref}
\usepackage{setspace}
\usepackage[normalem]{ulem}
\usepackage{braket}
\usepackage[toc,page]{appendix}
\usepackage{cancel}
\usepackage{enumitem}
\usepackage{ifthen}

\DeclareMathOperator{\tr}{tr}

	\newcommand{\RR}{\mathbb{R}}
	\newcommand{\CC}{\mathbb{C}}
	\newcommand{\NN}{\mathbb{N}}
	\newcommand{\MM}{\mathbb{M}}
	\newcommand{\BB}{\mathbb{B}}
	\renewcommand{\L}{\mathcal{L}}

	\newcommand{\id}{\hat{1}}

	\renewcommand{\G}{\mathcal{G}}

\DeclareRobustCommand*{\citen}[2][]{%
  \begingroup
    \romannumeral-`\x 
    \setcitestyle{numbers}%
    \cite[#1]{#2}%
  \endgroup   
}

\begin{document}
\title{Quasi-factorization and Multiplicative Comparison of Subalgebra-Relative Entropy}
\author{Nicholas LaRacuente}
\email{nicholas.laracuente@ucalgary.ca}
\affiliation{University of Chicago}
\begin{abstract}
\input{abstract}
\end{abstract}
\maketitle

\input{intro}

\input{relative_ent}

\input{log_constant}

\input{condex_combo}

\input{applications}

\input{thermal}

\input{conclusions}

\bibliographystyle{unsrt}
\bibliography{v2}

\appendix

\input{append}

\end{document}

%% file: abstract.tex
Purely multiplicative comparisons of quantum relative entropy are desirable but challenging to prove. We show such comparisons for relative entropies between comparable densities, including the relative entropy of a density with respect to its subalgebraic restriction. These inequalities are asymptotically tight in approaching known, tight inequalities as perturbation size approaches zero. Based on these results, we obtain a kind of inequality known as quasi-factorization or approximate tensorization of relative entropy. Quasi-factorization lower bounds the sum of a density's relative entropies to several subalgebraic restrictions in terms of its relative entropy to their intersection's subalgebraic restriction. As applications, quasi-factorization implies uncertainty-like relations, and with an iteration trick, it yields decay estimates of optimal asymptotic order on mixing processes described by finite, connected, undirected graphs.

%% file: intro.tex
\section{Introduction}
As relative entropy is at the core of quantum information theory, bounds and comparisons often yield results across many operational contexts. General bounds on relative entropy are however difficult to prove, as this quantity can be infinite and involves nonlinear functions of noncommuting operators. A motivation for such inequalities is to derive a strong form of decay estimate known as a modified logarithmic Sobolev inequality (MLSI) for quantum Markov semigroups. Decay inequalities characterize decoherence, noise, thermal relaxation, and related processes. To prove MLSI, however, requires inequalities that are non-trivial for arbitrarily small relative entropies - any density-independent, additive constant that holds in general will cause the inequality to reduce to positivity when involved relative entropies approach zero. In this paper, we derive asymptotically tight relative entropy inequalities without additive constants. We use these inequalities to prove a tightening of the entropic uncertainty principle and asymptotically optimal MLSI inequalities for processes described by finite graphs.

The strong subadditivity (SSA) of von Neumann entropy is a quintessential entropy inequality. SSA's impacts range from quantum Shannon theory \cite{wilde_quantum_2013} to holographic spacetime \cite{caceres_strong_2014}. Lieb and Ruskai proved SSA in 1973 \cite{lieb_proof_1973}. A later form by Petz in 1991 \cite{petz_certain_1991} generalizes from subsystems to subalgebras. Let $\M$ be a von Neumann algebra and $\N$ be a subalgebra. Associated with $\N$ is a unique conditional expectation $\E_\N$ that projects an operator in $\M$ onto $\N$. In tracial algebras, we denote by $\E_\N$ the unique conditional expectation from $\M$ to $\N$ that is self-adjoint with respect to the trace. We call subalgebras $\S, \T \subseteq \M$ a commuting square if $\E_\S \E_\T = \E_\T \E_\S = \E_{\S \cap \T}$, where $\E_{\S \cap \T}$ is the conditional expectation onto their intersection. The Umegaki relative entropy for matrices is defined by $D(\rho \| \sigma) = tr(\rho(\ln \rho - \ln \sigma))$ when $\text{supp}(\rho) \subseteq \supp(\sigma)$ and is infinite otherwise. In terms of Umegaki's relative entropy, SSA becomes:
\begin{theorem}[Petz's Conditional Expectation SSA] \label{thm:petzssa}
Let $\E_\S, \E_\T$ be conditional expectations to subalgebras $\S, \T \subseteq \M$. If $\S$ and $\T$ form a commuting square, then for all densities $\rho$,
\begin{equation} \label{eq:relcondexpsubadd}
D(\rho \| \E_\S(\rho)) + D(\rho \| \E_\T(\rho)) \geq D(\rho \| \E_{\S \cap \T}(\rho))
\end{equation}
\end{theorem}
As noted in \citen{gao_unifying_2017}, theorem \ref{thm:petzssa} implies SSA and an uncertainty relation for mutually unbiased bases, joining these concepts. In \citen{gao_unifying_2017}, it is shown that the condition $[\E_\S, \E_\T] = 0$ is necessary as well as sufficient for SSA in finite dimensions. That
\begin{equation} \label{eq:addcorr}
 D(\rho \| \E_\S(\rho)) + D(\rho \| \E_\T(\rho)) \geq D(\rho \| \E_{\S \cap \T}(\rho)) - c
\end{equation}
can only hold for all $\rho$ when $[\E_\S, \E_\T] \neq 0$ if $c > 0$. When conditional expectations don't commute, there are nonetheless inequalities approximating strong subadditivity. Perturbations of entropy and its inequalities often take additive forms \cite{berta_uncertainty_2010, winter_tight_2016, gao_capacity_2018, gao_uncertainty_2018}. If entropies are sufficiently small, then equation \eqref{eq:addcorr} becomes trivial and weaker than the statement that the left hand side is non-negative. Multiplicative bounds on relative entropy exist \cite{audenaert_continuity_2005, audenaert_continuity_2011, vershynina_upper_2019}, but these are constrained by the infinite divergence of entropy relative to states with smaller support.

The special case of $D(\rho \| \id/d)$, however, has range $[0, \ln d]$ on a system of dimension $d$. More generally, for a doubly stochastic conditional expectation $\E$, we refer to the function on a density $\rho$ given by $D(\rho \| \E(\rho))$ as \textit{subalgebra-relative entropy}. Having bounded range, subalgebra-relative entropy may support much stronger inequalities than general relative entropy. $D(\rho \| \id/d)$ is a special case of subalgebra-relative entropy, in which the algebra is $\CC 1$, that of the complex scalars. Subalgebra-relative entropy appears as measures of resources such as quantum coherence \cite{winter_operational_2016} and reference frame asymmetry \cite{vaccaro_tradeoff_2008, gour_measuring_2009, marvian_extending_2014}. One may write the conditional mutual information as $D(\rho \| \E_\S(\rho)) + D(\rho \| \E_\T(\rho)) - D(\rho \| \E_{\S \cap \T}(\rho))$ for subsystem-restricted algebras $\S$ and $\T$, which then applies to derived entanglement measures \cite{tucci_quantum_1999, tucci_entanglement_2002, christandl_squashed_2004}. Subalgebra-relative entropy is a natural measure of decoherence from processes that are self-adjoint with respect to the Hilbert-Schmidt inner product \cite{bardet_hypercontractivity_2018, gao_fisher_2020}. The maximum subalgebra-relative entropy for a given conditional expectation connects closely to the theory of subalgebra indices \cite{gao_relative_2020}. Hence subalgebra-relative entropy is fundamental to quantum information, motivating inequalities on this form.

More broadly, relative entropy should be comparable between densities that are comparable up to constants in the Loewner order, where these constants determine the strength of such comparisons. Again, there are simple ways to obtain bounds with additive corrections of this form. Multiplicative comparisons are more challenging but more powerful in many settings, especially when all of the relative entropies involved could be arbitrarily small. As an example in Section \ref{sec:ucr}, we show an uncertainty-like relation for incompatible projective measurements that remains non-trivial even for states approaching complete mixture. We contrast this with conventional entropic uncertainty relations, which usually reduce to positivity of relative entropy in these circumstances.

The notion of quasi-factorization for classical entropies was introduced and shown in \citen{cesi_quasi-factorization_2001}, yielding a multiplicative generalization of strong subadditivity for non-commuting conditional expectations. Several works consider quantum generalizations or related properties in a variety of settings \cite{capel_quantum_2018, bardet_modified_2021, bluhm_weak_2021}. This form of inequality is also known as approximate tensorization \cite{caputo_approximate_2015, bardet_approximate_2022} \footnote{We thank the authors of  \citen{bardet_approximate_2022} for access to an early draft that considered such an inequality in parallel with the writing of this manuscript. Their present version derives a comparable, multiplicative form in some cases.}, including a ``strong" form that is fully multiplicative and a ``weak" form that includes an additive correction term. In this paper, we refer primarily to the multiplicative form, which we generalize to any finite number of subalgebras:
\begin{definition}[Multiplicative Quasi-factorization]
Let $\{\E_j : j \in 1...J \in \NN\}$ be a set of conditional expectations and $\E$ the conditional expectation to their intersection algebra. We say that this set satisfies a strong quasi-factorization (SQF, or specifically $(\alpha_j)$-CSQF) if
\[ \sum_j \alpha_j D(\rho \| \E_j(\rho)) \geq D(\rho \| \E(\rho)) \pl. \]
for some $(\alpha_j > 0)_{j=1}^J$. We say that it satisfies complete, strong quasi-factorization (CSQF) if $\{\E_j \otimes \id\}$ has SQF for any finite-dimensional extension by an auxiliary system, where $\id$ acts as the identity on that auxiliary system. When $\alpha_j = \alpha_l$ for all $l,j \in 1...J$, we may write $\alpha$-(C)SQF.
\end{definition}
The original version of this paper, \citen[v1]{laracuente_quasi-factorization_2021}, showed such a bound for subalgebras with scalar intersection. In that version, it was left as a conjecture that the dimension could be replaced by a subalgebra index, the result generalized to arbitrary sets of subalgebras, and the inequality made tensor-stable.

Later, a new technique developed by Gao and Rouz\'e \cite{gao_complete_2021} showed general, tensor-stable, multiplicative quasi-factorization with constant determined by a subalgebra index $C$, rather than the system's dimension. Their result shows the existence of quasi-factorization for all finite-dimensional quantum systems. Incorporating one of the techniques of that result, we find a strengthened quasi-factorization that is asymptotically tight in the following sense: for a pair of conditional expectations $\E_1, \E_2$ with intersection conditional expectation $\E$ such that $\|\E_1 \E_2 - \E\|_{\Diamond} \rightarrow 0$, our bound approaches strong subadditivity. Our results have an asymptotic $\alpha \sim O(\ln C)$ dependence on the index for large $C$, improving on the asymptotic dependence from preceding versions of \citen{gao_complete_2021}. We find a strong, complete quasi-factorization like that in \citen{gao_complete_2021}, which also has asymptotic tightness like in \citen{bardet_approximate_2022} or \citen[v1]{laracuente_quasi-factorization_2021}, and logarithmic index scaling like that in \citen[v1]{laracuente_quasi-factorization_2021}. Furthermore, we explicitly show this for any finite number of conditional expectations.

Decay and decoherence are some of the most vexing challenges to quantum technology. We say that a semigroup $(\Phi^t)_{t=0}^\infty$ with fixed point conditional expectation $\Phi^\infty$ has MLSI with constant $\lambda$ ($\lambda$-MLSI) if for all $t \in \RR^+$,
\[ D(\Phi^t(\rho) \| \Phi^\infty(\rho)) \leq \exp(- \lambda t) D(\rho \| \Phi^\infty(\rho)) \pl. \]
MLSIs were introduced for classical systems in \citen{arnold_logarithmic_1998, bobkov_modified_2003} and for quantum systems in \citen{kastoryano_quantum_2013}, then recalled in \citen{bardet_estimating_2017}. MLSI was inspired by the earlier notion of the logarithmic Sobelev inequality \cite{gross_logarithmic_1975, gross_hypercontractivity_1975}, which does not hold as generally \cite{bardet_hypercontractivity_2018}. As defined in \citen{gao_fisher_2020}, a semigroup has $\lambda$-CMLSI if for all extensions by an auxiliary system $B$ and joint densities $\rho$ on the original system and $B$,
\[ D((\Phi^t \otimes \id^B)(\rho) \| (\Phi^\infty \otimes \id^B)(\rho))
	\leq \exp(- \lambda t) D(\rho \| (\Phi^\infty \otimes \id^B)(\rho)) \pl. \]
As a primary application, (C)SQF allows us to derive concrete (C)MLSI constants for quantum Markov semigroups. These do not follow from additive perturbations of strong subadditivity such as Equation \eqref{eq:addcorr}. It is shown in \citen{bardet_group_2021} that CMLSI upper bounds capacities of quantum channels, which are famously difficult to calculate due to superadditivity and hardness of numerics for high-dimensional quantum entropies. Furthermore, we note in that work that CMLSI implies tensor-stable decoherence time estimates, an important problem for quantum computing and memory.

The culminating result of this paper, Theorem \ref{graphthm}, derives CMLSI for semigroups described by finite, undirected graphs as represented on a basis in Hilbert space. This Theorem addresses an open problem, \citen[Remark 7.5]{gao_fisher_2020}. This example illustrates a broader principle known as transference, in which bounds on mixing rates of classical channels imply relative entropy decay rate bounds for quantum channels with related structure. Transference is used previously in \citen{gao_fisher_2020, bardet_group_2021}. This current work extends the idea to imply tensor-stable relative entropy inequalities based on order inequalities from classical vector spaces. The same principle applies to channels constructed from finite subgroups of the unitary group.

\subsection{Primary Contributions}
A quantum channel is a completely positive, trace-preserving map. By $\E$ we may denote a channel or a conditional expectation. By $\E_{\sigma}, \E_{\sigma *}$ we respectively denote a conditional expectation weighted by state $\sigma$ as in Section \ref{sec:asa} and its predual with respect to the trace. We write $\E_{\N, \sigma}$ and $\E_{\N, \sigma *}$ for a weighted conditional expectation to subalgebra $\N$ in order to explicitly emphasize the subalgebra. By $D(\cdot \| \cdot)$ we denote the relative entropy and by $H(\cdot)$ the von Neumann entropy. By $\id$ we denote the identity matrix. For systems $A,B,C,...$ or von Neumann algebras $\M,\N,...$ we denote by $|A|$ or $|\M|$ the dimension. The subsystem entropy is denoted $H(A)_\rho := H(\rho^A)$, where $\rho^A$ denotes the restriction to subsystem $A$ of a multipartite state $\rho$ on $A \otimes B \otimes C \otimes ...$. The state $\id/|A|$ or $\id/|\M|$ is the respective complete mixture on $A$ or $\M$. For a pair of densities $\rho,\sigma$, we use $\rho \geq \sigma$ (respective $\leq$, $>$, $<$) to denote the Loewner order. For a pair of channels $\Phi, \Psi$, we write $\Phi \geq \Psi$ if $\Phi(\rho) \geq \Psi(\rho)$ for all input densities $\rho$, and $\Phi \geq_{cp} \Psi$ if $(\Phi \otimes \id^B) \geq (\Psi \otimes \id^B)$ for all extensions via an auxiliary system $B$. All results of this paper assume finite-dimensional densities. Entropies should be read as using the natural logarithm, though when an inequality multiplicatively relates entropies to other entropies and logarithm-containing quantities, the inequality holds as long as the same base is taken for all logarithms.

As Proposition \ref{lem:worstsig}, we prove that for any densities $\rho, \sigma$ of the same dimension such that $\rho \succ \sigma$ ($\rho$ majorizes $\sigma$) and any $\zeta \in [0,1]$,
\[ D(\rho \| (1 - \zeta) \id/d + \zeta \sigma) \geq D(\rho \| (1 - \zeta) \id/d + \zeta \rho) \pl.\]
Combining Proposition \ref{lem:worstsig} with estimates of the derivatives of relative entropy with respect to complete mixture and the iteration technique of Section \ref{sec:asa}, we obtain a multiplicative bound on relative entropy to complete mixture as Theorem \ref{thm:relent}:
\begin{equation*}
D(\rho \| (1 - \zeta) \id/d + \zeta \sigma) \geq (1-a) D(\rho \| \id/d)
\end{equation*}
for any $a \in [0,1], b \in (0,1)$, and $\zeta \leq a \min \{(1-b)/(d + a(1-b) + 1), b/((1 - a b) d + a b + 1) \}$. This result as asymptotically tight in that we may take $a \rightarrow 1$ as $\zeta \rightarrow 0$.

In section \ref{sec:relent2}, we use the functional calculus as in \citen{gao_complete_2021} to generalize the relative entropy comparisons of Section \ref{sec:relent} from complete mixture to arbitrary conditional expectations. Rather than directly using the integral form of relative entropy for desired inequalities as in \citen{gao_complete_2021}, we use similar techniques to derive a perturbation result comparing relative entropy of related densities. This perturbative result, Theorem \ref{revconv}, is reminiscent of the triangle inequality for norms, allowing one to upper bound $D(\rho \| (1-\zeta) \sigma + \zeta \eta)$ in terms of $D(\rho \| \sigma)$ and $D(\rho \| \eta)$. Indeed, the primary idea of this proof is that by comparing the relative entropy to a weighted 2-norm, we may transfer the triangle inequality from the norm to entropy up to some constant factors.
\begin{theorem}[Triangle-like Relative Entropy Comparison]  \label{revconv}
Let $\rho, \sigma, \omega$ be densities such that $(1-\zeta) \sigma \leq \omega \leq (1 + \zeta(c-1)) \sigma$ for constants $\zeta \in (0,1)$ and $c \geq 1$. Let $\eta := (\omega - (1 - \zeta) \sigma)/\zeta$, so that $\omega = \zeta \eta + (1 - \zeta) \sigma$. Assume that $\rho \in \text{supp}(\sigma)$. Then $\eta$ is a density,
\begin{equation*}
\begin{split}
& (1 - \zeta)^2 D(\rho \| \sigma) \leq (1 + \zeta c + \zeta^2 (c-1) ) D(\rho \| \omega)
	+ \zeta (1 + \zeta) c D(\rho \| \eta) \text{, and}
\end{split}
\end{equation*}
\[ (1- 4 \zeta - \zeta^2) D(\rho \| \sigma) \leq (1 + \zeta c + \zeta^2 (c-1) ) D(\rho \| \omega)
	+ 2 \zeta(1 + \zeta) \frac{(c-1)^2}{c(\ln c - 1) + 1} D(\eta \| \sigma)  \pl. \]
\end{theorem}
Theorem \ref{revconv} is asymptotically tight in approaching the equality $D(\rho \| \sigma) = D(\rho \| \omega)$ as $\zeta \rightarrow 0$. When $D(\rho \| \sigma) \geq D(\eta \| \sigma)$, and the Theorem's conditions are satisfied,
\[ \Big (1 - 2 \zeta(1 + \zeta) \frac{(c-1)^2}{c(\ln c - 1) + 1} - 4 \zeta - \zeta^2 \Big ) D(\rho \| \sigma) 
	\leq (1 + \zeta c + \zeta^2 (c-1) ) D(\rho \| \omega) \pl. \] Theorem \ref{revconv}'s connection to quasi-factorization and subalgebra-relative entropy is apparent via Corollary \ref{cor:simpleG}: for any density $\rho$, quantum channels $\E,\Phi$ such that $\Phi \E = \E$, and constants $\zeta \in (0,1)$, $c \geq 1$ such that $(1-\zeta) \E(\rho) + \zeta \Phi(\rho) \leq (1 + \zeta(c-1)) \E(\rho)$, 
\begin{equation*}
D(\rho \| (1-\zeta) \E(\rho) + \zeta \Phi(\rho)) \geq \beta_{c, \zeta} D(\rho \| \E(\rho)) \text{ such that } \beta_{c,\zeta} = 1 - O(c \zeta) \pl.
\end{equation*}
As with the Theorem, Corollary \ref{cor:simpleG} is asymptotically tight in that $\beta_{c, \zeta} \rightarrow 1$ as $\zeta \rightarrow 0$. The more commonly stated criterion,
\begin{equation}
(1-\epsilon) \E \leq_{cp} \Phi \leq_{cp} (1+\epsilon) \E \pl,
\end{equation}
implies the conditions of Corollary \ref{cor:simpleG} with $\zeta = \epsilon, c = 2$. Conversely, we show as Proposition \ref{ephi} that with additional assumptions, $D(\rho \| \Phi(\rho))$ can be upper bounded in terms of $D(\rho \| \E(\rho))$.

Combining Corollary \ref{cor:simpleG} with the iterative technique of Section \ref{sec:asa} forms the base of a quasi-factorization result:
\begin{theorem}[Quasi-factorization] \footnote{After a version of Theorem \ref{asa2} appeared in v3 of this paper, \citen{gao_complete_2021} added comparable results (theorems 5.3 \& 5.4, corollary 5.5 in that paper). Nonetheless, the techniques of our Section \ref{sec:asa} originally appeared in v1 of this paper. It is these techniques that yield both the logarithmic dependence on $c$ described in Remark \ref{tightness} and the extension from two to many conditional expectations. Furthermore, Theorem \ref{asa2} has the advantage of approaching strong subadditivity in the appropriate commuting square limits while yielding a logarithmic (or no) index dependence otherwise.} \label{asa2}
Let $\{(\N_j, \E_{j *} ) : j = 1...J \in \NN, \N_j \subseteq \M\}$ be a set of $J \in \NN$ von Neumann algebras and associated (predual) conditional expectations within von Neumann algebra $\M$ and weighted respectively by densities $(\sigma_j)$. Let $\E$ be a channel such that $\E \E_{j *} = \E_{j *} \E = \E$ for each $\E_{j *}$.

Let $S = \cup_{m \in \NN} \{1...J\}^{\otimes m}$ be the set of finite sequences of indices. For any $s \in S$, let $\E^s$ denote the composition $\E_{j_1 *} ... \E_{j_m *}$ for $s = (j_1, ..., j_m)$. Let $\mu : S  \rightarrow [0,1]$ be a probability measure on $S$ and $k_{j,s}$ upper bound the number of times $\E_{j *}$ appears in each sequence $s$. If
\begin{equation} \label{eq:kzetacond}
 (1-\zeta) \E \leq_{cp} \sum_{s \in S} \mu(s) \E^s \leq_{cp} (1 + \zeta(c-1)) \E  \pl,
\end{equation}
then $\E$ is a projection, and for $\beta_{c,\zeta}$ given in Corollary \ref{cor:simpleG} and all input densities $\rho$ (including those with arbitrary extensions to auxiliary systems),
\[ \sum_{s \in S} \mu(s) \sum_j k_{s,j} D(\rho \| \E_{j *}(\rho)) \geq \beta_{c, \zeta} D(\rho \| \E(\rho)) \pl. \]
\end{theorem}
Recall the subalgebra indices
\begin{equation} \label{eq:indices}
\begin{split}
C(\M : \N) & = \inf \{c > 0 | \rho \leq c \E_\N(\rho) \forall \rho \in \M_* \} \\
C_{cb}(\M : \N) & = \sup_{n \in \NN} C(\M \otimes \MM_n : \N \otimes \MM_n)
\end{split}
\end{equation}
as considered in \citen{gao_relative_2020, gao_complete_2021} and originally by Pimsner and Popa \cite{pimsner_entropy_1986} as a finite-dimensional analog of the Jones index \cite{jones_index_1983}. When $\E$ is the (doubly stochastic) conditional expectation from $\M$ to $\N$, and $\Phi$ is a channel that leaves $\N$ invariant, $\rho, \Phi(\rho) \leq C(\M : \N) \E(\rho)$, and the bound holds up to arbitrary extensions with $C_{cb}(\M : \N)$ replacing $C(\M : \N)$. As Corollary \ref{index}, we show explicitly index-based bounds following Theorem \ref{asa2}. Furthermore, this Corollary shows how to obtain an $\alpha$-(C)SQF constant scaling logarithmically with the index.

Though $c$ can be upper bounded by the index as in Corollary \ref{index}, sometimes there is a better upper bound based on specific knowledge of the channels involved. As explained in Section \ref{sec:groups}, Theorems \ref{revconv} and \ref{asa2} are naturally strong in contexts reminiscent of transference, an idea present in \citen{gao_capacity_2018, gao_fisher_2020, li_graph_2020, bardet_group_2021}. Transference may compare quantum channels through analogous classical channels. When several quantum channels are weighted averages of the same unitary conjugations, we may often derive Loewner order inequalities by studying how operations and compositions affect the weights. These inequalities naturally map to the conditions of Theorems \ref{revconv} and \ref{asa2}.

Though it might not be obvious that all sets of subalgebras satisfy equation \eqref{eq:kzetacond} for some finite $k$ and $\zeta > 0$, Proposition \ref{near} shows that when the operator norm distance between a unital channel $\Phi$ and a conditional expectation $\E$ acting on half a Bell pair is sufficiently small, and $\E \Phi = \E$, $\Phi$ must be a convex combination of $\E$ with another channel $\Psi$ such that $\E \Psi = \E$. As shown in \citen{junge_noncommutative_2007}, it is always possible to find a convex combination of chains of conditional expectations from a finite set that approaches the conditional expectation to their intersection. Hence:
\begin{rem} \label{tightness}
For any set of finite-dimensional, doubly-stochastic conditional expectations $\{\E_j\}_{j=1}^J$, a quasi-factorization inequality in the form of Theorem \ref{asa2} holds.

Quasi-factorization is asymptotically tight. In particular, consider a set of continuously parameterized (predual) conditional expectations $\E_{1 *}^{(\theta)}, ..., \E_{J *}^{(\theta)}$ with intersection conditional expectation $\E$. If $\E_{1 *}^{(\theta)} ... \E_{J *}^{(\theta)}(\rho) \rightarrow \E(\rho)$ in diamond norm for all input densities $\rho$ as $\theta \rightarrow 0$, then $\{\E_j\}$ has $\alpha_\theta$-(C)SQF with $\alpha_\theta \rightarrow 1$. When $J=2$, such an arrangement approaches strong subadditivity as the conditional expectations approach a commuting square.
Proposition \ref{near} and Theorem \ref{asa2} yield a concrete continuity bound on the convergence rate.
\end{rem}

As a primary application, quasi-factorization allows us to combine MLSI estimates. Let $\Phi^t$ be a family of quantum channels in dimension $d$ parameterized by $t \in \RR^+$, such that $\Phi^s \circ \Phi^t = \Phi^{s+t}$ for all $s,t \in \RR^+$. This family of channels is thereby a semigroup under composition, and there exists a Lindbladian generator given by $\L = \lim_{t \rightarrow 0} (\id - \Phi^t)$ such that $\Phi^t = e^{- t \L}$. As introduced , we say that $\L$ has $\lambda$-MLSI if for any input density $\rho$,
\begin{equation}
D(\Phi^t(\rho) \| \Phi^\infty(\rho)) \leq e^{-\lambda t} D(\rho \| \Phi^\infty(\rho)) \pl.
\end{equation}
A Lindbladian $\L$ has $\lambda$-CMLSI (complete, modified logarithmic Sobolev inequality) if for any bipartite density $\rho^{AB}$,
\begin{equation}
D((\Phi^t \otimes \id^B)(\rho) \| (\Phi^\infty \otimes \id^B)(\rho)) \leq e^{-\lambda t} D(\rho \| (\Phi^\infty \otimes \id^B)(\rho)) \pl.
\end{equation}
We (re)prove:
\begin{prop} \label{prop:mlsimerge}
Let $\{\Phi^t_j : j \in 1...J \in \NN\}$ be self-adjoint quantum Markov semigroups such that $\Phi_j^t = \exp(- \L_j t)$ with fixed point conditional expectation $\E_{j *} = \lim_{t \rightarrow \infty} \Phi_j^t$ for each $j$ weighted respectively by $(\sigma_j)$. Let $\E_{\sigma *}$ be the weighted intersection fixed point conditional expectation, assuming $\E_{j *}$ are compatibly weighted so that it exists. Let $\Phi^t$ be the semigroup generated by $\L = \sum_j \alpha_j \L_j + \L_0$, where $\L_0$ generates $\Phi_0^t$ such that $\Phi_0^t \E_{\sigma *} = \E_{\sigma *} \Phi_0^t = \E_{\sigma *}$. If $\{\Phi_j^t\}$ has $\{ \alpha_j \}$-(C)SQF, and $\Phi^t_j$ has $\lambda$-(C)MLSI for each $j$, then $\Phi^t$ has $\lambda$-(C)MLSI.
\end{prop}
Proposition \ref{prop:mlsimerge} is not surprising, and the historical use of quasi-factorization in proving modified log Sobolev inequalities relies on essentially equivalent results. A simple proof of results like Proposition \ref{prop:mlsimerge} with $\Phi_0 = \id$ emerges from the Fisher information formulation of MLSI detailed in \citen{gao_fisher_2020}, as the Fisher information of a sum of Lindbladians is equal to the sum of their respective Fisher informations. An alternate form of proof appears in Appendix \ref{sec:approxrelent}. Proposition \ref{prop:mlsimerge} can also be useful when conditional expectations do commute, in which case quasi-factorization reduces to SSA. When multiple subsets of constituent conditional expectations lead to the same intersection algebra, $\alpha < 1$ is possible. The use of bipartite quasi-factorization to prove MLSI appears in \citen{bardet_approximate_2022} and \citen{capel_modified_2020}, so similar methods exist in the literature.
\begin{rem}
As shown in \citen[Section 3.2]{bardet_approximate_2022}, one can convert (C)SQF inequalities from the doubly stochastic setting to non-trivially weighted conditional expectations. The results of \citen{junge_stability_2019} compare (C)MLSI constants of Lindbladians with non-tracial invariant states to those with tracial invariant states, and the methods therein underlie the comparison in \citen[Section 3.2]{bardet_approximate_2022}. One may also use Lemma \ref{ephi} together with Theorem \ref{thm:relent} or \ref{revconv} similarly.
\end{rem}

We demonstrate uses of quasi-factorization in two primary examples. First, we show asymptotically tight, uncertainty-like relations for pairs of measurement bases. In particular, when $A$ is a $d$-dimensional subsystem of bipartite system $A B$ with respective matrix algebras $\A$ and $\B$, and $\S, \T \subseteq \A$ correspond respectively to measurement bases $\{\ket{i_S} : i = 1...d \}$ and $\{\ket{i_T} : i = 1...d\}$ such that $\xi = \min_{i,j} |\braket{i_S | j_T}|^2 > 0$, quasi-factorization implies that
\begin{equation*}
H(\E_{\S \otimes \B}(\rho)) + H(\E_{\T \otimes \B}(\rho)) \geq 2 H(\rho) + \frac{\beta_{d, \epsilon}}{\lceil \ln_{1 - d \xi} \epsilon \rceil} (\ln d + H(B)_\rho - H(\rho))
\end{equation*}
for any $\epsilon \leq 1 - d \xi$, where $\beta_{d, \epsilon}$ is as in Corollary \ref{cor:simpleG}. This form of inequality, detailed in Subsection \ref{sec:ucr}, strengthens the usual, entropic uncertainty principle for highly mixed states. More broadly, quasi-factorization yields uncertainty-like inequalities between relative entropies to the invariant subalgebras of finite groups.

Second we use quasi-factorization to show new entropy inequalities and decay estimates for mixing channels described by finite groups and graphs. As detailed in Subsection \ref{sec:graphs}, a finite, undirected graph $G$ with $n$ vertices can be represented on $n$-dimensional densities by conditional expectations given by
\begin{equation} \label{eq:graphcondexpintro}
\begin{split}
\E_{i,j}(\rho) = \frac{1}{2} (& \ket{i}\bra{i} \rho \ket{i}\bra{i} + \ket{j}\bra{j} \rho \ket{j}\bra{j} + \ket{i}\bra{j} \rho \ket{j}\bra{i} + \ket{j}\bra{i} \rho \ket{i} \bra{j}) \\ + & \Big ( \sum_l \ketbra{l} \Big ) \rho \Big ( \sum_r \ketbra{r} \Big )
\end{split}
\end{equation}
for each pair $(i,j)$ in $G$'s edges. Using big-$\Omega$ notation to denote asymptotic order:
\begin{theorem} \label{graphthm}
Let an $m$-regular, connected graph with $n$ vertices $G$ have subleading normalized adjacency matrix eigenvalue (also known as spectral gap) $\gamma$ as defined in Theorem \ref{expanderthm}. Consider the conditional expectations of Equation \eqref{eq:graphcondexpintro} possibly in tensor product with an arbitrary, finite-dimensional auxiliary system. Let $\E_G$ denote the conditional expectation to the invariant subspace of these for all $i,j$. Then
\begin{equation*}
\sum_{(i,j) \in E} D(\rho \| \E_{i,j}(\rho)) \geq \Omega \Big ( \frac{\ln (1/\gamma)}{\ln n} \Big ) D(\rho \| \E_G(\rho)) \pl.
\end{equation*}
This inequality is stable under tensor extensions by auxiliary systems. The Lindbladian
\[ \L_G(\rho) := \sum_{(i,j) \in E} (\rho - \E_{i,j}(\rho)) \]
has CMLSI with the same constant.
\end{theorem}
The technical version of Theorem \ref{graphthm} appears as Theorem \ref{graphthmtech}.
\begin{rem} \label{graphclassical}
Let $L_G$ be the degree-normalized Laplacian matrix corresponding to a finite, connected, undirected graph $G$ with $n$ vertices, and let $\L_G$ be the Lindbladian as constructed in Theorem \ref{graphthm}. $L_G$ generates a semigroup on $l_1^n$. For any probability vector $\vec{p} \in l_1^n$, let density $\rho \in S_1^n$ be such that $\vec{p} = \text{diag}(\rho)$, $L(\vec{p}) = \text{diag}(\L_G(\rho))$, and $\text{diag} : S_1^n \rightarrow l_1^n$ denotes restriction to the diagonal. Furthermore, if a density $\tilde{\rho} \in S_1^n \otimes B$ has the partially diagonal form
\[ \tilde{\rho} = \sum_{x \in 1...n} p_x \ketbra{x} \otimes \rho_x^B \pl, \]
so will $(\L \otimes \id^B)(\rho)$ for any finite-dimensional extension $B$, so restriction to the diagonal remains bijective. In this way, Theorem \ref{graphthm} bounds CMLSI constants for finite graphs in a sense compatible with that of \citen{diaconis_logarithmic_1996, gao_fisher_2020, li_graph_2020}. A similar argument holds for finite groups.
\end{rem}
Theorem \ref{graphthm} answers a problem left open by \citen[Remark 7.5]{gao_fisher_2020}, showing that the fastest, regular expander graphs have CMLSI with constant no worse than one over logarithmic in $n$, in line with expectations based on classical mixing times of expanders and classical MLSI \cite{bobkov_modified_2006}. It also yields expected decay times for cyclic graphs. For any graph with $\gamma(n)$ constant in $n$, Theorem \ref{graphthm} shows convergence in $O(\ln_{\gamma(n)}(1/n))$ time. This convergence time is believed to be of optimal asymptotic order in $n$, matching the best classical bounds in known cases.

Subsection \ref{sec:relent} proves Theorem \ref{thm:relent}, a special case of the more general Theorem \ref{revconv} proven in Subsection \ref{sec:relent2}. Section \ref{sec:asa} proves Theorem \ref{asa2}. Section \ref{sec:applications} describes applications that use quasi-factorization to tighten entropic uncertainty relations and to derive new inequalities on graphs and groups. Section \ref{sec:conclusions} concludes with some open problems.

%% file: relative_ent.tex
\section{Multiplicative Perturbations of Relative Entropy}
Since relative entropy is biconvex,
\[ D(\rho \| (1-\zeta) \omega + \zeta \sigma) \leq (1-\zeta) D(\rho \| \omega) + \zeta D(\rho \| \sigma) \]
for any densities $\rho, \omega, \sigma$ and $\zeta \in [0,1]$. In the other direction, there are examples of densities $\rho, \omega, \sigma$ for which $D(\rho \| (1-\zeta) \omega + \zeta \sigma) = 0$, but
\[ D(\rho \|  \omega), D(\rho \| \sigma) > 0 \pl. \]
Similarly, it is possible that $D(\rho \| (1-\zeta) \omega + \zeta \sigma)$ is finite when $D(\rho \|  \omega)$ or $D(\rho \| \sigma)$ is infinite. Hence in full generality, there is no way to multiplicatively restrict the extent of non-concavity of relative entropy.

In Subsection \ref{sec:relent}, we show that when $\omega = \id/d$, and $\rho$ majorizes $\sigma$, there is a multiplicatively adjusted form of concavity-like relation. Recall that $\rho$ majorizes $\sigma$ if
\[ \sum_{i=1}^n \rho_i \geq \sum_{i=1}^n \sigma_i \text{ for any } n \in 1...d \pl, \]
where $(\rho_i)_{i=1}^d$ and $(\sigma_i)_{i=1}^d$ are the respective eigenvalues of $\rho$ and $\sigma$ in non-increasing order. In Subsection \ref{sec:relent2}, we show an analogous bound when $\omega = \E(\rho)$ and $\sigma = \Phi(\rho)$ for channels $\Phi$ and $\E$ under certain conditions. These conditions are satisfied when $\E$ is a conditional expectation to an invariant subspace of $\Phi$.

\subsection{Perturbation of Relative Entropy to Complete Mixture} \label{sec:relent}
If we restrict to $\omega = \id/d$, then $D(\rho \| (1-\zeta) \id/d + \zeta \sigma)$ and $D(\rho \| \id/d)$ are both finite for any $\zeta \in [0,1)$. These conditions are still insufficient to lower bound $D(\rho \| (1-\zeta) \id/d + \zeta \sigma)$ in terms of $D(\rho \| \id/d)$. We may for instance take $\sigma$ pure, and $\rho = (1-\zeta) \id/d + \zeta \sigma$. The final condition we add is that $\rho$ majorizes $\sigma$ ($\rho \succ \sigma$), in which case we obtain Theorem \ref{thm:relent}.

The results of this subsection precede \citen{gao_complete_2021} and use a different approach from that of section \ref{sec:relent2}. Many of the results herein are subsumed by results of that section at least up to constants. Nonetheless, we include this subsection as illustrative of a more computationally inspired line of proof. Furthermore, this method yields intermediate results of potentially independent interest and gives more intuition for the subsequent generalization.
\begin{lemma}[Flattening] \label{lem:adjust2}
Let $\rho, \omega$ be simultaneously diagonalizable densities of dimension $d$. Let $i \neq j \in 1...d$ such that $\rho_i \geq \rho_j$, and $\rho_i \omega_j \geq \rho_j \omega_i$. Let $\omega_i \geq \delta > 0$. Let $\omega \rightarrow \tilde{\omega}$ under the replacement $\omega_i \rightarrow \tilde{\omega_i} = \omega_i - \delta, \omega_j \rightarrow \tilde{\omega}_j = \omega_j + \delta$. Then $D(\rho \| \omega) \leq D(\rho \| \tilde{\omega})$.
\end{lemma}
\begin{proof}
For any $a > b \in \RR^+$, $(a - b)/a \leq \ln a - \ln b \leq (a-b)/b$, as one can verify from $(d/dx)(\ln x) = 1/x$. Hence
\begin{equation*}
\begin{split}
D(\rho \| \tilde{\omega}) - D(\rho \| \omega) & = \rho_i(\ln \omega_i - \ln (\omega_i - \delta)) + \rho_j (\ln \omega_j - \ln (\omega_j + \delta)) \\
& \geq (\rho_i / \omega_i  - \rho_j / \omega_j) \delta \geq 0 \pl.
\end{split}
\end{equation*}
\end{proof}
\begin{figure}[h!] \scriptsize \centering
	\begin{subfigure}[b]{0.22\textwidth}
		\includegraphics[width=0.95\textwidth]{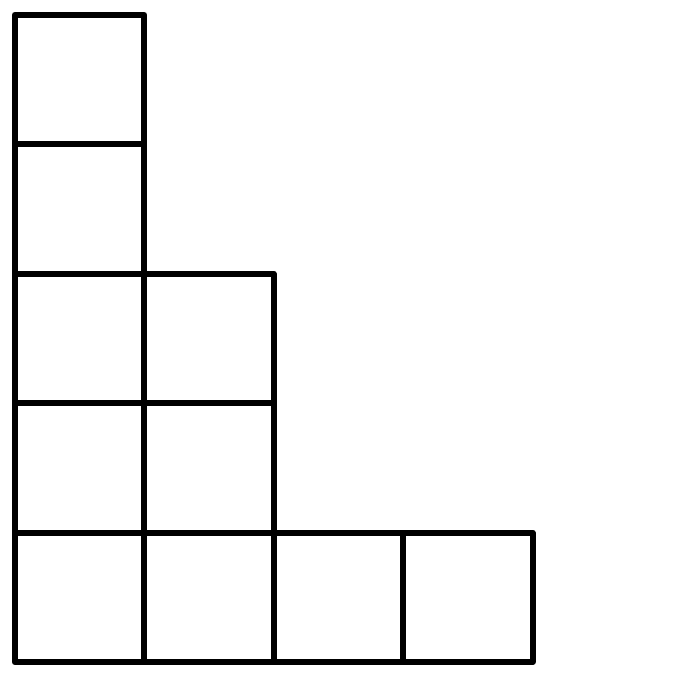}
	\end{subfigure}
	\begin{subfigure}[b]{0.22\textwidth}
		\includegraphics[width=0.95\textwidth]{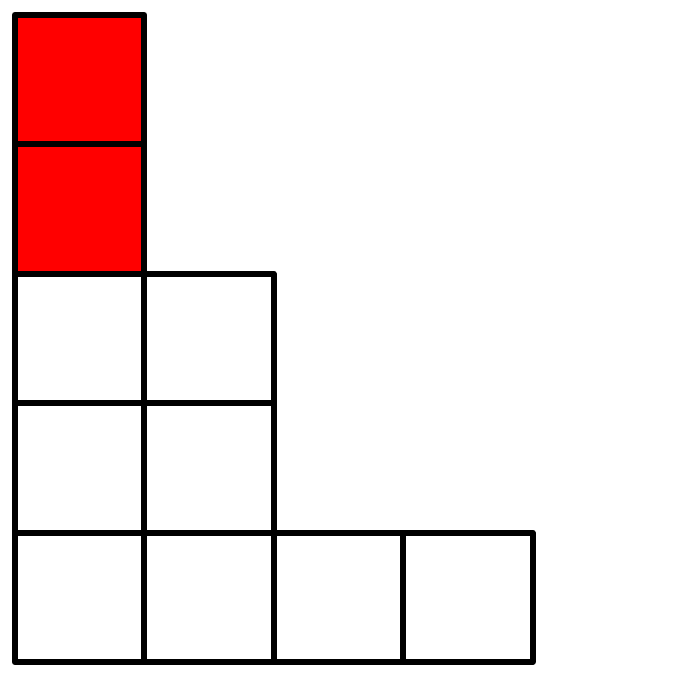}
	\end{subfigure}
	\begin{subfigure}[b]{0.22\textwidth}
		\includegraphics[width=0.95\textwidth]{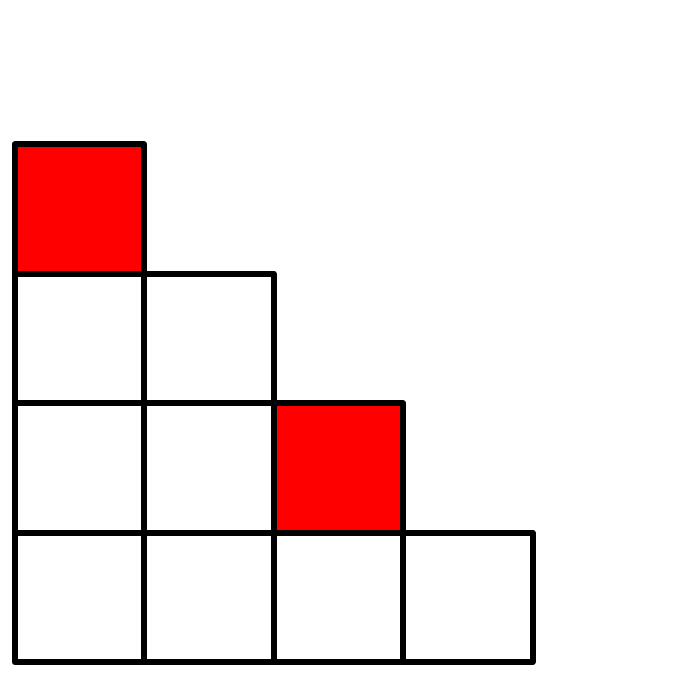}
	\end{subfigure}
	\begin{subfigure}[b]{0.22\textwidth}
		\includegraphics[width=0.95\textwidth]{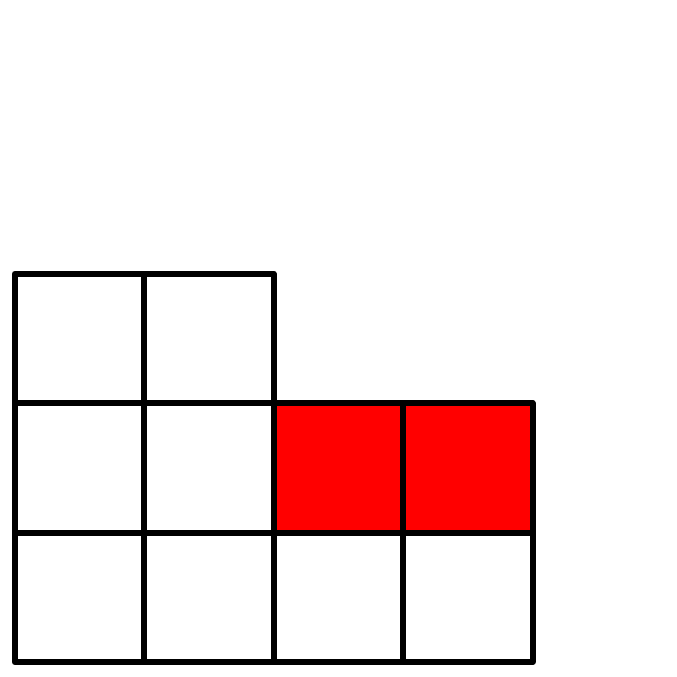}
	\end{subfigure}
	\begin{subfigure}[b]{0.22\textwidth}
		\includegraphics[width=0.95\textwidth]{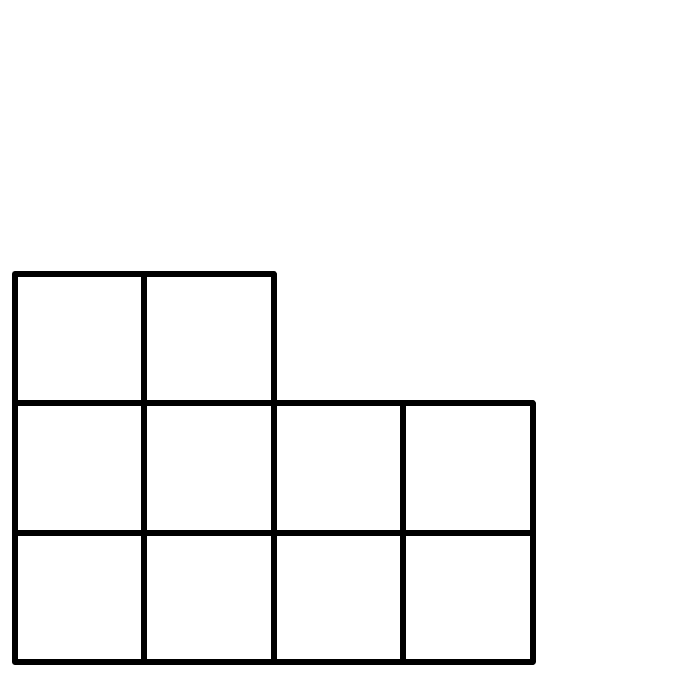}
	\end{subfigure}
	\begin{subfigure}[b]{0.22\textwidth}
		\includegraphics[width=0.95\textwidth]{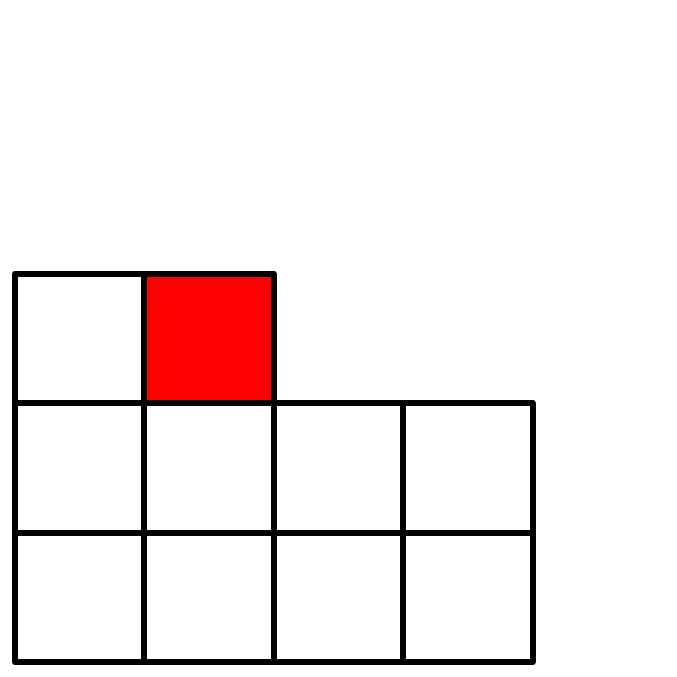}
	\end{subfigure}
	\begin{subfigure}[b]{0.22\textwidth}
		\includegraphics[width=0.95\textwidth]{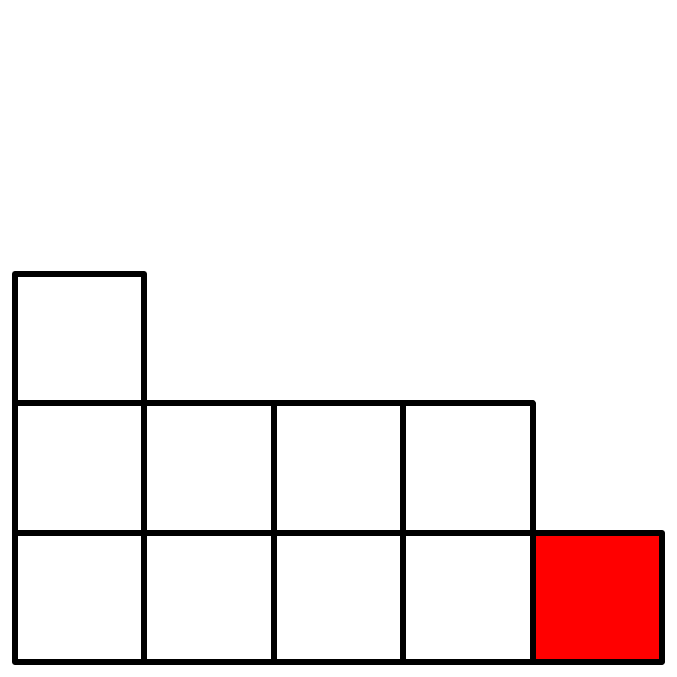}
	\end{subfigure}
	\begin{subfigure}[b]{0.22\textwidth}
		\includegraphics[width=0.95\textwidth]{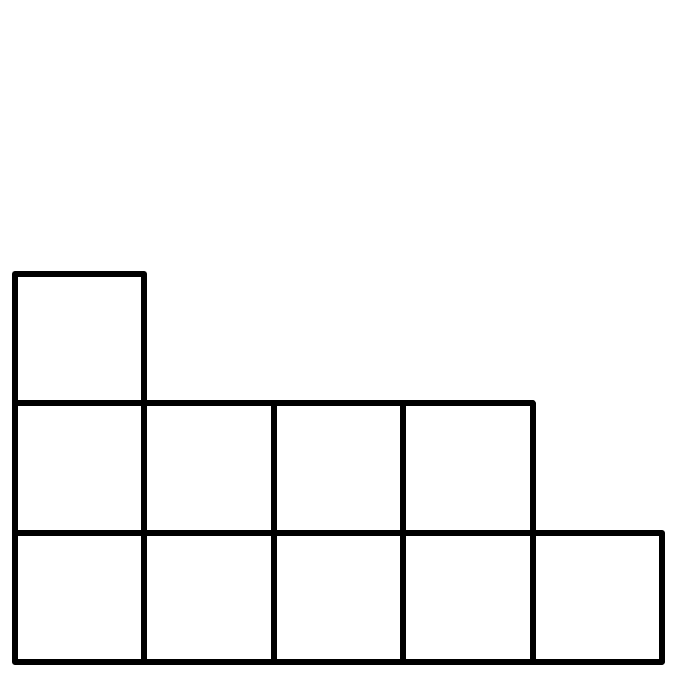}
	\end{subfigure}
	\caption{Visualization of a cascading redistribution, which converts the distribution on the top-left to that on the bottom-right. Here we visualize a density on Hilbert space dimension $d=5$. Each subfigure represents a density, where each block corresponds to a unit of probability equal to 0.1, and densities should be read from left to right starting with $\ketbra{0}$. The starting, top-left configuration corresponds to the density $\rho = 0.5 \ketbra{0} + 0.3 \ketbra{1} + 0.1 \ketbra{2} + 0.1 \ketbra{3}$. In the top row, probability is redistributed from the largest component to smaller components as in the algorithm from Proposition \ref{lem:worstsig}, yielding the density $0.3 \ketbra{0} + 0.3 \ketbra{1} + 0.2 \ketbra{2} + 0.2 \ketbra{3}$. The red blocks highlight the redistributed units of probability. In the bottom row, the algorithm continues, now subtracting from the $\ketbra{1}$ component and adding to $\ketbra{4}$. The ending configuration corresponds to the density $\sigma = 0.3 \ketbra{0} + 0.2 \ketbra{1} + 0.2 \ketbra{2} + 0.2 \ketbra{3} + 0.1 \ketbra{4}$. \label{fig:cascade}}
\end{figure}

\begin{prop} \label{lem:worstsig}
Let $\rho$ and $\sigma$ be two densities such that $\rho \succ \sigma$ ($\rho$ majorizes $\sigma$), and $\zeta \in [0,1]$. Then
\[ D(\rho \| (1 - \zeta) \id/d + \zeta \sigma) \geq D(\rho \| (1 - \zeta) \id/d + \zeta \rho) \pl.\]
\end{prop}
\begin{proof}
The main idea of this proof is that if $\rho \succ \sigma$, then flattening $\rho$ until it becomes $\sigma$ only increases the value of $D(\rho \| (1 - \zeta) \id/d + \zeta \rho)$. First,
\begin{equation*}
D(\rho \| (1 - \zeta) \id/d + \zeta \sigma) \geq D(\rho \| (1 - \zeta) \id/d + \zeta \E_{\braket{\rho}}(\sigma)) \pl,
\end{equation*}
by data processing under $\E_{\braket{\rho}}$, the conditional expectation onto the subalgebra generated by $\rho$. It is obvious that $\E_{\braket{\rho}}(\rho) = \rho$ and that $\E_{\braket{\rho}}(\sigma)$ commutes with $\rho$. We hence assume for the rest of the proof that that $[\rho, \sigma] = 0$. Let $\rho^\zeta = (1 - \zeta) \id/d + \zeta \rho$, and define $\sigma^\zeta$ analogously. Let $\vec{\rho}^\zeta$ and $\vec{\sigma}^\zeta$ be $d$-dimensional vectors of the eigenvalues of $\rho^\zeta$ and $\sigma^\zeta$ respectively, each in non-increasing order.

Let $\vec{\omega} = \vec{\rho}^\zeta$. We alter $\vec{\omega}$ via a cascading probability redistribution procedure consisting of the following steps, which transform it into a copy of $\vec{\sigma}^\zeta$:
\begin{enumerate}
	\item Start with the index $i$ set to 1.
	\item Let $\Delta = \vec{\omega}_i - \vec{\sigma}_i^\zeta$. If $\Delta \leq 0$, then continue to step (3). If $\Delta > 0$, then
	\begin{enumerate}
		\item Subtract $\Delta$ from $\vec{\omega}_i$. Let $j = i+1$. 
		\item If $\vec{\omega}_j < \vec{\sigma}_j^\zeta$, then let $\delta = \min\{\Delta, \vec{\sigma}^\zeta_j - \vec{\omega}_j\}$. Add $\delta$ to $\vec{\omega}_j$ and subtract it from $\Delta$.
		\item If $\Delta = 0$ (which must happen at or before $j=d$ for normalized densities), go to step (3). Otherwise, increment $j \rightarrow j+1$, and return to substep (2b).
	\end{enumerate}
	\item If $i < d$, increment $i \rightarrow i+1$ and return to step (2). Otherwise, the procedure is done.
\end{enumerate}
See Figure \ref{fig:cascade}. Since this procedure only subtracts from larger eigenvalues and adds to smaller ones, we apply Lemma \ref{lem:adjust2} at each step that transfers probability mass from one index to another. If $\vec{\rho}_i \geq \vec{\rho}_j$, then $\vec{\rho}_i / \vec{\rho}^\zeta_i \geq \vec{\rho}_j / \vec{\rho}^\zeta_j$. Furthermore, if $\vec{\rho}_i \geq \vec{\rho}_j$ for any $i$ and $j$, then it is always the case that $\vec{\omega}_i \leq \vec{\rho}_i^\zeta$ even as $\vec{\omega}$ changes throughout the algorithm, since we move probability mass out of $\vec{\omega}_i$ and into $\vec{\omega}_j \geq \vec{\rho}^\zeta_j$. Hence $\vec{\rho}_i / \vec{\rho}_j \geq \vec{\omega}_i / \vec{\omega}_j$ for all $i$ and $j$ such that $\vec{\rho}_i \geq \vec{\rho}_j$. Since each step of the flattening algorithm can only increase the relative entropy via Lemma \ref{lem:adjust2}, $D(\rho \| \rho^\zeta) = D(\vec{\rho} \| \vec{\rho}^\zeta) \leq D(\vec{\rho} \| \vec{\sigma}^\zeta)$. Finally, using the simultaneous diagonalizability of $\rho$ and $\sigma$, it is easy to see that $D(\vec{\rho} \| \vec{\sigma}^\zeta) \leq D(\rho \| \sigma^\zeta)$.
\end{proof}
\begin{lemma} \label{lem:logcomp1}
Let $b \geq 0$, and $a \in [0,1]$. Then for $\zeta \leq a/(1+b)$,
\begin{equation*}
a \ln (1 + b) \geq \ln (1 + \zeta b) \pl.
\end{equation*}
\end{lemma}
\begin{proof}
We exponentiate both sides and solve
\begin{equation*}
\begin{split}
& (1+b)^a \geq 1 + \zeta b \pl,
\end{split}
\end{equation*}
yielding
\begin{equation*}
\zeta \leq \frac{(1+b)^a - 1}{b} \pl.
\end{equation*}
We then estimate
\[ \frac{(1+b)^a - 1}{b} = \frac{(1+b)^{1+a} - 1 - b}{b(1+b)} \geq \frac{1 + (1+a) b - 1 - b}{b(1+b)} = \frac{a}{1+b} \]
by Bernoulli's inequality.
\end{proof}

\begin{lemma} \label{lem:towardident}
Let $\rho$ be given in a diagonal basis by $(\rho_i)_{i=1}^n$, where $n$ is the dimension of the system. Let $a, \beta \in (0,1)$, $i,j \in 1...n$ such that $\rho_i \geq 1/n \geq \rho_j$, and let $\zeta \in \RR^+$ such that
\[ 0 < \zeta \leq a \min \Big \{\frac{1-\beta}{n + a(1-\beta) + 1}, \frac{\beta}{(1 - a \beta) n + a \beta + 1} \Big \} \pl. \]
If $\rho_j > 0$, then
\[ \Big ( \frac{\partial}{\partial \rho_i} - \frac{\partial}{\partial \rho_j} \Big ) \tr(\rho (a \ln (n \rho) - \ln ((1 - \zeta) \id + \zeta n \rho))) \geq 0 \pl. \]
If $\rho_j = 0$, then letting $\tilde{\rho} = \rho - \epsilon \hat{i} + \epsilon \hat{j}$,
\[ \tr(\rho (a \ln (n \tilde{\rho}) - \ln ((1 - \zeta) \id + \zeta n \tilde{\rho}))) > 
	\tr(\rho (a \ln (n \rho) - \ln ((1 - \zeta) \id + \zeta n \rho))) \]
for sufficiently small $\epsilon$, where $\hat{i}$ and $\hat{j}$ denote the rank 1 unit densities according to the respective $i$th and $j$th basis vectors in the chosen diagonal basis of $\rho$.
\end{lemma}
Since the proof of Lemma \ref{lem:towardident} is technical and the Lemma a more specific alternative to methods in section \ref{sec:relent2}, we defer proof to Appendix \ref{sec:approxrelent}.

\begin{theorem} \label{thm:relent}
Given two densities $\rho,\sigma$ in dimension $d$ such that $\rho \succ \sigma$ ($\rho$ majorizes $\sigma$), 
\begin{equation*}
D(\rho \| (1 - \zeta) \id/d + \zeta \sigma) \geq (1-a) D(\rho \| \id/d)
\end{equation*}
for any $a \in [0,1], b \in (0,1)$ and 
\[\zeta \leq a \min \Big \{\frac{1-b}{d + a(1-b) + 1}, \frac{b}{(1 - a b) d + a b + 1} \Big \} \pl.\]
\end{theorem}
\begin{proof}
Let $n$ be the dimension of $\rho$, since $d$ may be confused with a derivative. The goal is to show that given some $a \in [0,1]$ and densities $\rho,\sigma$ such that $\rho \succ \sigma$,
\[ D(\rho \| (1 - \zeta) \id/n + \zeta \sigma) - (1-a) D(\rho \| \id/n) \geq 0 \]
for an appropriate value of $\zeta \in [0,1]$. We apply Lemma \ref{lem:worstsig} to replace $\sigma$ by $\rho$. For any states $\rho$ and $\omega$,
\[ D(n \rho \| n \omega) = \tr((n \rho) (\ln \rho + \ln n - \ln \omega - \ln n)) = n D(\rho \| \omega) \pl. \]
Hence it is sufficient to prove that
\begin{equation} \label{eq:ddiffderiv}
D(n \rho \| (1 - \zeta) \id + \zeta n \rho) - (1-a) D(n \rho \| \id) \geq 0 \pl,
\end{equation}
which expands as
\begin{equation*}
 ... = n \tr(\rho (a \ln (n \rho) - \ln((1 - \zeta) \id + \zeta n \rho))) \geq 0 \pl.
\end{equation*}
The main insight behind this proof is Lemma \ref{lem:towardident}. If $\rho = \id/n$, then both terms are 0, and the proof is trivially complete. If $\rho \neq \id/n$, then the total probability mass above $\id/n$ must equal that below $\id/n$ to maintain normalization. Hence we apply Lemma \ref{lem:towardident} to successive pairs of $i,j$ such that $\rho_i \geq 1/n \geq \rho_j$, flattening $\rho$ until we transform the second argument of the relative entropy in Equation \eqref{eq:ddiffderiv} from  $(1 - \zeta) \id + \zeta n \rho$ to $\id$ without increasing the relative entropy.
\end{proof}
We may optimize $a$ and $b$ in Theorem \ref{thm:relent} for given values of $d$ and $\zeta$. If we wish to avoid optimization, $b = 1/2$ is a reasonable value, and one may use the calculated bound of Corollary \ref{cor:approx}. We further see from this Corollary that \ref{thm:relent} is asymptotically tight: as $\zeta \rightarrow 0$ for fixed $d$,we may choose $a$ and $b$ such that $(1-a) \rightarrow 1$. Theorem \ref{thm:relent} relies on a comparison using telescopic relative entropy as introduced in \citen{audenaert_telescopic_2014}.
\begin{cor} \label{cor:approx}
Given $a \in [0,1]$ and two densities $\rho,\sigma$ in dimension $d$ such that $\rho \succ \sigma$ ($\rho$ majorizes $\sigma$), 
\begin{equation*}
D(\rho \| (1 - \zeta) \id/d + \zeta \sigma) - \Big (1-\frac{32(d+1) \zeta}{15 + (7 d - 24) \zeta} \Big ) D(\rho \| \id/d) \geq 0
\end{equation*}
is achieved whenever
\begin{equation*}
\zeta < \frac{15}{25 d + 56} \pl.
\end{equation*}
With $\zeta \rightarrow 0$ as $d$ is held fixed, we see that this expression is asymptotically tight. We may choose $a = 1 - 1/d$ and
\begin{equation*}
\zeta \leq \frac{15 d - 1}{32d^2 - 7 d^2 + 7 d + 32 d + 24 d - 24} = \frac{15 d - 1}{25 d^2 + 63 d - 24}
\end{equation*}
or $a = 1/2$ and
\begin{equation*}
\zeta \leq \frac{15}{57d + 88} \pl.
\end{equation*}
\end{cor}
The proof of this Corollary is contained in Appendix \ref{sec:approxrelent}. The proof is essentially a basic calculation with linear approximations.

\begin{rem}
For a bipartite classical-quantum state $\rho^{XB}$ with classical system $X$, $\rho = \sum_x p_x \vec{e}_x \otimes \rho_x^B$, where $\vec{e}_x$ is a classical basis vector. One can thereby expand $D(\rho \| \E(\rho)) = \sum_{x \in X} p_x D(\rho_x \| \E(\rho_x))$. Hence Theorem \ref{thm:relent} may include classical auxiliary systems.
\end{rem}

%% file: log_constant.tex
\subsection{Perturbation of Relative Entropy to a Subalgebra} \label{sec:relent2}
The primary result of this section is the proof of Theorem \ref{revconv}, generalizing and strengthening Theorem \ref{thm:relent} using the methods of \citen{gao_complete_2021}. For this proof, we recall 4 useful results of \citen{gao_complete_2021} and preceding works. First, as noted in \citen{gao_fisher_2020} or inferred from the form of weighted inner product constructed in \citen{temme_chi2_2010, carlen_gradient_2017}:
\begin{lemma}[Inverse-Weighted Norm] \label{wnorm}
The function
\begin{equation*}
\|X\|_{\rho^{-1}} := \sqrt{\int_0^\infty tr(X^\dagger (\rho + r)^{-1} X (\rho + r)^{-1} ) dr}
\end{equation*}
is a norm for strictly positive $\rho$ on spaces with finite trace and Hilbert-Schmidt norm.
\end{lemma}
Though knowing how $\|X\|_{\rho^{-1}}$ is induced by an inner product is in principle sufficient to deduce geometrically that it must be a norm, we here give an elementary proof:
\begin{proof}[Proof]
Let
\[ \Gamma^{-\alpha}_{\rho, r}(X) = (\rho + r)^{-\alpha} X (\rho + r)^{-\alpha} \]
as a more parameterized version of $\Gamma^{-1}$ as in \citen{gao_complete_2021}. Then as in \citen{gao_complete_2021},
\begin{equation*}
\|X\|_{\rho^{-1}}^2 = \int_0^\infty \braket{X, \Gamma^{-1}_{r, \rho} (X)} dr
	= \int \tr(X^\dagger (\rho + r)^{-1} X (\rho + r)^{-1} ) dr \pl.
\end{equation*}
Via the cyclic property of the trace,
\[ \|X\|_{\rho^{-1}}^2
	= \int \braket{\Gamma^{-1/2}_{r, \rho}(X), \Gamma^{-1/2}_{r, \rho} (X)} dr
	= \int \| \Gamma^{-1/2}_{r, \rho} (X) \|_2^2 dr \pl, \]
where $\|\cdot\|_2$ is the usual Schatten or Hilbert-Schmidt 2-norm. This is already enough to show positivity. By inspecting the form of $\Gamma^{-1/2}_{r,\rho}(X)$, we can also see that $\|a X\|_{\rho^{-1}} = |a| \|X\|_{\rho^{-1}}$ for all $a \in \CC$, and that $\|X\|_{\rho^{-1}}^2 = 0 \iff X = \hat{0}$, the zero matrix. Expanding and using the triangle inequality for the Schatten 2-norm,
\[ \|X + Y\|_{\rho^{-1}}^2 = \|X\|_{\rho^{-1}}^2 + \|Y\|_{\rho^{-1}}^2 + 
	2 \int  \| \Gamma^{-1/2}_{r, \rho} (X) \|_2 \| \Gamma^{-1/2}_{r, \rho} (Y) \|_2 dr \pl. \]
Proving the triangle inequality for the weighted norm then reduces to showing that
\[ \int  \| \Gamma^{-1/2}_{r, \rho} (X) \|_2 \| \Gamma^{-1/2}_{r, \rho} (Y) \|_2 dr
	\leq \sqrt{\Big (\int \| \Gamma^{-1/2}_{s, \rho} (X) \|_2^2 ds \Big ) \Big ( \int \| \Gamma^{-1/2}_{r, \rho} (Y) \|_2^2 dr \Big )} \pl. \]
The Hilbert-Schmidt norms involved are obviously positive, and they are square integrable for strictly positive $\rho$. Hence we may interpret these norms as the absolute values of complex-valued, strictly positive functions of $r$. The Cauchy-Schwarz inequality finishes the proof.
\end{proof}

Shown explicitly as \citen[Lemma 1]{gao_complete_2021} and implicit from Lemma \ref{wnorm} or from the methods of \citen{junge_stability_2019} is a comparison property for inverse-weighted norms:
\begin{lemma} \label{normcomp}
For any positive operator $\rho$, strictly positive $\sigma$, and $c \in \RR^+$ such that $\rho \leq c \sigma$, $\|X\|_{\sigma^{-1}}^2 \leq c \|X\|_{\rho^{-1}}^2$.
\end{lemma}
Though the following Lemma appears in \citen{gao_complete_2021}, it follows directly from taking a well-known integral representation of the second derivative of relative entropy $D((\rho, \sigma)_t \| \sigma)$ with respect to $t$, where $(\rho, \sigma)_t$ is defined in the Lemma:
\begin{lemma} \label{normform}
Let $(\rho, \sigma)_t = (1-t) \sigma + t \rho$ for $t \in [0,1]$, and $\rho \in \text{supp}(\sigma)$. Then
\[ D(\rho \| \sigma) = \int_0^1 \int_0^s \|\rho - \sigma \|_{(\rho, \sigma)_t^{-1}}^2 dt ds \]
\end{lemma}
We also use the following ``key" Lemma of \citen{gao_complete_2021}: 
\begin{lemma}[Lemma 2 from \citen{gao_complete_2021}] \label{keylem}
Let $\rho \leq c \sigma$ for $c > 0$ and strictly positive densities $\rho$ and $\sigma$. Then
\[ \frac{c (\ln c - 1) + 1}{(c-1)^2} \| \rho - \sigma \|_{\sigma^{-1}}^2  \leq D(\rho \| \sigma) \leq \|\rho - \sigma \|^2_{\sigma^{-1}} \pl. \]
\end{lemma}
We make the simple observation:
\begin{rem} \label{inttrick}
For any pair of bounded, non-negative, Riemann integrable scalar functions $f(t), g(t)$ and any $s > 0$,
\begin{equation*}
\int_0^s f(t) g(t) dt \leq \frac{1}{2} \int_0^s (f^2(t) + g^2(t)) dt \pl.
\end{equation*}
This inequality follows from the inequality of arithmetic and geometric means.
\end{rem}
Finally, we will often use a well-known continuity argument to bypass the assumption of strict positivity:
\begin{rem} \label{rem:strictpos}
For all densities $\rho$ and $\sigma$, $D(\rho \| \sigma) = D(P_\rho \rho P_\rho \| P_\rho \sigma P_\rho)$, where $P_\rho$ is the projection to the support of $\rho$. Hence without loss of generality, we may restrict to the support of $\rho$, on which $\rho$ is strictly positive. Since $D(\rho \| \sigma)$ is finite if and only if the support of $\rho$ is contained in that of $\sigma$, we may assume it is finite if and only if $P_\rho \rho P_\rho$ and $P_\rho \sigma P_\rho$ are both strictly positive.
\end{rem}
Using these known results, we derive the new results in the rest of this Subsection.

\begin{prop} \label{ephi}
Let $\E,\Phi$ be quantum channels and $\rho$ be a density such that $\kappa \E(\rho) \leq \Phi(\rho) \leq c \E(\rho)$, and $\rho \leq c\E(\rho)$ for some $\kappa, c > 0$. Assume that $D(\Phi(\rho) \| \E(\rho)) \leq D(\rho \|\E(\rho)) < \infty$. Then
\begin{equation*}
D(\rho \| \Phi(\rho)) \leq \frac{a (c-1)^2}{\kappa (c (\ln c - 1) + 1)} D(\rho \| \E(\rho))
\end{equation*}
for constant $a$. In general, $a \leq 4$. If $\Phi = (1 - b) \E + b \tilde{\Phi}$ for some $b \in [0,1]$ and $\tilde{\Phi}$ such that $\tilde{\Phi} \E = \E$, then we may improve the constant to $a = 1 + 2\sqrt{b} + b$.
\end{prop}
\begin{proof}
Since $\rho$ and $\E$ are order-comparable, $\rho \not \in \text{supp}(\Phi(\rho))$ if and only if $\rho \not \in \text{supp}(\E(\rho))$.  If these conditions hold, then both sides of the inequality are infinite.

Otherwise, by Lemma \ref{keylem}, $D(\rho \| \Phi(\rho)) \leq \| \rho - \Phi(\rho)\|_{\Phi(\rho)^{-1}}^2$. Since $\E(\rho) \leq \kappa^{-1} \Phi(\rho)$, $\| \rho - \Phi(\rho)\|_{\Phi(\rho)^{-1}}^2 \leq \kappa^{-1} \| \rho - \Phi(\rho)\|_{\E(\rho)^{-1}}^2$ using Lemma \ref{normcomp}. By the triangle inequality (shown by Lemma \ref{wnorm}),
\[ \| \rho - \Phi(\rho)\|_{\E(\rho)^{-1}} \leq \| \rho - \E(\rho)\|_{\E(\rho)^{-1}} + \| \E(\rho) - \Phi(\rho)\|_{\E(\rho)^{-1}} \pl. \]
Again applying Lemma \ref{keylem}, $\| \rho - \E(\rho)\|_{\E(\rho)^{-1}} \leq \sqrt{g(c) D(\rho \| \E(\rho))}$, and $\| \E(\rho) - \Phi(\rho)\|_{\E(\rho)^{-1}} \leq \sqrt{g(c) D(\Phi(\rho) \| \E(\rho))}$, where $g(c)$ is given by Lemma \ref{keylem}. Hence,
\begin{equation*}
\begin{split}
\| \rho - \Phi(\rho)\|_{\E(\rho)^{-1}}^2 & \leq g(c) \big (D(\rho \| \E(\rho)) \\ & \pl + 2 \sqrt{D(\rho \| \E(\rho)) D(\Phi(\rho) \| \E(\rho))} + D(\Phi(\rho) \| \E(\rho)) \big ) \pl.
\end{split}
\end{equation*}
To complete the Proposition, we use the assumption that $D(\Phi(\rho) \| \E(\rho)) \leq D(\rho \|\E(\rho))$. For cases in which $\sigma$ is not strictly positive, we refer to Remark \ref{rem:strictpos}.
\end{proof}

\begin{proof}[Proof of Theorem \ref{revconv}]
By the assumptions of the Theorem, $\sigma, \omega$, and $\eta$ all have the same support. By the assumptions of the inequality and Remark \ref{rem:strictpos}, we may assume that $\rho, \sigma, \omega$, and $\eta$ are all strictly positive on this support.

If $\omega \geq (1-\zeta) \sigma$, then $\eta := (\omega - (1 - \zeta) \sigma)/\zeta$ is a positive semidefinite matrix. Furthermore, when $\omega, \sigma$ are densities,
\[ tr(\eta) = (\tr(\omega) - (1 - \zeta) \tr(\sigma)) / \zeta
	= (1 - 1 + \zeta) / \zeta = 1 \pl. \]
Hence $\eta$ is positive semidefinite with trace 1, the conditions for it to be a density matrix. It also then holds that
\[ \omega = (1-\zeta) \sigma + \zeta \eta \pl. \]
This convex combination form will be useful in understanding the proof.

First, let $X = \rho - \omega = (1 - \zeta) (\rho - \sigma) + \zeta (\rho - \eta)$, and we rewrite
\[ \rho - \sigma = \frac{1}{1 - \zeta} X + \frac{\zeta}{1 - \zeta} (\eta - \rho) \pl. \]
By the triangle inequality (Lemma \ref{wnorm}),
\begin{equation} \label{eq:tri1}
 \| \rho - \sigma \|_{\xi^{-1}} \leq 
	\frac{1}{1 - \zeta} \|X\|_{\xi^{-1}}
	+ \frac{\zeta}{1 - \zeta} \| \eta - \rho \|_{\xi^{-1}}
\end{equation}
for any density $\xi$ that is strictly positive on the support of $\eta$. Let $(\xi, \phi)_t := (1 - t) \phi + t \xi$ for any pair of densities $\xi, \phi$. In this particular situation,
\begin{equation}
\begin{split}
& D(\rho \| \sigma) = \int_0^1 \int_0^s  \|\rho - \sigma \|_{(\rho, \sigma)_t^{-1}}^2 dt ds \\
& \leq \frac{1}{(1-\zeta)^2} \int \int \Big ( \|\rho - \omega \|_{(\rho, \sigma)_t^{-1}}^2
	+ 2 \zeta \|\rho - \omega \|_{(\rho, \sigma)_t^{-1}} \|\rho - \eta \|_{(\rho, \sigma)_t^{-1}}
	+ \zeta^2 \|\rho - \eta \|_{(\rho, \sigma)_t^{-1}}^2 \Big ) dt ds \pl.
\end{split}
\end{equation}
By Remark \ref{inttrick},
\begin{equation} \label{eq:decomp1}
\begin{split}
D(\rho \| \sigma) & \leq \frac{1}{(1-\zeta)^2} \int \int \Big ( (1 + \zeta) \|\rho - \omega \|_{(\rho, \sigma)_t^{-1}}^2
	+ \zeta(1 + \zeta) \|\rho - \eta \|_{(\rho, \sigma)_t^{-1}}^2 \Big ) dt ds \pl.
\end{split}
\end{equation}
Remark \ref{rem:strictpos} allows us to assume that $\|\rho - \xi \|_{(\rho, \sigma)_t^{-1}}$ is finite over the range of integration for any $\xi \in \{ \sigma, \omega, \eta \}$.

By the Theorem's assumptions,
\[ (1 - \zeta) (\rho, \sigma)_t \leq (\rho, \omega)_t
	\leq (1 + \zeta(c - 1)) (\rho, \sigma)_t \]
for all $t \in [0,1]$. Hence via Lemma \ref{normcomp},
\begin{equation} \label{eq:rhocomp1}
\begin{split}
\|  \rho - \omega \|_{(\rho, \sigma)_t^{-1}}^2\leq (1 + \zeta(c - 1))
	\| \rho - \omega \|_{(\rho, \omega)_t^{-1}}^2 .
\end{split}
\end{equation}
Since $\eta \leq c \sigma$, $(\rho, \eta)_t \leq c (\rho, \sigma)_t$, and Lemma \ref{normcomp} yields that
\[ \| \rho - \eta \|_{(\rho, \sigma)_t^{-1}}^2 \leq c \| \rho - \eta \|_{(\rho, \eta)_t^{-1}}^2 \pl. \]
Returning to Equation \eqref{eq:decomp1} and applying Lemma \ref{normform} completes the proof of the first Equation in this Theorem.

For the second Equation, the triangle inequality and Remark \ref{inttrick} imply that
\[ \int \int \| \rho - \eta \|_{(\rho, \sigma)_t^{-1}}^2 dt ds \leq 
	2 \int \int \Big ( \| \rho - \sigma \|_{(\rho, \sigma)_t^{-1}}^2
+ \| \eta - \sigma \|_{(\rho, \sigma)_t^{-1}}^2 \Big ) dt ds \pl. \]
The first term integrates to $D(\rho \| \sigma)$. For the latter term, we note that $\sigma \leq (\rho, \sigma)_t/(1-t)$. Via Lemmas \ref{normcomp} and \ref{keylem},
\[ \| \eta - \sigma \|_{(\rho, \sigma)_t^{-1}}^2 \leq 
	(1-t)^{-1} \| \eta - \sigma \|_{\sigma}^2
	\leq \frac{(c-1)^2}{(1-t)(c(\ln c - 1) + 1)} D(\eta \| \sigma) \pl. \]
We calculate
\[ \int_0^1 \int_0^s \frac{dt ds}{1-t} = - \int_0^1 \ln(1-s) ds = 1 \pl. \]
 Hence
\[ \int \int \| \rho - \eta \|_{(\rho, \sigma)_t^{-1}}^2 dt ds
	\leq 2 \Big ( D(\rho \| \sigma) + \frac{(c-1)^2}{c(\ln c - 1) + 1} D(\eta \| \sigma) \Big )  \pl. \]
Returning to Equation \eqref{eq:decomp1},
\[ (1- 4 \zeta - \zeta^2) D(\rho \| \sigma) \leq (1 + \zeta c + \zeta^2 (c-1) ) D(\rho \| \omega)
	+ 2 \zeta(1 + \zeta) \frac{(c-1)^2}{c(\ln c - 1) + 1} D(\eta \| \sigma) \pl. \]
\end{proof}

\begin{cor}[Multiplicative Perturbation of Subalgebra-relative Entropy] \label{cor:simpleG}
Let $\rho$ be a density and $\E,\Phi$ be quantum channels such that $(1-\zeta) \E(\rho) + \zeta \Phi(\rho) \leq (1 + \zeta(c-1)) \E(\rho)$ (equivalent to $\Phi(\rho) \leq c \E(\rho)$) for constants $\zeta \in (0,1)$ and $c \geq 1$). Assume $\rho \in \text{supp}(\E(\rho))$. Furthermore, assume $\Phi \E = \E$. Then
\begin{equation*}
D(\rho \| (1-\zeta) \E(\rho) + \zeta \Phi(\rho)) \geq \beta_{c, \zeta} D(\rho \| \E(\rho))
\end{equation*}
for
\[ \beta_{c, \zeta} := \frac{1}{1 + \zeta c + \zeta^2 (c-1)}\Big ( 1 - 2 \zeta(1 + \zeta) \frac{(c-1)^2}{c(\ln c - 1) + 1} - 4 \zeta - \zeta^2 \Big ) = 1 - O(c \zeta)  \pl. \]
When $\E(\rho) = \id / d$ in dimension $d$, we may replace $\beta_{c, \zeta}$ by $(1-a)$ as in Theorem \ref{thm:relent}.
\end{cor}
\begin{proof}
We use the lower inequality of Theorem \ref{revconv} with $\sigma = \E(\rho)$, and $\omega = (1 - \zeta) \E(\rho) + \zeta \Phi(\rho)$. By the assumptions of the Theorem, $\Phi(\rho)$ is strictly positive on and has the same support as $\E(\rho)$, so we may assume strict positivity of both by Remark \ref{rem:strictpos}. The Corollary follows from the data processing inequality when $\Phi \E = \E$.
\end{proof}

For any conditional expectation $\E_\sigma$, there is a basis for which
\begin{equation} \label{eq:blockcondexp}
\E_{\sigma *}(\rho) = \oplus_l \tr_{B_l}(P_l \rho P_l) \otimes \sigma^{B_l} \pl,
\end{equation}
where $\sigma^B$ is a $\rho$-independent density on the subsystem $B_l$, and $P_l$ is a projector to the $l$th diagonal block. In particular, $\oplus_l P_l \rho P_l$ is a block diagonal matrix with entries from $\rho$, effectively removing all coherence between blocks. We may subsequently interpret each such block as a bipartite system $A_l \otimes B_l$, then trace out subsystem $B_l$ and replace it by the fixed state $\sigma^{B_l}$. This block diagonal form is applied commonly in operator algebras - see \citen{gao_unifying_2017} for discussion of the doubly stochastic case and \citen{gao_complete_2021} in general. It follows from the fundamental result of von Neumann \cite{von-neumann_rings_1949} that every von Neumann algebra is decomposable as a direct integral of factors, which are von Neumann algebras in which only the identity commutes with all elements.  In infinite dimensions, ``$\oplus_l$" may take the form of a direct integral rather than a sum.
\begin{prop} \label{near}
Let $\E_{\sigma*}$ be a conditional expectation weighted by normal, faithful density $\sigma$ and $\Phi$ a quantum channel such that $\E_{\sigma*} \Phi = \Phi \E_{\sigma*} = \E_{\sigma*}$, both defined on matrices of dimension $d$. Let $\ket{\psi}$ be a $d \times d$ Bell state given by $\ket{\psi} = \frac{1}{\sqrt{d}} \sum_{i=1}^d \ket{i} \otimes \ket{i}$ in the computational basis, which we may assume without loss of generality is compatible with the block diagonal form of $\E_{\sigma*}$ as in \eqref{eq:blockcondexp}. Let $d_l$ denote the dimension of the $l$th diagonal block of $\id \otimes \E_{\sigma*}(\rho)$ and $m_l$ denote the dimension of the traced subsystem in that block. Then
\begin{equation} \label{eq:comboall}
\Phi \geq_{cp} (1 - \zeta) \E_{\sigma*} \text{, or equivalently } \Phi = (1 - \zeta) \E_{\sigma*} + \zeta \tilde{\Phi}
\end{equation}
for some channel $\tilde{\Phi}$ such that $\tilde{\Phi} \E_{\sigma*} = \E_{\sigma*} \tilde{\Phi} = \E_{\sigma*}$ whenever
\[ 1 \geq \zeta \geq \| \Phi(\ketbra{\psi}) - \E_{\sigma*}(\ketbra{\psi}) \|_\infty d \max_l m_l / (d_l \lambda_{min}(\sigma^{B_l})) \pl. \]
If $\E_{\sigma*} = \E$ as a unital conditional expectation weighted by the trace, then the above condition reduces to:
\[1 \geq \zeta \geq \| \Phi(\ketbra{\psi}) - \E(\ketbra{\psi}) \|_\infty d \max_l m_l^2 / d_l \pl, \]
and $\tilde{\Phi}$ is assured to be unital. If $\E(\rho) = \id/d$ for all $\rho$, then we may replace the condition on $\zeta$ by
\[ 1 \geq \zeta \geq d \sup_\eta \| \Phi(\eta) - \E(\eta) \|_\infty \pl, \]
where the supremum is over normalized densities.
\end{prop}
\begin{proof}
For $k \in 1...d$ and a $d \times d$ matrix $X$, let $\lambda_k(X)$ denote $X$'s $k$th eigenvalue in non-increasing order and $\lambda_{max}(X) = \lambda_1(X), \lambda_{min}(X) = \lambda_d(X)$. Let $\rho = \ketbra{\psi}$, the $d \times d$ Bell state.

For a pair of Hermitian matrices $X,Y$ and any $k \in 1...d$, Weyl's inequality \cite{horn_topics_1991} states that
\[ \lambda_d(Y) \leq \lambda_k(X + Y) - \lambda_k(X) \leq \lambda_1(Y) \pl. \]
Via Weyl's inequality, for each value of $k$, setting $X = \Phi(\rho) - (1 - \zeta)\E_{\sigma *}(\rho)$ and $Y = \E_{\sigma*}(\rho) - \Phi(\rho)$,
\[ \lambda_k(\E_{\sigma*}(\rho) - (1 - \zeta)\E_{\sigma *}(\rho)) - \lambda_k(\Phi(\rho) - (1 - \zeta) \E_{\sigma*}(\rho)) 
	\leq \lambda_1(\E_{\sigma*}(\rho) - \Phi(\rho)) \pl. \]
The absolute value of the right hand side is upper-bounded by the $\infty$-norm distance. Re-arranging and combining terms,
\begin{equation} \label{eq:evalinfnorm}
\lambda_k(\Phi(\rho) - (1 - \zeta) \E_{\sigma*}(\rho)) \geq 
	\zeta \lambda_k(\E_{\sigma*}(\rho)) - \| \Phi(\rho) - \E_{\sigma*}(\rho) \|_\infty \pl.
\end{equation}
Invoking Choi's theorem on completely positive maps, we recall the assumption that $\rho$ is a $d \times d$ Bell pair that we may express in a basis that is compatible with the block diagonal form of Equation \eqref{eq:blockcondexp}. We estimate $\lambda_k(\E_{\sigma*}(\rho))$. Let us consider $\E_{\sigma*}$ to act on the outer-indexed subsystem of the Bell pair. We may decompose $\E_{\sigma *} = \tilde{\E}_{\sigma *} \E_{bl}$, where $\E_{bl}(\rho) = \oplus_l P_l \rho P_l$, and $\tilde{\E}_{\sigma *}$ traces and replaces the $B_l$ subsystem within each block by $\sigma^{B_l}$. Each block in $\E_{bl}(\rho)$ is then itself a maximally entangled state in dimension $d_l \times d_l$ times a factor of $d_l / d$. We then apply a partial trace with replacement by complete mixture of an $m_l$-dimensional subsystem from each $l$th block, which because of the maximal entanglement yields $m_l^2$ distinct eigenvalues of magnitude each at least $d_l \lambda_{min}(\sigma^{B_l}) / (d m_l)$. To ensure positivity of the right hand side of Equation \eqref{eq:evalinfnorm}, we may choose any $\zeta \geq \| \Phi(\ketbra{\psi}) - \E_{\sigma*}(\ketbra{\psi}) \|_\infty d \max_l m_l / (d_l \lambda_{min}(\sigma^{B_l}))$. In the unital case, $\sigma^{B_l} = \id / m_l$, so $\lambda_{min}(\sigma^{B_l}) = 1 / m_l$. The map given by $\Phi - (1 - \zeta) \E_{\sigma*}$ is then completely positive by Choi's theorem. Since both $\Phi$ and $\E_{\sigma*}$ are trace-preserving, $tr((\Phi(\rho) - (1 - \zeta) \E_{\sigma *}(\rho))(X)) = \zeta tr(X)$ for any matrix $X$. Hence $\tilde{\Phi} = (\Phi - (1 - \zeta) \E_{\sigma *})/\zeta$ is a quantum channel. By linearity and the assumption that $\E_{\sigma *} \Phi = \Phi \E_{\sigma *} = \E_{\sigma *}$,
\[\E_{\sigma *} \tilde{\Phi}  = \tilde{\Phi} \E_{\sigma *} = \frac{1}{\zeta}(\E_{\sigma *} - (1 - \zeta) \E_{\sigma *}) = \E_{\sigma *} \pl. \]
 We may then rewrite
\[ \Phi = (1 - \zeta) \E_{\sigma *} + \zeta \tilde{\Phi}\]
as a convex combination of channels. To see that $\tilde{\Phi}$ is unital when $\sigma = \id / d$ and $\E_{\sigma *} = \E$, if $\tilde{\Phi}(\id) \neq \id$, then $\Phi(\id)$ cannot be $\id$, which would contradict the assumption that $\E$ is unital. 

When $\E(\eta) = \id/d$ for any $\eta$, we can simplify the above argument by noting that $\Phi(\eta)$ and $\id/d$ are always simultaneously diagonal. Whenever $(1 - \zeta) \leq d \lambda_k(\Phi(\eta))$ for all $k$, $\Phi(\eta) - (1 - \zeta) \id / d \geq 0$. Via the triangle inequality, we have the condition that $1/d - \|\Phi(\eta) - \E(\eta)\|_\infty - (1-\zeta)/d \geq 0$. Combining terms and re-arranging yields the final part of the Proposition.
\end{proof}
Even though Proposition \ref{near} depends on the dimension, this is the minimal dimension on which the conditional expectation may act, not including any extra, untouched subsystems. Hence the same bound applies for $\E \otimes \id^B, \Phi \otimes \id^B$ independently of $B$'s dimension, and any unitary embedding of such an extension preserves constants. Proposition \ref{near} is similar to \citen[Lemma VI.8]{bardet_group_2021}, which is attributed in that work to non-author Li Gao.

%% file: condex_combo.tex
\section{Combining Conditional Expectations} \label{sec:asa}
We start this section by recalling basic facts about conditional expectations. A doubly-stochastic conditional expectation $\E_\N$ is a projector to matrices in algebra $\N$ that is self-adjoint under the trace. For this reason, we do not distinguish between $\E_\N$ and $\E_\N^*$ or $\E_{\N*}$, which respectively denote the dual and predual of $\E_\N$. For a subalgebra $\N$, the doubly-stochastic conditional expectation $\E_\N$ is unique. We also consider weighted conditional expectations with block diagonal decompositions of the form in Equation \eqref{eq:blockcondexp}, which we rewrite equivalently as
\begin{equation} \label{eq:weightblock}
\E_{\N, \sigma *}(\rho) = \oplus_i p_i (\rho^{A_i} \otimes \sigma_i^{B_i})
\end{equation}
for some probability distribution $(p_i)$. One could specify an overall weighting density $\sigma = \oplus_i (\id^{A_i}/|A_i| \otimes \sigma_i)$. Given an arbitrary $\sigma$ not necessarily in this form, we recall the self-adjoint conditional expectation to the commutant algebra,
\begin{equation}
\E_{\N'}(\rho) = \oplus_i (\id^{A_i} / |A_i| \otimes \rho^{B_i}) \pl.
\end{equation}
For any $\sigma$, $\E_{\N'}(\sigma)$ yields each $\sigma^{B_i}$ in the block decomposition of Equation \eqref{eq:weightblock}. Given a weighting state $\sigma$ and a set of conditional expectations $\{\E_j\}_{j=1}^J$, we may thereby construct the set $\{\E_{j, \sigma *} = \E_{j, \E_{\N'}(\sigma) *} \}$ unambiguously.

The projections denoted $\E_{\N, \sigma}$ and $\E_{\N, \sigma *}$ are respective adjoints under the trace. Let $\tilde{\sigma}_\N$ be the unnormalized density in finite dimension $d$ given by
\begin{equation} \label{eq:unnorm}
\tilde{\sigma}_\N := d \E_{\N'}(\sigma) \pl.
\end{equation}
It then holds for any density $\rho$ that
\[ \E_{\N, \sigma *}(\rho) = \tilde{\sigma}_\N^{1/2} \E_{\N}(\rho) \tilde{\sigma}_\N^{1/2} \]
and for any operator $X$ that
\[ \E_{\N, \sigma}(X) = \E_\N(\tilde{\sigma}_\N^{1/2} X \tilde{\sigma}_\N^{1/2}) \pl. \]
Via its block diagonal form, we see that $\E_{\N, \sigma *}$ is idempotent, as is its adjoint. Finally, $\E_{\N} = \E_{\N, \tau *} = \E_{\N, \tau}$, where $\tau$ is the normalized identity or trace. It is simple to observe using the block diagonal forms that if $\E \E_j = \E$ for unweighted conditional expectations $\E_j$ and $\E$, then $\E_{\sigma *} \E_{j, \sigma *} = \E_{\sigma *}$ for the $\sigma$-weighted versions.

\begin{lemma}[Chain Rule] \label{lem:chainexp}
Let $\omega$ be a density and $\E_{\sigma}$ be a conditional expectation such that $\E_{\sigma *}(\omega) = \omega$. Then for any density $\rho$,
\begin{equation*}
D(\rho \| \omega) = D(\rho \| \E_{\sigma *}(\rho)) + D(\E_{\sigma *}(\rho) \| \omega)
\end{equation*}
\end{lemma}
This equality is well-known. Nonetheless, we include a simple proof for finite dimensions. Here the logarithm as denoted ``$\log$" is with respect to an arbitrary base.
\begin{proof}
Let $\E$ denote the doubly stochastic conditional expectation such that \[\E_\sigma(X) = \E(\tilde{\sigma}^{1/2} X \tilde{\sigma}^{1/2})\] for an operator $X$, where $\tilde{\sigma}$ is the unnormalized density as in Equation \eqref{eq:unnorm}. Then
\begin{equation} \label{eq:wchain1}
\begin{split}
D(\rho \| \E_{\sigma *}(\rho)) + D(\E_{\sigma *}(\rho) \| \omega)
& = \tr(\rho \log \rho - \rho \log (\E_{\sigma *}(\rho)) + \E_{\sigma *}(\rho) \log (\E_{\sigma *}(\rho)) - \E_{\sigma *}(\rho) \log \omega ) \\
& = \tr(\rho \log \rho - \rho \log (\E_{\sigma *}(\rho)) + \rho \E_{\sigma}(\log (\E_{\sigma *}(\rho))) - \rho \E_{\sigma}( \log \omega)) \pl.
\end{split}
\end{equation}
Examining the logarithm of the block diagonal form in Equation \eqref{eq:weightblock},
\[ \log \E_{\sigma *}(\rho) = \oplus_i (\log(\rho_i) \otimes \id^{B_i} + \id^{A_i} \otimes \log(\sigma_i)) \pl, \]
and
\[ \E_{\sigma} (\log \E_{\sigma *}(\rho)) = \oplus_i (\log(\rho_i) \otimes \id^{B_i}
	+ \id^{A_i} \otimes \id^{B_i} \tr(\tilde{\sigma}_i^{1/2} \log(\sigma_i) \tilde{\sigma}_i^{1/2})) \pl. \]
Let
\[ \eta_\rho := \E_{\sigma} (\log \E_{\sigma *}(\rho)) - \log \E_{\sigma *}(\rho)
	= \oplus_i (\id^{A_i} \otimes \id^{B_i} \tr(\tilde{\sigma}_i^{1/2} \log(\sigma_i)\tilde{\sigma}_i^{1/2})
	- \id^{A_i} \otimes \log(\sigma_i)) \pl. \]
Since $\eta_\rho$ has no dependence on $\rho$, we define $\eta_\omega = \eta_\rho$ analogously. Comparing to Equation \eqref{eq:wchain1},
\[ D(\rho \| \omega) =  D(\rho \| \E_{\sigma *}(\rho)) + D(\E_{\sigma *}(\rho) \| \omega)
	+ \tr(\rho(\eta_\omega - \eta_\rho)) \pl, \]
and $\eta_\omega - \eta_\rho = 0$.
\end{proof}

\begin{lemma} \label{lem:subadd}
Let $\E_\sigma$ be a conditional expectation and $\rho, \omega$ be densities. Then
\begin{equation*}
D(\rho \| \E_{\sigma *}(\rho)) + D(\rho \| \omega) \geq D(\rho \| \E_{\sigma *}(\omega))
\end{equation*}
\end{lemma}
\begin{proof}
By data processing on the 2nd term,
\begin{equation*}
D(\rho \| \E_{\sigma *}(\rho)) + D(\rho \| \omega)
	\geq D(\rho \| \E_{\sigma *}(\rho)) + D(\E_{\sigma *}(\rho) \| \E_{\sigma *}(\omega))  \pl.
\end{equation*}
By Lemma \ref{lem:chainexp} and the idempotence of conditional expectations, we obtain the Lemma.
\end{proof}

In the rest of this Section, we may sometimes drop the explicit subscript of the weighting state, e.g. writing $\E_{j *}$ for a weighted conditional expectation. This is to reduce the verbosity of notation when considering sets of potentially weighted conditional expectations.
\begin{lemma} \label{lem:asawithf}
Let $\{(\N_j, \E_{j *} ) : j = 1...J \in \NN, \N_j \subseteq \M\}$ be a set of von Neumann algebras and associated (predual) conditional expectations within von Neumann algebra $\M$ weighted respectively by densities $(\sigma_j)$. Let $\E$ be a channel such that $\E \E_j = \E_j \E = \E$ for each $\E_j$.

Let $S = \cup_{m \in \NN} \{1...J\}^{\otimes m}$ be the set of sequences of indices. For any $s \in S$, let $\E^s$ denote the composition of conditional expectations $\E_{j_1 *} ... \E_{j_m *}$ for $s = (j_1, ..., j_m)$. Let $\mu : S  \rightarrow [0,1]$ be a probability measure on $S$ and $k_{j,s}$ upper bound the number of times $\E_{j *}$ appears in each sequence $s$. If
\[ \sum_{s \in S} \mu(s) \E^s(\rho) = (1-\epsilon) \E(\rho) + \epsilon \omega \]
for densities $\rho$ and $\omega$, then
\[ \sum_{s \in S} \mu(s) \sum_{j=1}^J k_{s,j} D(\rho \| \E_j(\rho)) \geq D(\rho \| (1-\epsilon) \E(\rho) + \epsilon \omega) \pl. \]
\end{lemma}
\begin{proof}
For each $s \in S$, we apply Lemma \ref{lem:subadd} iteratively, finding that
\[ \sum_{j=1}^J k_{s,j} D(\rho \| \E_{j *}(\rho)) \geq D(\rho \| \E^s(\rho)) \pl. \]
By convexity, we move the weighted average over $s$ inside the relative entropy, completing the Lemma.
\end{proof}

\begin{theorem}[Restatement of Theorem \ref{asa2}]
Let $\{(\N_j, \E_{j *} ) : j = 1...J \in \NN, \N_j \subseteq \M\}$ be a set of $J \in \NN$ von Neumann algebras and associated (predual) conditional expectations within von Neumann algebra $\M$ and weighted respectively by densities $(\sigma_j)$. Let $\E$ be a channel such that $\E \E_j = \E_j \E = \E$ for each $\E_j$.

Let $S = \cup_{m \in \NN} \{1...J\}^{\otimes m}$ be the set of finite sequences of indices. For any $s \in S$, let $\E^s$ denote the composition $\E_{j_1 *} ... \E_{j_m *}$ for $s = (j_1, ..., j_m)$. Let $\mu : S  \rightarrow [0,1]$ be a probability measure on $S$ and $k_{j,s}$ upper bound the number of times $\E_{j *}$ appears in each sequence $s$. If
\begin{equation} \label{eq:kzetacond2}
 (1-\zeta) \E \leq_{cp} \sum_{s \in S} \mu(s) \E^s \leq_{cp} (1 + \zeta(c-1)) \E  \pl,
\end{equation}
then $\E$ is a projection, and for $\beta_{c,\zeta}$ given in Corollary \ref{cor:simpleG} and all input densities $\rho$ (including those with arbitrary extensions to auxiliary systems),
\[ \sum_{s \in S} \mu(s) \sum_j k_{s,j} D(\rho \| \E_{j *}(\rho)) \geq \beta_{c, \zeta} D(\rho \| \E(\rho)) \pl. \]
\end{theorem}
\begin{proof}[Proof of Theorem \ref{asa2}]
Note that Equation \eqref{eq:kzetacond2} implies that
\begin{equation} \label{eq:comboform}
\sum_{s \in S} \mu(s) \E^s(\rho) = (1-\zeta) \E(\rho) + \zeta \Phi(\rho)
\end{equation}
for some channel $\Phi$ and constant $c$ such that $\Phi(\rho) \leq c \E(\rho)$.

We may rewrite the left inequality of Equation \eqref{eq:kzetacond2} using $\E_j \E = \E \E_j = \E$ as
\[ \sum_{s \in S} \mu(s) \E^s = (1 - \zeta) \E + \zeta \Psi \]
for some channel $\Psi$ such that $\E \Psi = \Psi \E = \E$. Hence
\[ \Big ( \sum_{s \in S} \mu(s) \E^s \Big )^m = (1 - \zeta^m) \E + \zeta \tilde{\Psi}_m \]
for some channel $\tilde{\Psi}_m$ such that $\E \tilde{\Psi}_m = \tilde{\Psi}_m \E = \E$. As $\tilde{\Psi}_m$ is bounded in diamond norm, and $\zeta^m \rightarrow 0$ as $m \rightarrow \infty$, this limit exists, and $\lim_{m \rightarrow \infty} (\sum_{s \in S} \mu(s) \E^s)^m = \E$. We further see that $\E$ is a projection, since idempotence is implied by the existence of the limit.

We have by the assumptions of the Theorem and Lemma \ref{lem:asawithf} that
\[ \sum_{s \in S} \mu(s) \sum_{j=1}^J k_{s,j} D(\rho \| \E_j(\rho)) \geq D(\rho \| (1-\zeta) \E(\rho) + \zeta \Phi(\rho)) \pl. \]
We then note that since
\[ ((1 - \zeta) \E + \zeta \Phi)  \E =
	(1 - \zeta) \E + \zeta \Phi \E = \E \pl, \]
it must hold that $\E \Phi = \E$, and similarly, $\Phi \E = \E$. We then use Corollary \ref{cor:simpleG}.
\end{proof}
One may substitute Equation \eqref{eq:comboform} for Equation \eqref{eq:kzetacond2} for particular states and may find that in some cases, doing so yields better constants than would the full cp-order inequality. The Theorem is nonetheless stated in cp-order form for its easier interpretability and connection to other results in this paper.

\begin{cor} \label{index}
Let $\M$ denote a von Neumann algebra with density $\rho$ on quantum system $A$, $\N_j$ denote the subalgebra projected to by each doubly stochastic conditional expectation $\E_j$ for $j \in 1...J$, and $\N$ denote the intersection algebra projected to by $\E$. We may extend $\rho$ to a bipartite density $\rho^{AB}$, where $\E_j$ acts on $A$ as does $\E$. Let $\tilde{\E}_j = \E_j \otimes \id^B$, and $\tilde{\E} = \E \otimes \id^B$.

Let $S = \cup_{m \in \NN} \{1...J\}^{\otimes m}$ be the set of finite sequences of indices. For any $s \in S$, let $\E^s$ denote the composition $\E_{j_1 *} ... \E_{j_m *}$ for $s = (j_1, ..., j_m)$. Let $\mu : S  \rightarrow [0,1]$ be a probability measure on $S$ and $k_{j,s}$ upper bound the number of times $\E_{j *}$ appears in each sequence $s$.If $|B| \leq |A|$ (including if $B$ is a trivial system of dimension 1), and
\begin{equation} \label{eq:indexzetacond}
\sum_{s \in S} \mu(s) \E^s(\rho) = (1-\zeta) \E(\rho) + \zeta \Phi(\rho)
\end{equation}
for some channel $\Phi$ and $\zeta > 0$, then for any $\epsilon \leq \zeta$,
\begin{equation} \label{eq:indexbound}
\frac{\lceil \log_{\zeta} \epsilon \rceil}{\beta_{c, \epsilon}} \sum_{s \in S} \mu(s) \sum_j k_{s,j} D(\rho \| \tilde{\E}_j(\rho)) \geq
	D(\rho \| \tilde{\E}(\rho))
\end{equation}
with $c = C(\M \otimes \B : \N \otimes \B)$, and $\B$ being the algebra of bounded operators on $B$. If $|B| \geq |A|$, then we may set $c = C_{cb}(\M : \N)$ for a complete, strong quasi-factorization.

For any $j \in 1...J$, we may replace $C(\M : \N)$ or $C_{cb}(\M : \N)$ respectively by $C(\N_j : \N)$ or $C_{cb}(\N_j : \N)$ at the cost of taking $k_{s,j} \rightarrow k_{s,j} + 1$.

For weighted conditional expectations when $\M = \BB(A)$, Equation \eqref{eq:indexbound} holds respectively with $c = \lambda_{min}(\sigma) C(\M \otimes \B : \N \otimes \B) / |A|$ or $\lambda_{min}(\sigma) C_{cb}(\M : \N) /|A|$, where $\lambda_{min}(\sigma)$ is the minimal eigenvalue.
\end{cor}
\begin{proof}
First, $\E$ projects to a fixed point subspace of $\{\E_j\}$ for all $j$, so for any sequence $s$ of $k$-many $j$ indices, $\E$ also projects to a fixed point subspace of $\prod_{j = 1...k} \E_j$. For all densities $\rho$, it must then hold that $\E \Phi(\rho) = \Phi \E(\rho) = \E(\rho)$, so $\Phi$ also has $\N$ as a fixed point subalgebra.

Note that $((1 - \zeta) \E + \zeta \Phi)^\kappa(\rho) = (1 - \zeta^\kappa) \E(\rho) + \zeta^\kappa \Phi^\kappa(\rho)$. As in  Lemma \ref{lem:asawithf} and for any $\kappa \in \NN$, we may apply Lemma \ref{lem:subadd} and convexity of relative entropy iteratively to obtain that
\begin{equation} \label{eq:cork}
\begin{split}
\kappa \sum_{s \in S} \mu(s) \sum_j k_{s,j} D(\rho \| \E_j(\rho))
	& \geq D\Big (\rho \Big \| \Big ( \sum_{s \in S} \mu(s) \E^s(\rho) \Big )^{\kappa} (\rho) \Big )
	\\ & \geq D(\rho \| (1 - \zeta^\kappa) \E(\rho) + \zeta^\kappa \Phi^\kappa(\rho))
\end{split}
\end{equation}
To achieve a desired $\epsilon$, we may take $\kappa = \lceil \log_{\zeta} \epsilon \rceil$. For convenience of notation, we denote $C := C(\M : \N)$, and $C_{cb} := C_{cb}(\M : \N)$. It follows from the definitions of $C$ and $C_{cb}$ that
\begin{enumerate}
	\item $\Phi(\rho) \leq C \E(\Phi(\rho)) = C \E(\rho)$ for all $\rho$.
	\item $\Phi \otimes \id^B(\rho) \leq C_{cb} \E \otimes \id^B(\Phi \otimes \id^B(\rho)) = C_{cb} \E \otimes \id^B(\rho)$ for $\rho^{AB}$.
\end{enumerate}
To obtain Equation \eqref{eq:indexbound}, we apply Corollary \ref{cor:simpleG}, substituting $\zeta^\kappa$ for $\zeta$ and $C$ or $C_{cb}$ for $c$.

Recall that via the chain rule (Lemma \ref{lem:chainexp}),
\[ D(\rho \| \E(\rho)) = D(\rho \| \E_{j}(\rho)) + D(\E_{j}(\rho) \| \E(\rho)) \pl. \]
A single term of $D(\rho \| \E_j(\rho))$ cancels the first term on the right hand side of the above, so for any $\alpha \geq 1$,
\begin{equation*}
\begin{split}
& \alpha \sum_{s \in S} \mu(s) \Big ( (k_{s,j} + 1) D(\rho \| \E_{j}(\rho)) + \sum_{l \neq j} k_{s,l} D(\rho \| \E_{l}(\rho)) \Big )
	 - D(\rho \| \E(\rho))
	\\ & \geq \alpha \sum_{s \in S} \mu(s) \sum_l k_{s,l} D(\rho \| \E_{l}(\rho)) - D(\E_{j}(\rho) \| \E(\rho)) \pl.
\end{split}
\end{equation*}
By Equation \eqref{eq:cork} and with an additional round of the data processing inequality,
\begin{equation*}
\begin{split}
& \alpha \sum_{s \in S} \mu(s) \sum_l k_{s,l} D(\rho \| \E_{l}(\rho)) - D(\E_{j}(\rho) \| \E(\rho)) \\
	& \pla \geq \alpha D(\rho \| (1 - \zeta^\kappa) \E(\rho) + \zeta^\kappa \Phi^\kappa(\rho)) - D(\E_{j}(\rho) \| \E(\rho)) \\
	& \pla \geq \alpha D(\E_j(\rho) \| (1 - \zeta^\kappa) \E(\rho) + \zeta^\kappa \E_j(\Phi^\kappa(\rho))) - D(\E_{j}(\rho) \| \E(\rho)) \pl.
\end{split}
\end{equation*}
Since $\E_j \Phi^\kappa \E = \E \E_j \Phi^\kappa = \E$, we may use \ref{cor:simpleG} as before. Hence we may replace the index of $\N$ in $\M$ by that of $\N$ in $\N_j$ at the cost of taking $k_{s,j} \rightarrow k_{s,j} + 1$ for each $s$ with non-zero measure.

For weighted conditional expectations, the block diagonal form in Equation \eqref{eq:weightblock} still yields an upper bound on $c$ such that $\rho \leq c \E_{\sigma *}(\rho)$ for all $\rho$. Whenever $\Phi \E_{\sigma *} = \E_{\sigma *} \Phi = \E_{\sigma *}$, the same value of $c$ suffices for $\Phi(\rho)$.
\end{proof}

\begin{proof}[Proof of Remark \ref{tightness}]
Proposition \ref{near} shows that for a channel $\Phi$ and conditional expectation $\E_{\sigma *}$ such that $\Phi \E_{\sigma *} = \E_{\sigma *} \Phi = \E_{\sigma *}$ with $\Phi(\rho)$ sufficiently close to $\E_{\sigma *}(\rho)$ in operator norm distance, $\Phi$ is a non-trivial convex combination of $\E_{\sigma *}$ with another channel $\tilde{\Phi}$ such that $\tilde{\Phi} \E_{\sigma *} = \E_{\sigma *} \tilde{\Phi} = \E_{\sigma *}$. Via \citen[Theorem 6.7]{junge_noncommutative_2007}, an iterated product of doubly-stochastic conditional expectations $(\E_
{1} ... \E_{J})^n$ converges to the intersection algebra's conditional expectation as $n \rightarrow \infty$. Together, these suffice to show that a quasi-factorization inequality with some constant holds.

Corollary \ref{cor:simpleG} shows a correction factor scaling as $1 + O(c \zeta)$ for small $\zeta$. Hence as $\zeta \rightarrow 0$ with constant $c$, the correction factor approaches 1. This shows asymptotic tightness.
\end{proof}
The caveat to Remark \ref{tightness} is that as $\zeta \rightarrow 0$, $\E$ must not change. For example and as discussed in Section \ref{sec:ucr}, if we take two incompatible measurement bases at arbitrarily small angle to each other, they will converge to the same basis, though for any finite angle $\E$ will be a trace of the subsystem.

%% file: applications.tex
\section{Applications} \label{sec:applications}
Well-known uses of quasi-factorization (sometimes via Proposition \ref{prop:mlsimerge} and similar) include strengthening quantum uncertainty principles and estimating decay rates to thermal equilibrium in many-body systems \cite{bardet_approximate_2022, capel_modified_2020}. MLSI and similar estimates also bound decoherence times and quantum capacities \cite{bardet_group_2021, gao_relative_2020}. Taking this idea a step further, quasi-factorization gives strong bounds on certain combinations of quantum resources, such as the relative entropy of coherence in incompatible bases \cite{winter_operational_2016} or the asymmetry with respect to overlapping symmetry groups \cite{vaccaro_tradeoff_2008, gour_measuring_2009, marvian_extending_2014}. Similarly, combining MLSI is a powerful way to estimate decoherence or resource decay for systems undergoing several noise processes simultaneously. In Subsection \ref{sec:ucr} we show how the improved quasi-factorization yields strong, asymptotically tight, uncertainty-like bounds for mixtures.
\subsection{Uncertainty Relations} \label{sec:ucr}
Let $\S$ and $\T$ correspond to bases of subsystem $A$ within $A \otimes B$ such that $|A| = d$, not necessarily mutually unbiased. Let $\{\ket{i_S} : i = 1...d \}$ and $\{\ket{i_T} : i = 1...d\}$ be the states of these bases. Corresponding to each basis is a conditional expectation that is the pinching map for that basis. These bases and pinching maps naturally extend to bipartite $\rho^{AB}$ with $|A| = d$ and conditional expectations acting on $A$. On a joint state $\rho$,
\[ \ket{j} \braket{i | \rho | i} \bra{j} = (\ket{j} \bra{i}^A \otimes \id^B) \rho (\ket{i} \bra{j}^A \otimes \id^B) = \ketbra{j}^A \otimes \braket{i | \rho | i}^B = tr(\braket{i | \rho | i} ) \ketbra{j} \otimes \rho^B_i\]
for any $i,j \in 1...d$, with density $\rho^B_i = \braket{i | \rho | i} / \tr(\braket{i | \rho | i})$. Explicitly, the pinching maps are given by
\[ \E_\S(\rho) = \sum_{i=1}^d \ketbra{i_S} \otimes \braket{i_S | \rho | i_S} , \text{ and }
	\E_\T(\rho) = \sum_{i=1}^d \ketbra{i_T} \otimes \braket{ i_T |  \rho | i_T} \]
for any input density $\rho$.
 Hence
 \begin{equation*}
 \begin{split}
     (\E_\S \E_\T \otimes \id^B)(\rho) & = \sum_{i=1}^d \ketbra{i_S} \otimes \Big ( \sum_{j=1}^d |\braket{i_S | j_T}|^2 \braket{j_T | \rho |j_T} \Big )
     \\ & = \sum_{i=1}^d \ketbra{i_S} \Big ( \sum_{j=1}^d \tr(\braket{j_T | \rho^A | j_T}) |\braket{i_S | j_T}|^2 \Big ) \otimes \tilde{\rho}^B_{i} \pl
 \end{split}
 \end{equation*}
for some densities $(\tilde{\rho}^B_{i})_{i=1}^d$. If $\S$ and $\T$ correspond to mutually unbiased bases, then $\E_\S \E_\T(\rho) = \E = \id/d \otimes \rho^B$. When the bases are not mutually unbiased, $\E_\S \E_\T \neq \E$. In full generality, $\E_\S$ and $\E_\T$ might leave mutual subspaces invariant. We will however assume that $0 < |\braket{i_S|j_T}| < 1$ for all $i,j \in 1...d$, excluding cases that leave subspaces of $A$ invariant. Let $\xi = \min_{i,j} |\braket{i_S | j_T}|^2$, which by our assumptions is larger than zero. Note that $\xi \in (0, 1/d]$. Then since each $\ketbra{i_S} \otimes \braket{j_T | \rho | j_T}$ is positive semidefinite,
\[ (\E_\S \E_\T \otimes \id^B)(\rho) \geq \xi \sum_{i=1}^d \ketbra{i_S} \otimes \sum_{j=1}^d \braket{j_T | \rho |j_T} = d \xi  \frac{\id}{d} \otimes \rho^B \pl. \]
Repeated applications increasingly replace a density by one that is completely mixed on $A$. Hence $\E_\S \E_\T \geq_{cp} (d \xi) \E$, and $\E$ is a completely depolarizing channel on the $A$ subsystem. By the assumption that $|\braket{i_S | j_T}|^2 \leq 1$, we obtain using the Choi matrix that $\E_\S \E_\T \leq_{cp} d \E$. For any input $\rho$, if $|B|=1$, then $\E_\S \E_\T(\rho) \leq d \E(\rho)$. These observations will allow us to derive entropic uncertainty-like bounds from quasi-factorization.

First, we recall some known bounds for comparison. The conventional Maassen-Uffink uncertainty relation \cite{maassen_generalized_1988} states that
\begin{equation} \label{eq:maucr}
D(\rho \| \E_\S(\rho)) + D(\rho \| \E_\T(\rho)) \geq  D(\rho \| \id / d) - \log (d \max_{i,j} |\braket{i_S | j_T}|^2) \pl,
\end{equation}
which may be extended by an auxiliary system \cite{berta_uncertainty_2010}. When $\rho \approx \id/d$ and/or when there is high overlap between bases, the right hand side of the conventional uncertainty relation becomes negative, and the bound becomes trivial. In contrast, $\alpha$-(C)SQF still gives a positive, non-trivial bound on the sum of basis entropies. We also recall the result of \citen[Corollary 2]{bardet_approximate_2022},
\begin{equation} \label{eq:bardet}
D(\rho \| \E_\S(\rho)) + D(\rho \| \E_\T(\rho)) \geq (1 - d \max_{i,j} \big |\pl |\braket{i_S|j_T}|^2 - 1/d \pl \big | )  D(\rho \| \id / d) \pl.
\end{equation}
This bound remains non-trivial as $\rho \rightarrow \id/d$ but fails if $d \max_{i,j} \big |\pl |\braket{i_S|j_T}|^2 - 1/d \pl \big |$ is too large. The quasi-factorization from \citen{gao_complete_2021} also yields an uncertainty-like bound with extendibility by an auxiliary system, in this case stable under extension by auxiliary systems.

Now we apply quasi-factorization to derive bounds that are tighter in some circumstances. Letting $\zeta = 1 - d \xi$ in Corollary \ref{index},
\begin{equation} \label{eq:asaucr}
D(\rho \| \tilde{\E}_\S(\rho)) + D(\rho \| \tilde{\E}_\T(\rho)) \geq \frac{\beta_{d, \epsilon}}{\lceil \log_{\zeta} \epsilon \rceil} D(\rho \| \id/d \otimes \rho^B)
\end{equation}
for any $\epsilon < \zeta$, where $\tilde{\E} = \E \otimes \id^B$, and $\tilde{\E}_\S, \tilde{\E}_\T$ are defined analogously. Expanding the relative entropies in terms of von Neumann entropies, this is equivalent to
\[ H(\S \otimes B)_\rho + H(\T \otimes B)_\rho \geq 2 H(\rho) + \frac{\beta_{d, \epsilon}}{\lceil \log_{\zeta} \epsilon \rceil} (
\ln d + H(B)_\rho - H(\rho)) \pl, \]
where $H(\S \otimes B)_\rho, H(\T \otimes B)_\rho$, and $H(B)_\rho$ are defined respectively as the entropies of the outputs of $\E_\S$, $\E_\T$, and $\E$ on $\rho$. When $|B| = 1$, we may use Corollary \ref{cor:approx} and Theorem \ref{asa2} to replace $\beta_{d, \epsilon}$ by
\[ \tilde{\beta} = 1 - \frac{32(d+1) \zeta}{15 + (7 d - 24) \zeta} \]
as long as $\zeta \leq 15/(25 d + 46)$, and $|B| = 1$. When $\zeta$ is not this small, we can still obtain via the final Equation from Corollary \ref{cor:approx} that
\begin{equation} \label{eq:asaucr1}
D(\rho \| \E_\S(\rho)) + D(\rho \| \E_\T(\rho)) \geq \frac{D(\rho \| \id/d)}{2 \lceil \log_{\zeta} (15/(57 d + 88))\rceil} \pl.
\end{equation}
As the two bases approach mutual unbias, $\zeta$ approaches 0, and both forms of the inequality approach the entropic uncertainty relation implied by strong subadditivity (see Petz's version as theorem \ref{thm:petzssa}).

Finally, we consider two bases such that $|\braket{i_S|j_T}|^2 = \delta_{i,j}(1 - \theta) + \theta/d$ for $\theta \in [0, 1]$. This situation may arise, for instance, when taking a partial rotation into a Fourier transform of the original basis. As $\theta \rightarrow 1$, the bases approach mutual unbias, and CSQF approaches the bound given by Petz's subalgebra SSA as in Equation \eqref{eq:relcondexpsubadd}. As $\theta \rightarrow 0$, $\xi$ approaches 0. In this latter regime, $1 - d \xi$ is close to 1, so to achieve a sufficiently small $\epsilon$ for the inequality to be non-trivial, $\log_{1 - d \xi} \epsilon$ must be large. Unlike equation \eqref{eq:bardet}, Theorem \ref{asa2} does not become completely trivial until the bases become the same basis. When $\E_\S = \E_\T$, the intersection conditional expectation ceases to be $\E$, instead becoming $\E_\S = \E_\T$. As noted in Corollary \ref{tightness}, the bases approach a different intersection algebra. SSA for two of the same bases reduces to the trivial statement that $2 D(\rho \| \E_\S(\rho)) \geq D(\rho \| \E_\S(\rho))$. Meanwhile, when $\xi << 1/d$ but is still finite, quasi-factorization still compares to the same intersection algebra. If we take for instance the pure state $\ketbra{0_S}$ as a test density, we see that
\[ D(\ketbra{0_S} \| \E_\S(\ketbra{0_S})) + D(\ketbra{0_S} \| \E_\T(\ketbra{0_S})) = D(\ketbra{0_S} \| \E_\T(\ketbra{0_S})) \rightarrow 0\]
as $\theta \rightarrow 0$. In contrast, $D(\ketbra{0_S} \| \E(\ketbra{0_S})) = D(\ketbra{0_S} \| \id / d) = \log d$ for arbitrarily small but finite $\theta$. Hence we should not expect to find $\alpha$-(C)SQF without $\alpha \rightarrow \infty$ as $\theta \rightarrow 0$.

%% file: thermal.tex
\newcommand{\ad}{\text{ad}}
\subsection{Finite Groups and Transference} \label{sec:groups}
A common form of conditional expectation is
\begin{equation} \label{eq:groupform}
\E_{\G} = \E_{\G *} = \frac{1}{|\G|} \sum_{u \in \pi(\G)} \ad_u \pl,
\end{equation}
where $\G$ is a finite group, $\pi(\G)$ is a unitary representation in some Hilbert space, and $\ad_u(X) := u X u^\dagger$ for any unitary $u$ and matrix $X$. Since this conditional expectation is self-adjoint with respect to the trace, we need not distinguish between $\E_{\G}$ and $\E_{\G *}$. Groups may induce collective channels as considered in \citen{bardet_group_2021}. We build on a simplified version of some ideas from that paper and \citen{gao_fisher_2020}.

Let the map $\Phi_{\vec{p}}$ be given by
\begin{equation} \label{eq:convex}
\Phi_{\vec{p}}(\rho) = \sum_{g \in \G} p_g u(g) \rho u(g)^\dagger
\end{equation}
for any probability vector $\vec{p} \in l_1(\G)$, where $u(g)$ is an element of a particular unitary representation of group $\G$. We may think of $\Phi_{\vec{p}}$ as a quantum channel parameterized by $\vec{p}$. We may also think of $\Phi$ as a map that takes a probability vector $\vec{p}$ as input and outputs a quantum channel that applies a correspondingly weighted convex combination of unitary conjugations. There is an analogous classical channel $\Psi_{\vec{p}} : l_1(\G) \rightarrow l_1(\G)$ given in the left regular representation by
\begin{equation} \label{eq:transfer}
\Psi_{\vec{p}}(\vec{q}) = \sum_{g \in \G} p_g g (\vec{q}) = \sum_{g \in \G} p_g \sum_{h \in \G} q_{g^{-1} h} \mathbf{h} \pl,
\end{equation}
and $\mathbf{h}(f) = \delta_{h,f}$ for any $f \in \G$. In this formulation, $g(\cdot)$ denotes the action of $\G$ promoted to probability distributions in $l_1(\G)$, and $p_g$ denotes the probability of group element $g$ as given by $\vec{p}$. Here $l_1(\G)$ is the 1-normed vector space on group elements.

A key insight for quasi-factorization of finite groups (and graphs as considered in Subsection \ref{sec:graphs}) is that of transference. Let $u(g)$ be a unitary representation of a finite group group $\G$ on $|\G|$-dimensional Hilbert space with basis $\{\ket{g} : g \in \G\}$ given by $u_g \ket{h} = \ket{g h}$. This is a Hilbert space version of the left regular representation of $\G$ on itself. For any pair of probability distributions $\vec{p}, \vec{q} \in l_1(\G)$ and any input density $\rho$,
\begin{equation} \label{eq:transference}
(\Phi_{\vec{p}} \circ \Phi_{\vec{q}}) (\rho) = \sum_{g \in \G} \sum_{h \in \G} p_g q_h u(g) u(h) \rho u(h)^\dagger u(g)^\dagger =  \Phi_{\Psi_{\vec{p}}(\vec{q})}(\rho) \pl.
\end{equation}
Given a sequence of composed quantum channels $\Phi_{\vec{p}^{(1)}} \circ ... \circ \Phi_{\vec{p}^{(k)}}$ for some $k \in \NN$, we can calculate the final unitary weights via $\Psi_{\vec{p}^{(1)}} \circ ... \circ \Psi_{\vec{p}^{(k)}}$. In many cases, the latter will be easier to handle, as it is a composition of classical rather than quantum channels. Furthermore, there are many circumstances in which classical mixing processes on groups or graphs have strong, known results that were previously not known to extend to quantum analogs. The principle of transference was used in \citen{gao_fisher_2020, bardet_group_2021}. Following the relative entropy inequalities established in this paper, we are able to extend the technique to yield tensor-stable relative entropy comparisons.

Let $\vec{1} / |\G|$ denote the classical probability vector weighting each finite group element equally. If
\begin{equation} \label{eq:boundtotrans}
\Psi_{\vec{p}} = (1 - \zeta) \vec{1} / |\G| + \zeta \Psi_{\vec{q}}
\end{equation}
for some probability distribution $\vec{q}$ and $\zeta \in (0,1)$, then
\[ \Phi_{\vec{p}} = (1 - \zeta) \E_\G + \zeta \Phi_{\vec{q}} \pl. \]
The main idea of this Section is to transfer bounds of the form in Equation \eqref{eq:boundtotrans} to those on the corresponding quantum channel $\Phi_{\vec{p}}$. In doing so, we derive Loewner order bounds on quantum representations from classical vector order bounds that are often simpler to calculate. From these order inequalities, the techniques of this paper yield relative entropy inequalities. Also, the convex combination form of equation \eqref{eq:convex} allows us to upper bound the constant $c$ in Theorem \ref{revconv} by $|\G|$.

A nice case of quasi-factorization with transference involves subgroups $\G_1, ..., \G_J \subseteq \G$. If $\cup_{j=1}^J \G_j$ contains generators of the entire group, then we may conclude that at least some chain of conditional expectations from the set $\{\G_j\}_{j=1}^J$ of length $|\G|$ or shorter will include $\E_\G$ in convex combination. (C)SQF follows with good constants depending on the specific structure of the group. The highlighted example appears in the next section, where we consider transference analogs on conditional expectations and semigroups derived from finite graphs.

Though we highlight a particular group representation for clarity, the techniques of this Section apply to other representations. One may for instance consider both left and right regular representations, in which case nearly equivalent results hold. It is not always valid, however, to mix one classical representation with a different quantum representation. It is also not assured that distinct representations will always yield the same optimal constants.

\begin{exam}[Symmetric Group of Degree 3] \label{symmgroup} \normalfont
The group of permutations on 3 indices, known both as the symmetric group of degree 3 ($S_3$) and as the dihedral group of degree 3, is the smallest non-abelian group. It contains 6 elements, which we may represent on 6-dimensional probability space or on 6-dimensional Hilbert space as above. We may describe convex combinations of unitaries from this group by
\[ \Phi_{\vec{p}} = \sum_{j=1}^6 p_j \ad_{u_j} \pl, \]
where $\vec{p} = (p_1, ..., p_6)$ is a probability vector, and each $u_j$ is a unitary conjugation representing a distinct element of $S_3$. When $p_1 = ... = p_6 = 1/6$, the channel becomes a conditional expectation to the invariant subalgebra of the group. Repeated applications of two non-redundant generators of $S_3$ will eventually generate every element of the entire group. Hence a sufficiently long chain of applications of $\Phi_{\vec{p}}$ approaches the fixed point conditional expectation as long as $\vec{p}$ is non-zero on at least two non-redundant generators. Mathematically, such a channel $\Phi_{\vec{p}}$ will have that $\lim_{k \rightarrow \infty} \Phi_{\vec{p}}^k = \E_{S_3}$ as in the form of Equation \eqref{eq:groupform}, where we recall that $\E_{S_3}$ is self-adjoint.

As in Equation \ref{eq:transfer}, we may construct the channel $\Psi_{\vec{p}}(\cdot) : l_1(S_3) \rightarrow l_1(S_3)$, where $l_1(S_3)$ is the 6-dimensional, 1-normed vector space on elements of $S_3$. Via Equation \eqref{eq:transference},
\[\Phi_{\vec{p}}^k = \Phi_{(\Psi_{\vec{p}})^k(\mathbf{1})} \pl,\]
where $\mathbf{1}$ denotes the probability vector corresponding to the group's identity element (not to be confused with $\vec{1}$, which equally weights all elements). Hence we may determine $\zeta$ and $c$ as in Corollary \ref{cor:simpleG} precisely for any $\vec{p}$ by finding the distribution induced by $k$ repeated applications of $\Psi(\vec{p})$ as discrete classical Markov chain to a probability vector. This process may allow us to bypass the potentially harder problem of calculating $\Phi_{\vec{p}}^k$ for arbitrary quantum states. Applying Theorem \ref{asa2}, for sufficiently large $k$ and when $\vec{p}$ has non-zero weight on at least some pair of non-redundant generators,
\[D(\rho \| \Phi_{\vec{p}}^k(\rho)) \geq \beta_{c, \min(\Psi_{\vec{p}}^k(\mathbf{1}))} D(\rho \| \E_{S_3}(\rho))\]
where $\min(\Psi_{\vec{p}}^k(\mathbf{1}))$ is the minimum element of $\Psi_{\vec{p}}^k(\mathbf{1})$. We see via the convex combination form of $\E_{S_3}$ that $c \leq 6$.

Finally, we may consider extensions of the Hilbert space by an arbitrary, finite-dimensional, auxiliary system that is untouched by $\Phi_{\vec{s}}$ for any $\vec{s} \in l_1(S_3)$. The same transference as above nearly holds, but we note the possibility to define a maximally entangled (pure) state on the two-copy space.
\end{exam}

\begin{rem} \label{rem:groupindex}
Let $\G$ be a finite group of cardinality $|\G|$. For any unitary representation, there is a natural conditional expectation given by Equation \eqref{eq:groupform}. Since the representation of the identity element is the identity unitary, $\E_{\G}(\rho)$ ($ = \E_{\G *}$ via self-adjointness) contains $(1/|\G|) \rho$ as a term. For any input density $\rho$, $\rho = |\G| (1/|\G|) \rho \leq |\G| \E_{\G}(\rho)$, so the first Pimsner-Popa index as in Equation \eqref{eq:indices} is upper-bounded by $|\G|$. In general, Corollary \ref{index} yields quasi-factorization for subgroups' conditional expectations with constants depending on the group structure but not otherwise on the Hilbert space dimension of the representation.
\end{rem}
The relative entropy of frameness or asymmetry of a density $\rho$ takes the form $D(\rho \| \E_{\G}(\rho))$, where $\E_{\G}$ is the conditional expectation to the invariant subspace of some group in some representation \cite{vaccaro_tradeoff_2008, gour_measuring_2009, marvian_extending_2014}. In many of these cases, however, the group is compact but not finite. Nonetheless, replacing the sum in Equation \eqref{eq:convex} by an integral, it may still be possible to apply transference and related techniques. See \citen{bardet_group_2021} for examples.

\subsection{Finite, Connected, Undirected Graphs} \label{sec:graphs}
In this section, we will use the symbol $G = (V,E)$ to denote a graph with vertex set $V$ and edge set $E \subseteq V \times V$. We reserve $\G$ for a group. As in the previous Subsection, we will usually work with self-adjoint conditional expectations, for which we need not distinguish between the conditional expectation $\E$ and its predual $\E_*$.

The group and graph scenarios relate closely. An undirected, $n$-vertex graph $G$ will have a naturally corresponding group $\G$ with action on $1...n$ in which each edge $(i, j) \in V \times V$ corresponds to the swap operation $i \leftrightarrow j$. This association is not unique - we may for instance identify a cyclic graph with a one-generator cyclic group, or with a multi-generator group of self-inverse swaps. Conversely, a finite group $\G$ has a corresponding Cayley graph.

For simplicity and to facilitate concrete calculations, we here define a representation. In the computational basis $\{\ket{i} : i \in 1...n\}$, let the unitary representation of undirected edge $(l,j) \in E$ be given on a bipartite system $A \otimes B$ by
\begin{equation} \label{eq:graphu}
u_\theta(l,j) := \Big ( e^{i \theta} \ket{l}\bra{j} + e^{-i \theta} \ket{j}\bra{l} + \sum_{r \neq l,j} \ketbra{r} \Big ) \otimes \id^B \pl.
\end{equation}
Note the extra parameter $\theta$: in addition to swapping the $l$th and $j$th basis states, such a representation may apply a relative phase on the switched elements. Here the graph representation acts on a system $A$ of dimension $n$, which we may extend by an arbitrary, finite-dimensional system $B$ on which the graph acts as the identity. Correspondingly, one may naturally define rank 1 basis vectors $\{\vec{v}_j : j \in 1...n\}$ on the classical probability space in $l_1^n$. We will use the classical representation of an edge $(l,j)$ given by the exchange $\vec{v}_l \leftrightarrow \vec{v}_j$. Under the identification $\vec{v}_j \leftrightarrow \ketbra{j}$ between $l_1^n$ and the diagonal densities on dimension $n$, then $u_\theta(l,j)$ acts as this swap operation regardless of $\theta$. In the literature \cite{diaconis_logarithmic_1996, gao_fisher_2020, li_graph_2020}, it is actually common to restrict to the diagonal or matrix-valued probability space in which $A \cong l_1^n$, while $B$ is a space of densities.

Indeed, the parameter $\theta$ suggests an ambiguity in the representation and complicates the analogy with the classical space. To restore this analogy, we work with channels of the form
\begin{equation} \label{eq:edgechan}
\begin{split}
& \Phi_{l,j}(\rho) = \int_0^{2 \pi} u_\theta(l,j) \rho u_{\theta}(l,j) d \theta \\
	& \pl = \ket{l}\braket{j | \rho | j} \bra{l} + \ket{j} \braket{l | \rho | l} \bra{j}
		+ \Big (\sum_{s \neq l,j} \ketbra{s} \otimes \id^B \Big ) \rho \Big (\sum_{r \neq l,j} \ketbra{r} \otimes \id^B \Big ) \pl.
\end{split}
\end{equation}
The channel $\Phi_{l,j}$ dephases the $i$th and $j$th matrix elements as well as swapping. Hence it is non-unitary but in some ways more closely analogous to the classical swap. Acting on $l_1^n \otimes B$, $\Phi_{l,j}$ is equivalent to $u_\theta(l,j)$ for any $\theta$.

In \citen{li_graph_2020}, it was shown that connected graphs as represented above on $l_1^n \otimes B$ have CMLSI with constant at least $1/O(n^2)$. This result arises by comparing all graphs to the broken cycle, the slowest-decaying of connected graphs. Similarly, \citen{gao_geometric_2021, gao_complete_2021} showed general CMLSI, resolving at least the existence of CMLSI for graphs. In \citen{gao_fisher_2020}, complete graphs including two-vertex, single-edge graphs were shown to have CMLSI with constant of $O(1)$, as the channels given by these graphs are already convex combinations that include $\E_G$. Missing so far have been explicit estimates of the CMLSI constants for expanders and similar graphs, which are often expected to be stronger than $\Omega(1/n^2)$ but worse than $O(1)$.

We define the diagonal projection conditional expectation:
\[ \E_{\text{diag}(j)}(\rho) = \ket{j}\braket{j | \rho^A | j }\bra{j} +  
    \Big (\sum_{s \neq j} \ketbra{s} \otimes \id^B \Big ) \rho \Big (\sum_{r \neq j} \ketbra{r} \otimes \id^B \Big ) \pl. \]
For each edge there is a natural conditional expectation given by
\begin{equation} \label{eq:eedge}
\begin{split}
\E_{l,j}(\rho) = &  \frac{1}{2} \big ( \ket{l}\braket{j | \rho | j} \bra{l} + \ket{j} \braket{l | \rho | l} \bra{j}
	\pl +  \ket{l}\braket{l | \rho | l} \bra{l} + \ket{j} \braket{j | \rho | j} \bra{j} \big )
	\\ & + \Big (\sum_{s \neq l,j} \ketbra{s} \otimes \id^B \Big ) \rho \Big (\sum_{r \neq l,j} \ketbra{r} \otimes \id^B \Big ) \\
& = \frac{1}{2} \E_{\text{diag}(l)} \E_{\text{diag}(j)}(\rho + \Phi_{l,j}(\rho)) \pl.
\end{split}
\end{equation}
We see that $u_\theta(l,j) \E_{l,j}(\rho) u_\theta(l,j) = \E_{l,j}(u_\theta(l,j) \rho u_\theta(l,j))$ for all $\theta$ via the first form in Equation \eqref{eq:eedge}, implying that $\Phi_{l,j} \E_{l,j} = \E_{l,j} \Phi_{l,j} = \E_{l,j}$.  We note that $\E_{l,j} = \E_{l,j} \E_{\text{diag}(j)} = \E_{\text{diag}(j)} \E_{l,j}$, and similarly with $\E_{\text{diag}(l)}$. Also, $[\E_{l,j}, \E_{\text{diag}(k)}] = [\E_{\text{diag}(j)}, \E_{\text{diag}(k)}] = 0$ for $k \neq l,j$. To see that $\E_{l,j}$ is indeed a conditional expectation, we confirm idempotence by calculating
\[ \E_{l,j}^2(\rho) = \frac{1}{4} \E_{\text{diag}(l)} \E_{\text{diag}(j)} (1 + 2 \Phi_{l,j}(\rho) + \Phi_{l,j}^2(\rho)) = \E_{l,j}(\rho) \pl. \]
Furthermore, as $\tr(\ket{l}\braket{j | X | j} \bra{l} Y) = \tr(X \ket{j} \braket{ l | Y | l} \bra{j})$ for all $X,Y$ in the domain of $\E_{l,j}$, and the summed terms in Equation \eqref{eq:eedge} are explicitly self-adjoint, $\E_{l,j}$ is self-adjoint. Once idempotence and self-adjointness are established, we have that $\tr(\E_{l,j}(X) \E_{l,j}(Y)) = \tr(X \E_{l,j}(Y))$ for any $X,Y$ in the domain of $\E_{l,j}$, confirming that it is a conditional expectation. One may also see that the range of $\E_{l,j}$ has the block diagonal form of Equation \eqref{eq:blockcondexp}, explicitly defining the corresponding von Neumann subalgebra.

Let $\E_{\text{diag}} = \prod_{j=1}^n \E_{\text{diag}(j)}$ denote the conditional expectation to the computational basis. There is also a natural conditional expectation to the invariant subspace of $\{ \Phi_{i,j} : (i,j) \in E\}$ given by
\begin{equation} \label{eq:egraph}
\E_G(\rho) = \Big ( \frac{1}{n} \sum_{i, j \in 1...n} \ket{i}\bra{j} \rho \ket{j}\bra{i} \Big )
	= \frac{\id^A}{n} \otimes \tr_{A}(\rho) \pl.
\end{equation}
\begin{figure}[h!] \scriptsize \centering
	\begin{subfigure}[b]{0.22\textwidth}
		\includegraphics[width=0.95\textwidth]{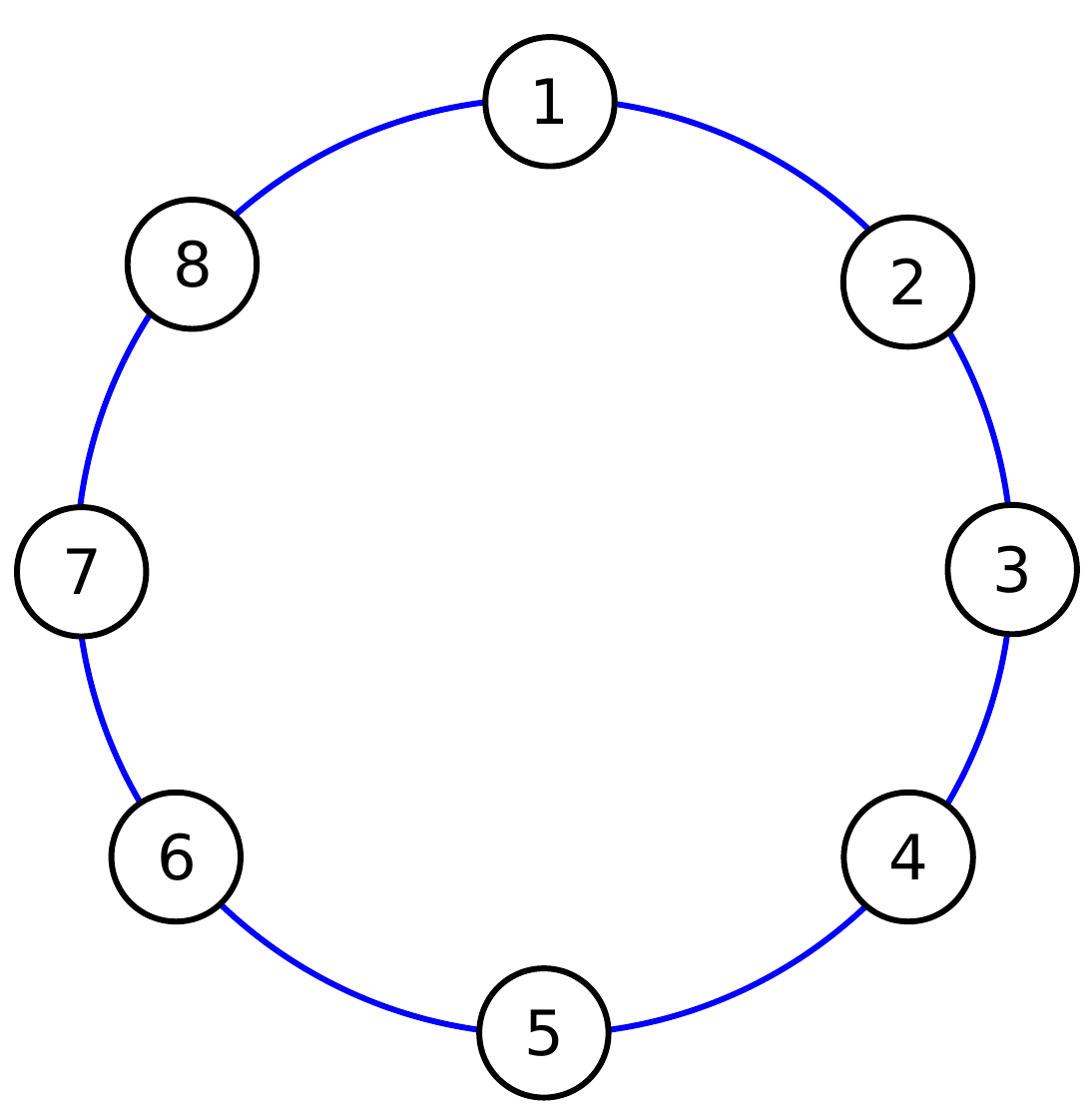}
	\end{subfigure}
	\begin{subfigure}[b]{0.22\textwidth}
		\includegraphics[width=0.95\textwidth]{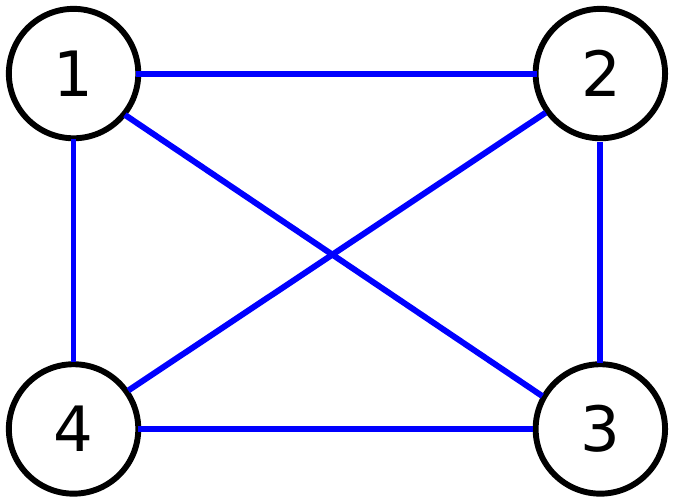}
		\vspace{5mm}
	\end{subfigure}
	\begin{subfigure}[b]{0.22\textwidth}
		\includegraphics[width=0.95\textwidth]{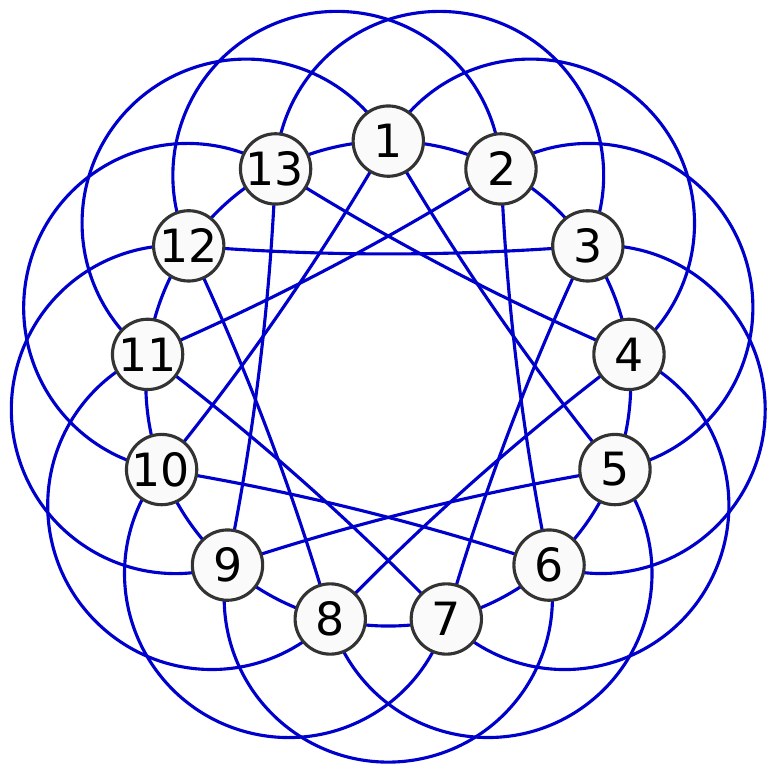}
	\end{subfigure}
	\caption{Visualizations of a cyclic graph, complete graph, and Paley graph (a member of a family that is Ramanujan for sufficiently many vertices) as reproduced from Wikipedia Page ``Paley graph" January 2021, Wikipedia Page Version ID: 1000144544, author David Eppstein (released to the public domain). \label{fig:graphs}}
\end{figure}
Any $k$-fold composition $\Phi_{i_1, j_1} \circ ... \circ \Phi_{i_k, j_k}$ can be written as a sum of ketbra conjugations $\rho \rightarrow \ket{i}\bra{j} \rho \ket{j} \bra{i}$ for $i,j \in 1...n$. Furthermore, as no $\Phi_{i,j}$ escapes the diagonal basis, neither does any product of them, nor any convex combination of such products. Hence any quantum channel $\Phi_{\vec{p}}$ that is a convex combination of composed $\Phi_{i,j}$ following or followed by $\E_{\text{diag}}$ has the form
\begin{equation} \label{eq:graphchanquantum}
\Phi_{\vec{p}}(\rho) = \sum_{i,j = 1}^n p_{i,j} \ket{i}\bra{j} \rho \ket{j}\bra{i}
\end{equation}
for a probability vector $\vec{p} \in l_1(n) \otimes l_1(n)$ that weights transitions between basis states. The form of channel defined here allows a greater range of processes than expressed by symmetric graphs, such as those in which not every transition is balanced by its inverse with equal weight. $\Phi$ is not as general as an arbitrary combination of $\ket{i}\bra{j}$ terms, as properties such as normalization are implicit in its structure. Hence not all values of $\vec{p}$ are valid. We need not explicitly check these constraints as long as the weights arise from a composition of valid physical processes.

We say that an undirected graph $G$ is $m$-regular when each vertex has $m$ incoming/outgoing edges. We denote by $A$ the normalized adjacency matrix of a graph $G$, which for an $m$-regular graph is the adjacency matrix divided by $m$. We recall an important result on classical graphs:
\begin{theorem}[Theorem 3.3 from \citen{hoory_expander_2006}] \label{expanderthm}
Let $G$ be a connected, undirected, $m$-regular graph with $n$ vertices. Let the eigenvalues of $G$'s normalized adjacency matrix $A$ be denoted $1 = \lambda_1 \geq ... \geq \lambda_n$, and $\gamma = \max\{|\lambda_n|,|\lambda_2|\}$. Then for any normalized probability vector $\vec{x} \in l_1^n$ and $t \in \NN$,
\[ \|A^t \vec{x} - \vec{1}/n \|_2 \leq \|\vec{x} - \vec{1}/n \|_2 \gamma^t \leq \gamma^t \pl. \]
\end{theorem}
From this point on, we focus on $m$-regular, undirected graphs. Theorem \ref{expanderthm} is a form of spectral gap condition analogous to that considered in \citen[Lemma 2.6]{gao_complete_2021}. In the graph literature, however, $1 - \lambda_2$ is often referred to as the spectral gap for a graph's Laplacian or normalized adjacency matrix. To avoid confusion, we define $\gamma$ explicitly rather than refer to it as the spectral gap.
\begin{rem} \label{graphconvex}
Let $\Psi : l_1^n \rightarrow l_1^n$ be a classical channel on the probabilities over $n$ elements. Let 
\[ \tilde{\Psi}(\vec{x}) := \frac{1}{n \| \Psi(\vec{x}) - \vec{1}/n \|_\infty} \big ( \Psi(\vec{x}) - (1 - n \| \Psi(\vec{x}) - \vec{1}/n \|_\infty) \vec{1} / n \big ) \pl. \]
Since $\Psi(\vec{x}) \geq \vec{1}/n - \| \Psi(\vec{x}) - \vec{1}/n \|_\infty \vec{1}$ in vector order, $\tilde{\Psi}(\vec{x}) \geq 0$ for all $\vec{x}$. Since $\Psi(\vec{x})$ and $\vec{1}/n$ are 1-normalized, so is $\tilde{\Psi}(\vec{x})$. Finally,
\begin{equation}
\Psi(\vec{x}) = (1 - n \| \Psi(\vec{x}) - \vec{1}/n \|_\infty) \vec{1} / n + n \| \Psi(\vec{x}) - \vec{1}/n \|_\infty \tilde{\Psi}(\vec{x}) \pl,
\end{equation}
suggesting transference of the vector order inequality to a cp-order inequality on quantum channels as in Equation \eqref{eq:boundtotrans} with $n \| \Psi(\vec{x}) - \vec{1}/n \|_\infty$ playing the role of ``$\zeta$" therein. This Remark yields a similar result to Proposition \ref{near}, but the argument is simplified in the classical setting.
\end{rem}
The following Lemma formalizes the notion of transference for a graph:
\begin{lemma} \label{graphtransf}
Let $G$ be a graph with edge set $E$ having cardinality $|E|$. Let $\Phi_{\vec{p}}, \Phi_{\vec{q}}$ be quantum channels of the form in Equation \eqref{eq:graphchanquantum} for probability vectors $\vec{p}, \vec{q} \in l_1(E)$. Let $\Psi_{\vec{p}}, \Psi_{\vec{q}}$ be the respectively corresponding classical channels with the same weighting defined by applying the quantum channel on diagonal densities with the identification $\vec{v}_i \leftrightarrow \ketbra{i}$ for vertex basis vector $\vec{v}_i$. Then for any $\zeta \in (0,1)$,
\begin{equation} \label{eq:graphcompquantum}
\Phi_{\vec{p}}
    \geq_{cp} (1 - \zeta) \Phi_{\vec{q}}
\end{equation}
if and only if
\begin{equation} \label{eq:graphcompclassical}
\Psi_{\vec{p}}(\vec{x}) \geq (1 - \zeta) \Psi_{\vec{q}}(\vec{x})
\end{equation}
for all input vectors $\vec{x}$.
\end{lemma}
\begin{proof}
That Equation \eqref{eq:graphcompquantum} implies Equation \eqref{eq:graphcompclassical} follows immediately from diagonality preservation and the classical-quantum identification.

One may interpret the identification $\vec{v}_i \leftrightarrow \ketbra{i}$ as a pair of mutually inverse maps $\iota, \iota^{-1}$, such that $\Psi_{\vec{p}}(\vec{x}) = \iota^{-1}(\Phi_{\vec{p}}(\iota(\vec{x})))$. Whenever Equation \eqref{eq:graphcompclassical} holds for $\Psi_{\vec{p}}$, Equation \eqref{eq:graphcompquantum} holds on diagonal densities in the vertex basis. What remains to be shown is that it extends to all densities. When $\Phi_{\vec{p}}(\ketbra{i}) \geq (1 - \zeta) \Phi_{\vec{q}}(\ketbra{i})$ for each $i \in 1...n$, $p_{j,i} \geq (1-\zeta) q_{j,i}$ for every $i,j \in 1...n$ in the explicit form given by Equation \eqref{eq:graphchanquantum}. Hence $\vec{p} \geq (1-\zeta) \vec{q}$, and $\Phi_{\vec{p}} \geq_{cp} (1 - \zeta) \Phi_{\vec{q}}$.
\end{proof}

To transfer the result of Theorem \ref{expanderthm} to the setting of \ref{graphconvex}, we compare repeated applications of single edge conditional expectations as in Equation \eqref{eq:eedge} to a random walk. One step of a random walk on a graph $G$ with edge set $E$ is represented by
\begin{equation} \label{eq:phirw}
\Phi_{RW(G)}:= \frac{1}{m} \sum_{(i,j)} \Phi_{i,j}
\end{equation}
The corresponding classical channel $\Psi_{RW(G)}$ via the identification of Lemma \ref{graphtransf} is equivalent to left multiplication by the normalized adjacency matrix of $G$ as in Theorem \ref{expanderthm}.

\begin{lemma} \label{lem:randwalk}
Let $G$ be an $m$-regular graph with $n$ vertices. For any $k \in \NN$, $t \in \NN$ such that $t < k m / 2$, and $\rho$ diagonal in the basis of $G$,
\begin{equation} \label{eq:graphgeneral}
\begin{split}
& k \sum_{(i,j) \in E} D(\rho \| \E_{i,j}(\rho)) \geq
	\big (1 - \exp(- k m n D(\nu / n \| 1 / n ) / 2) \big ) D(\rho \| \Phi_{RW(G)}^{(\geq t)}(\rho)) \\
& \pla	\geq \big (1 - \exp(- k m (1 - \nu(1 + \ln(1/\nu) + 1/n - \nu/n )) / 2) \big ) D(\rho \| \Phi_{RW(G)}^{(\geq t)}(\rho)) \pl,
\end{split}
\end{equation}
 where $\nu = 2 t / k m$, $\Phi_{RW(G)}$ is defined by Equation \eqref{eq:phirw}, and
 \[ \Phi_{RW(G)}^{(\geq t)} = \sum_{s \in t...|E| k} \mu(s) \Phi_{RW(G)}^s \] 
for some probability distribution $\mu$ on the set $t ... |E| k \subset \NN$. Here $D(p \| q)$ is the relative entropy of the two-outcome probability distribution $(p,1-p)$ to $(q,1-q)$ for $p, q \in [0,1]$.
\end{lemma}
\begin{proof}
Let $|E| = m n / 2$ denote the number of edges in $G$. Via Lemma \ref{lem:subadd} and the convexity of relative entropy,
\[ \frac{1}{|E|} \sum_{(i,j) \in E} D(\rho \| \E_{i,j}(\rho)) + D(\rho \| \Theta(\rho))
	\geq D \Big ( \rho \Big \| \frac{1}{|E|} \sum_{(i,j) \in E} \E_{i,j} (\Theta(\rho)) \Big ) \]
for any channel $\Theta$ from the input space to itself. Iterating $|E| k$ times starting with $\Theta = \id$ yields that
\begin{equation} \label{eq:rwform}
k \sum_{(i,j) \in E} D(\rho \| \E_{i,j}(\rho)) \geq D \Big ( \rho \Big \| \Big ( \frac{1}{|E|} \sum_{(i,j) \in E} \E_{i,j} \Big )^{k |E|} (\rho) \Big ) \pl.
\end{equation}

Each edge conditional expectation $\E_{(i,j)}$ has probability $1/2$ to apply $\Phi_{i,j}$. Let $\Phi_G := (1/|E|) \sum_{i,j} \E_{i,j}$. As in Lemma \ref{graphtransf}, we define a corresponding classical channel $\Psi_G^{(k)} : l_1(n) \rightarrow l_1(n)$. For the rest of the proof, we analyze $\Psi_G^{(k)}$. Let $\vec{v}_i$ denote the $i$th element in the canonical basis of $l_1^n$, which has probability fully concentrated at the $i$th vertex. Since $\Psi_G$ applies the conditional expectation corresponding to each edge with equal probability, and each edge touches two nodes, it has a probability of $2 / n$ to apply a conditional expectation that would affect $\vec{v}_i$. If it does so, then with probability 1/2 it applies $\Psi_{i,j}$ for some $(i,j) \in E$, which we may regard as one step in a random walk. Since the action on $\{\ketbra{i}\}$ defines the action on diagonal vectors, $\Psi_G$ applies $\Psi_{RW(G)}$ with probability $1/n$ and identity otherwise.

Let $t$ denote the number of steps in a random walk. Since each application of $\Psi_G$ is equivalent to applying at least one step in a random walk with probability at least $1/n$, we obtain a binomial distribution to take $t$ steps in $k |E|$ trials. By the Chernoff bound (see Theorem 1 in \citen{arratia_tutorial_1989}), for any $\nu \in [0, 1)$,
\[ p(t \leq \nu k |E| / n) \leq \exp(- k |E| D(\nu / n \| 1 / n )) \pl. \]
Since the relative entropy is positive and convex, for any $t \leq \nu k |E| / n = \nu k m / 2 $,
\begin{equation} \label{eq:cher1}
D(\rho \| \Psi_G^{k |E|} (\rho)) \geq (1 - \exp(- k |E| D(\nu / n \| 1 / n ))) D(\rho \| \Psi_{RW(G)}^{( \geq t)}(\rho)) \pl.
\end{equation}
Combining the above with Equation \eqref{eq:rwform} completes the first inequality of the Lemma.

For the 2nd inequality, we expand and estimate the binary relative entropy.
\begin{equation*}
\begin{split}
D(\nu / n \| 1 / n) = \frac{\nu}{n} \ln \nu
	+ \Big (1 - \frac{\nu}{n} \Big ) \ln \Big ( \frac{1 - \nu / n}{1 - 1 / n} \Big ) \pl.
\end{split}
\end{equation*}
By elementary properties of the natural logarithm,
\begin{equation*}
\begin{split}
 \ln \Big ( \frac{1 - \nu / n}{1 - 1 / n} \Big ) & \geq 1 - \frac{1 - 1 / n}{1 - \nu / n}
  = \frac{1}{1 - \nu / n}  \Big ( \frac{1}{n} - \frac{\nu}{n} \Big )
 \geq \frac{1}{n} \Big ( 1 - \nu\Big ) \pl.
\end{split}
\end{equation*}
Hence
\[ D(\nu / n \| 1 / n) \geq \frac{1}{n} (1 - \nu(1 + \ln(1/\nu) + 1/n - \nu/n )) \pl. \]
To complete the Lemma, we substitute in Equation \eqref{eq:cher1}.
\end{proof}
Via transference, we obtain relative entropy bounds on the quantum representation of a random walk on the partly diagonal algebra $l_1^n \otimes S_1^n$:
\begin{lemma} \label{lem:rwbound}
For any $t \in \NN$ and $m$-regular graph $G$ with $\gamma$ as defined in Theorem \ref{expanderthm}, $\E_G$ as in Equation \eqref{eq:egraph}, and $\Phi_{RW(G)}^{(\geq t)}$ as defined in Lemma \ref{lem:randwalk},
\[ (1-n \gamma^t) \E_G \leq_{cp} \Phi_{RW(G)}^{(\geq t)} 
\circ \E_{\text{diag}} \leq_{cp} (1+n \gamma^t) \E_G \pl, \]
and
\[ D(\E_{\text{diag}}(\rho) \| \Phi_{RW(G)}^{(\geq t)}(\E_{\text{diag}}(\rho))) \geq \beta_{2, n \gamma^t} D(\E_{\text{diag}}(\rho) \| \E_G(\rho)) \pl. \]
\end{lemma}
\begin{proof}
Let $\Psi_{RW(G)}$ be the transferred classical analog of $\Phi_{RW(G)}$ from Equation \eqref{eq:phirw}. We apply Theorem \ref{expanderthm} to $\Psi_{RW(G)}$ with $\gamma$ as defined therein, recalling that $\Psi_{RW(G)}$ is equivalent to left multiplication by $G$'s adjacency matrix. Since the 2-norm is an upper bound for the $\infty$-norm,
\begin{equation} \label{eq:graphinf}
\sup_{\vec{x}} \| \Psi_{RW(G)}^{(\geq t)} (\vec{x}) - \vec{1}/n \|_\infty \leq \gamma^t \pl.
\end{equation}
Via Remark \ref{graphconvex},
\begin{equation} \label{eq:atform}
\Psi_{RW(G)}^{(\geq t)}(\vec{x}) = (1 - n \gamma^t) \vec{1} / n + n \gamma^t  \Theta (\vec{x})
\end{equation}
for any input probability vector $\vec{x}$ and some classical channel $\Theta$. By Lemma \ref{graphtransf} and the properties of the cp-order,
\[ \Phi_{RW(G)}^{(\geq t)} \circ \E_{\text{diag}} = (1 - n \gamma^t) \E_G + n \gamma^t \tilde{\Phi} \]
for some quantum channel $\tilde{\Phi}(\rho)$. Furthermore, one may confirm by applying $\E_G$ to $\Phi_{RW(G)}^{(\geq t)}$ that $\E_G \tilde{\Phi} = \tilde{\Phi} \E_G = \E_G$. Hence
\[ D(\E_{\text{diag}}(\rho) \| \Phi_{RW(G)}^{(\geq t)}(\E_{\text{diag}}(\rho))) = D(\rho \| (1 - n \gamma^t) \E_G(\rho) + n \gamma^t \tilde{\Phi}(\rho)) \]
for any input density $\rho$.

A possible next step would be to apply Corollary \ref{cor:simpleG} with $\zeta = n \gamma^t$ and $c = n$. Instead, we obtain a better estimate of the constant $c$. Returning to Equation \eqref{eq:graphinf}, the maximum element of $\Psi_{RW(G)}^{(\geq t)}(\vec{x})$ is at most $\gamma^t + 1/n$ for any $\vec{x}$. Recalling Equation \eqref{eq:atform} and applying Lemma \ref{graphtransf} again, $ \Theta(\vec{x}) \leq 2 \times \vec{1} / n$, so $\tilde{\Phi}(\rho) \leq 2 \E_G(\rho)$. We conclude that $c \leq 2$. Using Corollary \ref{cor:simpleG} with $\zeta = n \gamma^t$ and $c \leq 2$ completes the Lemma.
\end{proof}
Combining results of this Subsection, we obtain the technical version of Theorem \ref{graphthm}:
\begin{theorem} \label{graphthmtech}
Let an $m$-regular graph with $n$ vertices $G$ have $\gamma$ as defined in Theorem \ref{expanderthm}. Then for any $k,t \in \NN$ such that $\max\{\lceil \log_\gamma (1/n) \rceil, 2/m \} \leq t \leq (1 - 1/2^{m-1}) k $,
\begin{equation*}
\begin{split}
 & k \sum_{(i,j) \in E} D(\rho \| \E_{i,j}(\rho))
	\geq \beta_{2, n \gamma^t} \big (1 - \exp(- n m k D( t/k \| 1/n) / 2 ) \big ) D(\rho \| \E_G(\rho)) \pl,
\end{split}
\end{equation*}
where $D(p \| q)$ is defined for $p, q \in [0,1]$ as the relative entropy of the two-outcome probability distribution $(p,1-p)$ to $(q,1-q)$.
Hence the Lindbladian given by
\[ \L_G = \sum_{(i,j) \in E} (\rho - \E_{i,j}(\rho)) \]
has CMLSI with the same constant.
\end{theorem}
\begin{proof}
To use Lemmas involving $\Phi_{\vec{p}}$, we must first reduce the desired results to those on densities that are diagonal on the subsystem on which the graph acts. Invoking the chain rule of relative entropy and data processing inequality,
\begin{equation*}
\begin{split}
D(\rho \| \E_{l,j}(\rho)) & =
    D(\rho \| \E_{\text{diag}(l)}(\rho))
    + D(\E_{\text{diag}(l)}(\rho) \| \E_{l,j}(\rho))
\\ & \geq D(\rho \| \E_{\text{diag(l)}}(\rho))
    + D(\E_{\text{diag}}(\rho) \| \E_{l,j}(\E_{\text{diag}}(\rho))) \pl .
\end{split}
\end{equation*}
Similarly,
\[ D(\rho \| \E_G(\rho)) = D(\rho \| \E_{\text{diag}}(\rho)) + D(\E_{\text{diag}}(\rho) \| \E_G(\rho)) \pl . \]
To estimate $D(\rho \| \E_G(\rho))$ in terms of $\sum_{l,j} D(\rho \| \E_{l,j}(\rho))$, we handle the relative entropy to the diagonal subalgebra, then handle the diagonal case.

To estimate $D(\rho \| \E_{\text{diag}}(\rho))$ in terms of $\sum_l D(\rho \| \E_{\text{diag(l)}}(\rho))$, we use  Lemma \ref{lem:subadd} to combine diagonalizing conditional expetations. Since each vertex appears $m$ times on each side of an edge, but we only use one vertex in each edge, we double the sum over edges to count each vertex exactly $m$ times. Hence
\[ 2 \sum_{l,j} D(\rho \| \E_{\text{diag}(l)}(\rho)) \geq
    m D \Big (\rho \| \prod_l \E_{\text{diag}(l)}(\rho) \Big ) \]
Since $[\E_{\text{diag(l)}}, \E_{\text{diag(j)}}] = 0$ for all $l,j \in 1...n$,
\[ \sum_{l,j} D(\rho \| \E_{\text{diag}(l)}(\rho)) \geq (m/2) D(\rho \| \E_{\text{diag}}(\rho)) \pl. \]
Therefore, for any coefficient $1 \geq \alpha > 0$,
\[ \sum_{l,j} D(\E_{\text{diag}}(\rho) \| \E_{i,j}(\E_{\text{diag}}(\rho)))
    \geq \alpha D(\E_{\text{diag}}(\rho) \| \E_G(\rho))
    \implies
\sum_{l,j} D(\rho \| \E_{i,j}(\rho)) \geq \alpha D(\rho \| \E_G(\rho)) \pl. \]
Without loss of generality we may henceforth assume that $\rho = \E_{\text{diag}}(\rho)$.

Lemma \ref{lem:randwalk} implies that
\begin{equation} \label{eq:graphgeneralthm1}
k \sum_{(i,j) \in E} D(\rho \| \E_{i,j}(\rho)) \geq \big (1 - \exp(- n m k D( t/k \| 1/n) / 2 ) \big ) D(\rho \| \Phi_{RW(G)}^{(\geq t)}(\rho)) \pl.
\end{equation}
Then, Lemma \ref{lem:rwbound} implies (assuming that $\rho = \E_{\text{diag}}(\rho))$) that
\[ k \sum_{(i,j) \in E} D(\rho \| \E_{i,j}(\rho)) \geq
	\beta_{2, n \gamma^t} 
	\big (1 - \exp(- n m k D( t/k \| 1/n) / 2 ) \big ) D(\rho \| \E_G(\rho)) \pl. \]
A re-arrangement of this Equation yields the quasi-factorization part of the Theorem.

To obtain CMLSI, we construct the Lindbladian $\L_{i,j *}(\rho) := \rho - \E_{i,j}(\rho)$. Since $\exp(- t \L_{i,j *}) = (1 - \exp(-t)) \E_{i,j}(\rho) + \exp(-t)\rho$, convexity of relative entropy implies 1-CMLSI for $\L_{i,j}$. We then rewrite $\L_G = \sum_{i,j} \L_{i,j}$
and apply Proposition \ref{prop:mlsimerge} to complete the CMLSI part of the Theorem.
\end{proof}
For the simplified Theorem \ref{graphthm} as in the introduction, we replace the binary entropy by its estimate as in Lemma \ref{lem:randwalk}. To simplify the expression, we choose $t = k / 2$. We then choose $k = 2 \lceil \log_\gamma(1/n) \rceil$. Doing so, we arrive at the estimate that
\begin{equation*}
\begin{split}
& \sum_{(i,j) \in E} D(\rho \| \E_{i,j}(\rho)) \\ & \pl \geq \frac{\beta_{2, 1/n}}{2 \lceil \log_\gamma (1/n) \rceil}
	(1 - \exp(- m \lceil \log_\gamma (1/n) \rceil (1 - (\ln m + 1) / m - O(1/n)))) D(\rho \| \E_G(\rho)) \pl. 
\end{split}
\end{equation*}
One may confirm by calculating its derivative and value at $m=2$ that $(\ln m + 1) / m < 1$, so for sufficiently large $n$, the subtracted exponential asymptotes to zero. This Equation yields the simplified Theorem \ref{graphthm}.

Theorem \ref{graphthmtech} is especially powerful when $\gamma = O(1)$ in $n$, in which case it obtains the expected logarithmic mixing time for Ramanujan and similar graphs. Graphs with $\gamma = O(1)$ are commonly known as fast expanders. These graphs have the fastest mixing times possible with fixed degree.

In contrast to graphs with constant $\gamma$ is the cyclic graph, which is 2-regular and has the slowest mixing time up to constants of any connected, regular graph. In \citen{li_graph_2020}, it was shown that Lindbladians corresponding to cyclic graphs having CMLSI with constant $O(1/n^2)$. Based on classical mixing times, we expect this dependence to be optimal. It is well-known that the eigenvalues of a cycle's normalized adjacency matrix take values $\cos(2 \pi l / n )$ for $l \in 0...n-1$, so $\gamma = \cos(2 \pi / n) \approx 1 - 4 \pi^2 / n^2$. Hence $\log_\gamma(1/n) = O(n^2 \ln n)$, which is larger than the expected $O(n^2)$ time for a random walk to converge on a cyclic graph. We may however obtain better estimates through a more case-specific, fine-grained analysis:
\begin{exam}[Cyclic Graph] \normalfont
Here we will not go through Theorem \ref{expanderthm} but will instead directly calculate bounds on the minimum and maximum probability after $O(n^2)$ steps. For any $a\in \NN$, a random walker on an infinite, one-dimensional lattice has probability of landing $s$ steps away from its original position after $a n^2$ is given by a binomial distribution approximated as
\begin{equation} \label{eq:cyclepeak1}
\binom{a n^2}{a n^2/2+ s} \Big ( \frac{1}{2} \Big )^{a n^2}
	\approx \sqrt{\frac{2}{a \pi}} \Big ( \frac{1}{1 - 4 s^2 / a^2 n^4} \Big )^{a n^2 + \frac{1}{2}}
		\Big ( \frac{1 - 2 s / a n^2}{1 + 2 s / a n^2} \Big )^s \frac{1}{n}
\end{equation}
within multiplicative error $\exp(1/(144 a n^2 + 12 n))$. The approximation follows from Robbins's precise form of Stirling's approximation \cite{robbins_remark_1955}. Any location's probability on the cycle will be at least as large as its corresponding probability on the infinite lattice, because all of the possibilities for the walker to escape the cycle instead return and contribute positively.

For a lower bound on the probability that the walker is at its furthest and therefore least likely point at the cycle, $s = n / 2$, we start with the $a=1$ case. For sufficiently large $n$, we coarsely overestimate this probability as at least $1/(\sqrt{a \pi} n)$ by bounding the $s$-dependent factors in Equation \eqref{eq:cyclepeak1} as greater than $1/\sqrt{2}$. Hence at $a=1$,
\[ \Psi^{n^2}_{RW(G)}(\vec{x}) \geq \frac{1}{\sqrt{\pi}} \frac{\vec{1}}{n} + \Big (1 - \frac{1}{\sqrt{\pi}} \Big )\Psi'(\vec{x}) \ \]
for some classical channel $\Psi'$. Since $\vec{1}/n$ is invariant under the application of random walks, we then have for any $a \in \NN$ that
\[ \Psi^{a n^2}_{RW(G)}(\vec{x}) = \Big (1 - \Big (1-\frac{1}{\sqrt{\pi}} \Big )^a \Big ) \frac{\vec{1}}{n}
	+ \Big (1 - \frac{1}{\sqrt{\pi}} \Big)^a \Psi_a'(\vec{x}) \]
for some classical channel $\Psi_a'$. Hence the minimum location probability is lower bounded as $p_{min}(a) \geq (1 - (1-1/\sqrt{\pi}^a))/n$.

For an upper bound on the largest possible probability, we must account for the chance that the walker loops around and returns to its original position. We must count the contributions of $a n^2/2 + r n$ for many values of $r \in \mathbb{Z}$. The $r=0$ term is bounded by $\sqrt{2 / \pi a} / n$ times a factor that can be made arbitrarily close to 1 for large enough $n$, so we may use the simplified overestimate $2 / (n \sqrt{ \pi a})$. On an infinite line, the walker's probability of landing on a location $+1$ steps away from its original location is less than that of landing $n$ steps away, which is upper bounded by $1/2n$. Since left and right paths lead to the same location on the cycle, the probability of $r=1$ is upper bounded by $1/n$. To simplify this problem, we use Hoeffding's inequality \cite{hoeffding_probability_1963}, obtaining
\[ p(|s| \geq |r| n) \leq 2 \exp(- 2 a n^2 ({}^-_+ r/a n)^2 ) = 2 \exp(- 2 r^2 / a) \pl, \]
where $p(|s| \geq |r| n)$ is the probability of taking a net $s$ steps in either direction. The total probability of being at any location declines with distance from the origin, and there are $n$ total positions. The probability to get past $r=1$ is in total less than $2 \exp(-2/a)$, bounding the total probability for $r$ to remain between 2 and 3. Hence the least likely $s \in [1 n, 2 n)$ has probability at most $2 \exp(-2/a) / n$, which upper bounds the probability of $r = {}^+_- 2$ by $2 \exp(-2/a) / n$. Iterating, the total probability of $r \geq 2$ is bounded by the series $\sum_{l=1}^\infty \exp(-l^2/a) / n$. This series is not easily expressed in closed form via elementary functions, but it is upper bounded by the geometric series $\sum_l \exp(- l / a) = 1/(e^{1/a} - 1)$. Since $e^{1/a} \geq 1 + 1/a$, the series's total is upper bounded by $a$, yielding a contribution of $a / n$. Adding the $r=0$, $r=1$, and $r > 1$ cases, we find a total of $p_{max}(a) \leq 2 / (n \sqrt{\pi a}) + 1/n + a/n$.

Using the transference technique from Lemma \ref{graphtransf} and a simple calculation,
\[ \Phi_{RW(G)}(\vec{x}) \geq (1 -\zeta) \E_G + \zeta \Theta (\vec{x}) \]
for $\zeta = 1 - n p_{min}(a)$, and $\Theta \leq_{cp} n(p_{min} + p_{max}/(1 - n p_{min})) \E_{G}$. For large enough $a$ that is constant in $n$ and sufficiently large $n$, a non-trivial bound follows from \ref{cor:simpleG}. Using Lemma \ref{lem:randwalk},
\[ b n^2 \sum_{i,j} D(\rho \| \E_{i,j}(\rho)) \geq \kappa D(\rho \| \E_G(\rho)) \]
for a large enough constant $b$ with $\kappa$ lower bounded above zero independently from $n$. Hence the cycle has quasi-factorization constant $O(n^2)$. Via Theorem \ref{asa2}, the corresponding Lindbladian has $O(1/n^2)$-CMLSI. This Example shows that for the cycle, quasi-factorization and multiplicative relative entropy comparison suffice to obtain bounds of the expected best asymptotic order.
\end{exam}

Another graph that is very unlike the expander or cycle is the complete graph, which is $n$-regular. CMLSI for complete graphs was studied in \citen{gao_fisher_2020}. For a complete graph,
\begin{equation} \label{eq:completeconvex}
 \frac{1}{|E|} \sum_{(i,j) \in E} \E_{i,j} =
	\frac{1}{2} \Big ( 1 - \frac{1}{|E|} \Big ) \id + \frac{1}{2} \Big (1 + \frac{1}{|E|} \Big ) \E_G \pl.
\end{equation}
Using the convexity of relative entropy, $\sum_{i,j} D(\rho \| \E_{i,j}(\rho)) \geq |E| D(\rho \| \frac{1}{|E|} \sum_{(i,j) \in E} \E_{i,j}(\rho))$. The right hand side of Equation \eqref{eq:completeconvex} is already a convex combination of $\E_G$ with another channel, so we simply use convexity of relative entropy to obtain that
\begin{equation*}
\begin{split}
\sum_{i,j} D(\rho \| \E_{i,j}(\rho)) & \geq
	|E| D\Big (\rho \Big \| \frac{1}{2|E|} ((|E| - 1) \id + (|E| + 1) \E_G)(\rho) \Big ) \pl.
\end{split}
\end{equation*}
With iteration,
\begin{equation*}
\begin{split}
	& \geq O(n) D(\rho \| ((1 - O(2^{-n})) \E + O(2^{-n}) \id) (\rho)) \pl.
\end{split}
\end{equation*}
We find a quasi-factorization constant of $O(1/n)$. That the dependence is less than $o(1)$ in $n$ arises because we have not normalized by the degree - were we to construct such a jump process, the probability of making any jump in an infinitesimal time interval would grow proportionally to $n$. If we normalize with $1/m = 1/n$, this leads to an $O(1)$ CMLSI constant. As the Lindbladian constructed from a complete graph (again up to normalizing factors) already generates a convex combination including $\E_G$, earlier notions and methods \cite{bardet_estimating_2017, bardet_hypercontractivity_2018, gao_fisher_2020} already suffice to show CMLSI for complete graphs.

Finally, we note again that there are many ways to represent a graph or construct its corresponding Lindbladian. For example, \citen{li_graph_2020} uses a single unitary $u$ such that $u^n = \id$ to generate the edge transitions of a cyclic group, corresponding to a cyclic graph. In many of these cases, we can still use Theorem \ref{expanderthm} and the same methods of analysis as in Theorem \ref{graphthmtech} to similar effect, and via the classical Laplacian construction of Remark \ref{graphclassical}, we see that in many cases the results are comparable.

%% file: conclusions.tex
\section{Conclusions and Outlook} \label{sec:conclusions}
The underpinnings of this paper's results combine the entropy-geometry links from \citen{carlen_gradient_2017, gao_fisher_2020} that have continued to \citen{li_graph_2020, gao_geometric_2021}, functional calculus as in \citen{gao_complete_2021}, and iterative, computation-like techniques as detailed in sections \ref{sec:relent}, \ref{sec:asa}, and \ref{sec:graphs}. Using these methods, we derive comparisons between relative entropies of particular relevance to scenarios that combine subalgebra restrictions or decay processes. This combination of techniques may be useful in future studies.

It remains open what the best possible quasi-factorization constants are and whether these are calculable in a simple, closed form expression. In \citen{gao_complete_2021}, a two-sided bound is shown in terms of an $L_2 \rightarrow L_2$ norm difference between a composition of conditional expectations and their intersection, but the bound is not necessarily tight on either side.

A still open question is to what extent quasi-factorization and related inequalities hold for infinite dimensions. Proposition \ref{near} breaks down if the minimum dimension of a conditional expectation $\E$ is infinite. Theorem \ref{revconv} is expected to hold beyond finite dimensions and even in non-tracial von Neumann algebras, but it remains to check that assumptions and cited results hold. Once this is verified, most results of this paper will probably carry through to infinite-dimensional settings.

\section{Acknowledgments}
NL is supported by IBM as a Postdoctoral Scholar at UChicago \& the Chicago Quantum Exchange. NL was previously supported by the Department of Physics at the University of Illinois at Urbana-Champaign.

I thank Marius Junge for early feedback on these results. I also acknowledge interactions with Li Gao as helping to motivate this project and with \'{A}ngela Capel as inspiring some improvements. I thank the anonymous reviewers for careful reading and constructive suggestions.

This article may be downloaded for personal use only. Any other use requires prior permission of the author and AIP Publishing. This article appeared in J. Math. Phys. 63, 122203 (2022) and may be found at \url{https://doi.org/10.1063/5.0053698}.

\section{Author Declarations}
The authors have no conflicts to disclose.

%% file: append.tex
\numberwithin{equation}{section}

\section{Technical Proofs and Calculations} \label{sec:approxrelent}
Note: within this section, we use the letter $n$ rather than $d$ for dimension to avoid possible confusion with the derivative.

\begin{lemma}[Restatement of Lemma \ref{lem:towardident}]
Let $\rho$ be given in a diagonal basis by $(\rho_i)_{i=1}^n$, where $n \geq 2$ is the dimension of the system. Let $a, b \in (0,1)$, $i,j \in 1...n$ such that $\rho_i \geq 1/n \geq \rho_j$, and let $\zeta \in \RR^+$ such that
\[ \zeta \leq a \min \Big \{\frac{1-b}{n + a(1-b) + 1}, \frac{b}{(1 - a b) n + a b + 1} \Big \} \pl. \]
If $\rho_j > 0$, then
\[ \Big ( \frac{\partial}{\partial \rho_i} - \frac{\partial}{\partial \rho_j} \Big ) \tr(\rho (a \ln (n \rho) - \ln ((1 - \zeta) \id + \zeta n \rho))) \geq 0 \pl. \]
If $\rho_j = 0$, then letting $\tilde{\rho} = \rho - \epsilon \hat{i} + \epsilon \hat{j}$,
\[ \tr(\rho (a \ln (n \tilde{\rho}) - \ln ((1 - \zeta) \id + \zeta n \tilde{\rho}))) > 
	\tr(\rho (a \ln (n \rho) - \ln ((1 - \zeta) \id + \zeta n \rho))) \]
for sufficiently small $\epsilon$, where $\hat{i}$ and $\hat{j}$ denote the rank 1 unit densities according to the respective $i$th and $j$th basis vectors in the chosen diagonal basis of $\rho$.
\end{lemma}
\begin{proof}[(Proof of Lemma \ref{lem:towardident})]
We first focus on the first inequality to be shown. Define $\delta \equiv \rho_i - \rho_j \geq 0$. Since $\tr(\ln (n \rho)) = \tr(\ln \rho) + \ln n$, derivatives with respect to $\rho$ ignore the difference between these expressions. We directly calculate,
\begin{equation}
\label{eq:entderiv}
\frac{\partial}{\partial \rho_k} (\rho_k \ln ((1 - \zeta) + \zeta n \rho_k)) = \ln ((1 - \zeta) + \zeta n \rho_k) + \frac{\zeta n \rho_k}{(1 - \zeta) + \zeta n \rho_k} \pl.
\end{equation}
for any $k \in 1...n$. By setting $\zeta = 1$,
\begin{equation*}
\Big ( \frac{\partial}{\partial \rho_i} - \frac{\partial}{\partial \rho_j} \Big ) \tr (\rho \ln \rho) = \ln (\rho_i / \rho_j) = \ln(1 + \delta/\rho_j) \pl.
\end{equation*}
We have assumed that $\rho_i \geq 1/n > 0$. In the limit as $\rho_j \rightarrow 0$, this difference of derivatives diverges toward positive infinity, while the subtracted term in the desired result remains finite. Because
\begin{equation*}
\frac{x}{1 + x} \leq \ln (1 + x) \leq x
\end{equation*}
for all $x \geq 0$, we have for any $b \in [0,1]$ that
\begin{equation*}
\ln (1 + \delta / \rho_j) \geq (1 - b) \ln (1 + \delta / \rho_j) + b \frac{\delta / \rho_j}{1 + \delta / \rho_j} \pl.
\end{equation*}
This will allow us to deal with the two terms in equation \eqref{eq:entderiv} individually.

First, we handle the logarithm term, $\ln ((1 - \zeta) + \zeta n \rho_k)$, by finding $\zeta$ such that
\begin{equation*}
a (1 - b) \ln (1 + \delta / \rho_j) \geq \ln \Big (\frac{(1 - \zeta) + \zeta n \rho_j + \zeta n \delta}{(1 - \zeta) + \zeta n \rho_j} \Big ) \pl.
\end{equation*}
We rewrite the right hand side as
\begin{equation*}
\ln \Big (1 + \frac{\zeta n \delta}{(1 - \zeta) + \zeta n \rho_j} \Big ) \pl.
\end{equation*}
On the left hand side, since $\rho_j \leq 1/n$, $\delta / \rho_j \geq n \delta$. We aim to show that
\begin{equation*}
a (1 - b) \ln (1 + n \delta) \geq \ln \Big (1 + \frac{\zeta}{(1 - \zeta) + \zeta n \rho_j} n \delta \Big ) \pl,
\end{equation*}
which by Lemma \ref{lem:logcomp1} is achieved when
\begin{equation*}
\frac{\zeta}{(1 - \zeta) + \zeta n \rho_j} \leq \frac{a(1-b)}{1 + n \delta} \pl.
\end{equation*}
We know $\zeta/((1 - \zeta) + \zeta n \rho_j) \leq \zeta/(1 - \zeta)$ for any $n \rho_j \geq 0$. Hence any
\begin{equation} \label{eq:zc1}
 \zeta \leq \frac{a(1-b)}{1 + n \delta + a(1-b)}
\end{equation}
is sufficiently small.

Next, we handle the fraction terms by finding $\zeta$ such that
\begin{equation*}
a b \frac{\delta / \rho_j}{1 + \delta / \rho_j} \geq \frac{\zeta n (\rho_j + \delta)}{(1 - \zeta) + \zeta n (\rho_j + \delta)}
	- \frac{\zeta n \rho_j}{(1 - \zeta) + \zeta n \rho_j} \pl.
\end{equation*}
One may observe from the Lemma's assumptions on $\zeta$ that $\zeta < 1$ via the first expression in the minimum. If $\zeta = 0$, then $(1 - \zeta) + \zeta n (\rho_j + \delta) = 1$ with no $\delta$ dependence and is thereby trivially analytic in $\delta$. Assuming $ 0 < \zeta < 1$, and $n > 0$, we may write rewrite $(1 - \zeta) + \zeta n (\rho_j + \delta) = r + s \delta$ for some $r,s > 0$. It is well-known in calculus that $\delta \mapsto 1/(r + s \delta)$ is analytic for $\delta > - r /s$, which includes $\delta > 0$ and any sufficiently small neighborhood around $\delta = 0$. Taylor expanding $\delta \mapsto 1/((1 - \zeta) + \zeta n (\rho_j + \delta))$ around $\delta = 0$,
\begin{equation*}
\begin{split}
& \frac{\zeta n (\rho_j + \delta)}{(1 - \zeta) + \zeta n (\rho_j + \delta)} - \frac{\zeta n \rho_j}{(1 - \zeta) + \zeta n \rho_j} \\
= & \frac{\zeta n (\rho_j + \delta)}{(1 - \zeta) + \zeta n \rho_j} \sum_{k=0}^\infty \Big ( \frac{- \zeta n \delta}{(1 - \zeta) + \zeta n \rho_j} \Big )^k 
	- \frac{\zeta n \rho_j}{(1 - \zeta) + \zeta n \rho_j}  \pl.
\end{split}
\end{equation*}
Canceling the 0th order term,
\begin{equation*}
\begin{split}
... = & \frac{\zeta n \rho_j}{(1 - \zeta) + \zeta n \rho_j} \sum_{k=1}^\infty \Big ( \frac{- \zeta n \delta}{(1 - \zeta) + \zeta n \rho_j} \Big )^k 
	+ \frac{\zeta n \delta}{(1 - \zeta) + \zeta n \rho_j} \sum_{k=0}^\infty \Big ( \frac{- \zeta n \delta}{(1 - \zeta) + \zeta n \rho_j} \Big )^k \\
 = & \frac{\zeta n \rho_j}{(1 - \zeta) + \zeta n \rho_j} \sum_{k=1}^\infty \Big ( \frac{- \zeta n \delta}{(1 - \zeta) + \zeta n \rho_j} \Big )^k 
	- \sum_{k=1}^\infty \Big ( \frac{- \zeta n \delta}{(1 - \zeta) + \zeta n \rho_j} \Big )^k \\
 = & \Big ( \frac{\zeta n \rho_j}{(1 - \zeta) + \zeta n \rho_j} - 1 \Big ) \sum_{k=1}^\infty \Big ( \frac{- \zeta n \delta}{(1 - \zeta) + \zeta n \rho_j} \Big )^k \pl.
\end{split}
\end{equation*}
Now to turn this back into the form of a fraction,
\begin{equation*}
\begin{split}
... = & \frac{- (1 - \zeta)} {(1 - \zeta) + \zeta n \rho_j} \sum_{k=1}^\infty \Big ( \frac{- \zeta n \delta}{(1 - \zeta) + \zeta n \rho_j} \Big )^k \\
= & \frac{(1 - \zeta) \zeta n \delta} {((1 - \zeta) + \zeta n \rho_j)^2} \sum_{k=0}^\infty \Big ( \frac{- \zeta n \delta}{(1 - \zeta) + \zeta n \rho_j} \Big )^k \\
= & \frac{(1 - \zeta) \zeta n \delta} {((1 - \zeta) + \zeta n \rho_j)^2} \frac{1}{1 + \frac{\zeta n \delta}{(1 - \zeta) + \zeta n \rho_j}} \\
= & \frac{1} {(1 - \zeta) + \zeta n \rho_j} \frac{\zeta n \delta}{1 + \frac{\zeta}{1 - \zeta} (\rho_j + \delta) n} \pl.
\end{split}
\end{equation*}
Since $0 \leq \rho_j \leq 1/n$,
\begin{equation*}
\begin{split}
... \leq \frac{\frac{\zeta}{1 - \zeta} n \delta}{1 + \frac{\zeta}{1 - \zeta} n \delta} \pl.
\end{split}
\end{equation*}
We must find a $\zeta$ for which
\begin{equation*}
a b \ln (1 + \delta / \rho_j) \geq a b \frac{\delta / \rho_j}{1 + \delta/\rho_j} \geq a b \frac{n \delta}{1 + n \delta} \geq \frac{\frac{\zeta}{1 - \zeta} n \delta}{1 + \frac{\zeta}{1 - \zeta} n \delta} \p.
\end{equation*}
One can easily check that this is satisfied by
\begin{equation*}
\frac{\zeta}{1 + \zeta} \leq \frac{a b}{(1 - a b)n \delta + 1} \pl,
\end{equation*}
which follows from
\begin{equation} \label{eq:zc2}
\zeta \leq \frac{a b}{(1 - a b) n \delta + a b + 1} \pl.
\end{equation}
Finally, any
\[ \zeta \leq \min \Big \{\frac{a(1-b)}{1 + n + a(1-b)}, \frac{a b}{(1 - a b) n + a b + 1} \Big \} \]
satisfies both inequalities \eqref{eq:zc1} and \eqref{eq:zc2} for all $\delta \in [0,1]$.

For the latter inequality, we restrict attention to the $i$th and $j$th basis vectors. The left hand side of the desired inequality becomes equivalent to
\[ (\rho_i - \epsilon) (a \ln (n (\rho_i - \epsilon) )- \ln((1-\zeta) + \zeta n (\rho_i - \epsilon)))
	+ \epsilon (a \ln (n \epsilon) - \ln((1-\zeta) + \zeta n \epsilon)) \pl. \]
The right hand side becomes
\[ \rho_i (a \ln (n (\rho_i) )- \ln((1-\zeta) + \zeta n \rho_i) \pl. \]
Using that $(x-1)/x \leq \ln x \leq x - 1$ for any $x > 0$, we may estimate that to leading polynomial orders in $\epsilon$ that the difference is $O(\epsilon \ln \epsilon) - O(\epsilon)$. Hence for small enough $\epsilon$, it is positive.
\end{proof}

\begin{cor}[Restatement of \ref{cor:approx}]
Given $a \in [0,1]$ and two densities $\rho,\sigma$ in dimension $d$ such that $\rho \succ \sigma$ ($\rho$ majorizes $\sigma$), 
\begin{equation*}
D(\rho \| (1 - \zeta) \id/d + \zeta \sigma) - (1-\frac{32(d+1) \zeta}{15 + (7 d - 24) \zeta}) D(\rho \| \id/d) \geq 0
\end{equation*}
is achieved whenever
\begin{equation*}
\zeta < \frac{15}{25 d + 56} \pl.
\end{equation*}
With $\zeta \rightarrow 0$ as $d$ is held fixed, we see that this expression is asymptotically tight. We may choose $a = 1 - 1/d$ and
\begin{equation*}
\zeta \leq \frac{15 d - 1}{32d^2 - 7 d^2 + 7 d + 32 d + 24 d - 24} = \frac{15 d - 1}{25 d^2 + 63 d - 24}
\end{equation*}
or $a = 1/2$ and
\begin{equation*}
\zeta \leq \frac{15}{57d + 88} \pl.
\end{equation*}
\end{cor}
\begin{proof}[Proof of \ref{cor:approx}]
Given that
\[\zeta \leq a \min \Big \{\frac{1-b}{n + a(1-b) + 1}, \frac{b}{(1 - a b) n + a b + 1} \Big \} \pl, \]
we seek
\begin{enumerate}
	\item A value of $b$ at which both expressions are equal, such that the minimum is maximized;
	\item The maximum possible constraint on $\zeta$;
	\item For a given $\zeta$, the corresponding best achievable $a$.
	\item A formula for $\zeta$ given a reasonable $a$.
\end{enumerate}
The optimal solution may use numerics. Here we derive an approximation.

First, we solve for an equalizing $b$,
\begin{equation*}
\begin{split}
& (1 - b) (n - a b n + a b + 1) = b (n + a (1 - b) + 1) \\ 
& n - a n b + 1 - b n + a n b^2 - b = b n + b \\
& a n b^2 - (a n + 2 n + 2) b + (n+1) = 0 \pl.
\end{split}
\end{equation*}
By the quadratic formula,
\begin{equation*}
b = \frac{1}{2 a n} \big ( a n + 2 n + 2 {}^+_- \sqrt{(a n + 2 n + 2)^2 - 4 a n (n + 1)} \big ) \pl.
\end{equation*}
Since $a n + 2 n + 2 > 1$, and $b \in [0,1]$, only the ``$-$" root is relevant. Examining the expression in the square root,
\begin{equation*}
\begin{split}
& (a n + 2 n + 2)^2 - 4 a n (n + 1) \\
= & (2 + a)^2 n^2 + 4 (a + 2) n + 4 - 4 a n^2 - 4 a n \\
= & 4 n^2 + 4 a n^2  + a^2 n^2 + 4 a n + 8 n + 4 - 4 a n^2 - 4 a n \\
= & 4 n^2 + a^2 n^2 + 8 n + 4 \\
= & 4(n+1)^2 + a^2 n^2 \pl.
\end{split}
\end{equation*}
Via the binomial approximation, which is a reasonable if rough estimate when $a < 4$,
\[ \sqrt{4(n+1)^2 + a^2 n^2} \approx 2(n+1) \Big (1 + \frac{a^2 n^2}{8 (n+1)^2} \Big ) \pl. \]
We then calculate
\begin{equation*}
\begin{split}
b \approx \frac{1}{2 a n} \Big (a n - \frac{2 (n+1) a^2 n^2}{8(n+1)^2} \Big ) = \frac{1}{2} \Big (1 - \frac{a n}{4(n+1)} \Big )
	\approx \frac{1}{2} \Big (1 - \frac{a}{4} \Big )  \pl.
\end{split}
\end{equation*} \small
Plugging this $b$ approximation back into the original minimization,
\[\zeta \leq a \min \Big \{\frac{4 + a}{8 n + 4 a + a^2 + 8}, \frac{4-a}{(8 - 4 a + a^2) n + 4 a - a^2 + 8} \Big \} \pl, \]
We may find a common numerator by multiplying the left by $4 - a$ and the right by $4 + a$. The left denominator becomes
\[ 32 n + 16 a + 4 a^2 + 32 - 8 n a - 4 a ^2 - a^3 - 8 a = 32 n - 8 n a + 32 + 8 a - a^3 \pl. \]
The right becomes
\begin{equation*}
\begin{split}
& 32 n - 16 n a + 4 n a^2 + 16 a - 4 a ^2 + 32
	+ 8 n a - 4 n a^2 + n a^3 + 4 a^2 - a^3 + 8 a \\
= & 32 n - 8 n a + n a^3 + 32 + 24 a - a^3 \pl.
\end{split}
\end{equation*}
To eliminate nonlinear terms in $a$, we approximate $0 \leq a^3 \leq a^2 \leq a \leq 1$, and overestimate denominators. We conclude that
\begin{equation}
\zeta \leq \zeta_+(a) := \frac{15 a}{32 n - 7 n a + 32 + 24 a} \pl.
\end{equation}
We directly calculate the derivative with respect to $a$,
\begin{equation*}
\frac{d}{da} \zeta_+(a) = \frac{15}{32 n - 7 n a + 32 + 24 a} + \frac{15 a (7 n - 24)}{(32 n - 7 n a + 32 + 24 a)^2} \geq 0
\end{equation*}
to see that $\zeta_+(a)$ is increasing with $a$ on the interval $[0,1]$. Hence by comparing to $a=1$, it is sufficient to find
\begin{equation*}
\zeta \leq \zeta_+(a) \leq \frac{15}{25 n + 56} \pl.
\end{equation*}
When $\zeta$ satisfies this condition, we may find an attainable $a$ by solving a linear equation, yielding
\begin{equation*}
a = \frac{32(n+1) \zeta}{15 + (7 n - 24) \zeta} \pl.
\end{equation*}
For a non-trivial inequality, we must take the inequality on $\zeta$ to be strict, otherwise we solve for $a=1$ and the right hand side becomes 0.

Fixing $a = 1 - 1/n$, we approximate
\begin{equation*}
\begin{split}
\zeta \leq \frac{15 n - 1}{32n^2 - 7 n^2 + 7 n + 32 n + 24 n - 24} = \frac{15 n - 1}{25 n^2 + 63 n - 24}  \pl.
\end{split}
\end{equation*}
This approximation is expected to be good for large $n$. Alternatively, if $a = 1/2$,
\begin{equation*}
\zeta \leq \frac{15}{57n + 88} \pl.
\end{equation*}
\end{proof}
For asymptotically large $n$, our estimated maximal $\zeta$ approaches $3/(5n)$.

Here we show a proof of Proposition \ref{prop:mlsimerge} that does not rely on Fisher information. While the proof is more involved, this version gives some intuition that might be useful in future work. To facilitate this proof, we recall a continuity bound for subalgebra-relative entropy and its weighted generalizations. Let
\begin{equation} \label{eq:wsb}
\text{wsb}(\epsilon, \N \subseteq \M, \sigma) := 2 \epsilon \kappa + (1 + 2 \epsilon) h \Big ( \frac{2 \epsilon}{1 + 2 \epsilon}\Big ) \pl,
\end{equation}
where $h$ is the binary entropy, $\N \subset \M$ is a von Neumann subalgebra, and
\begin{equation} \label{eq:indexboundlast}
\kappa := \sup_{\rho, \omega} D(\rho \| \E_{\N, \sigma *}(\rho)) - D(\omega \| \E_{\N, \sigma *}(\omega)) \leq \sup_\rho D(\rho \| \E_{\N, \sigma *}(\rho)) \pl.
\end{equation}
When $\sigma = \id / d$ in dimension $d$, $\kappa = D(\M \| \N) = \sup_\rho D(\E_\M(\rho) \| \E_\N(\rho))$ as described in \citen{gao_relative_2020} and applied to estimate $\text{wsb}$ as Proposition 3.7 therein. In general, Lemma 7 in \citen{winter_tight_2016} implies that
\begin{equation} \label{eq:wsbbound}
|D(\rho \| \E_{\N, \sigma *}(\rho) ) - D(\omega \| \E_{\N, \sigma *}(\omega))| \leq \text{wsb}(\| \rho - \omega\|_1, \N \subseteq \M, \sigma) \pl.
\end{equation} As long as $\sigma$ is faithful, $\kappa$ is finite.
\begin{prop}[Restatement of \ref{prop:mlsimerge}]
Let $\{\Phi^t_j : j \in 1...J \in \NN\}$ be self-adjoint quantum Markov semigroups such that $\Phi_j^t = \exp(- \L_j t)$ with fixed point conditional expectation $\E_{j *} = \lim_{t \rightarrow \infty} \Phi_j^t$ for each $j$ weighted respectively by $(\sigma_j)$. Let $\E_{\sigma *}$ be the weighted intersection fixed point conditional expectation, assuming $\E_{j *}$ are compatibly weighted so that it exists. Let $\Phi^t$ be the semigroup generated by $\L = \sum_j \alpha_j \L_j + \L_0$, where $\L_0$ generates $\Phi_0^t$ such that $\Phi_0^t \E_{\sigma *} = \E_{\sigma *} \Phi_0^t = \E_{\sigma *}$. If $\{\Phi_j^t\}$ has $\{ \alpha_j \}$-(C)SQF, and $\Phi^t_j$ has $\lambda$-(C)MLSI for each $j$, then $\Phi^t$ has $\lambda$-(C)MLSI.
\end{prop}
Though this proposition is stated for particular weightings of the Lindbladians that effectively cancel the contributions of potentially different quasi-factorization constants, one can always adjust these weightings and correspondingly adjust the (C)MLSI constants. One may for instance obtain (C)MLSI for given weightings this way or optimize the weightings to effectively reweight the given (C)MLSI constants.
\begin{proof}
First, assume $\L_0 = 0$. Let $t \in \RR^+$,  $\tau \in (0, t)$, and $\hat{\rho} := \Phi^{t - \tau}(\rho)$ such that $\Phi^t(\rho) = \Phi^\tau(\hat{\rho})$. Let $\N$ be the fixed point algebra of $\Phi^t$ and $\N_j$ the fixed point algebra of $\Phi^t_j$. Let $\M$ be the original algebra containing $\rho$.

By the Suzuki-Trotter expansion $\|\Phi^\tau(\tilde{\rho}) - \Phi^{\alpha_1 \tau}_1 ... \Phi^{\alpha_J \tau}_J (\tilde{\rho})\|_1 \leq O(\tau^2)$. By Equation \eqref{eq:wsbbound},
\begin{equation} \label{eq:trotter}
|D(\Phi^\tau(\hat{\rho}) \| \E_{\sigma *}(\hat{\rho})) - D(\Phi^{\alpha_j \tau}_j ... \Phi^{\alpha_J \tau}_J (\hat{\rho}) \| \E_{\sigma *}(\hat{\rho}))| \leq \text{wsb}(O((J-j) \tau^2), \N \subseteq \M, \sigma).
\end{equation}
for any $j \in 1...J$. Let
\begin{equation} \label{eq:gammaassump}
\gamma_j : = \frac{D(\Phi^{\alpha_{j+1} \tau}_{j+1} ... \Phi^{\alpha_J \tau}_J(\hat{\rho}) \| \E_{j *}(\Phi^{\alpha_{j+1} \tau}_{j+1}
	... \Phi^{\alpha_J \tau}_J(\hat{\rho})))}{D(\Phi^\tau(\hat{\rho}) \| \E_{\sigma *}(\hat{\rho}))}
\end{equation}
when $\hat{\rho} \neq \E_{\sigma *}(\hat{\rho})$. Since both entropies in the ratio are subalgebra-relative, they are both finite. The denominator is zero only if $\Phi^\tau(\hat{\rho})$ is in the intersection subalgebra, which for sufficiently small $\tau$ is true only if $\hat{\rho}$ was in the subalgebra. If $\hat{\rho}$ is in the intersection subalgebra, then it is also in larger subalgebras, so $D(\Phi^{\alpha_{j+1} \tau}_{j+1} ... \Phi^{\alpha_J \tau}_J(\hat{\rho}) \| \E_{j *}(\Phi^{\alpha_{j+1} \tau}_{j+1} ... \Phi^{\alpha_J \tau}_J(\hat{\rho}))) = 0$, and the desired exponential decay holds trivially. Otherwise,
\begin{equation*}
\begin{split}
& D(\Phi^{\alpha_{j} \tau}_{j} ... \Phi^{\alpha_J \tau}_J(\hat{\rho}) \| \E_{\sigma *}(\hat{\rho})) \\
& = D(\Phi^{\alpha_{j} \tau}_{j} ... \Phi^{\alpha_J \tau}_J(\hat{\rho})
		\| \E_{j *} (\Phi^{\alpha_{j} \tau}_{j} ... \Phi^{\alpha_J \tau}_J(\hat{\rho}))) 
	+ D(\E_{j *} (\Phi^{\alpha_{j} \tau}_{j} ... \Phi^{\alpha_J \tau}_J(\hat{\rho}))) \| \E_{\sigma *}(\hat{\rho})) \\
& = D(\Phi^{\alpha_{j} \tau}_{j} ... \Phi^{\alpha_J \tau}_J(\hat{\rho})
		\| \E_{j *} (\Phi^{\alpha_{j} \tau}_{j} ... \Phi^{\alpha_J \tau}_J(\hat{\rho}))) 
	+ D(\E_{j *}(\Phi^{\alpha_{j+1} \tau}_{j+1} ... \Phi^{\alpha_J \tau}_J(\hat{\rho})) \| \E_{\sigma *}(\hat{\rho})) \\
& \leq (1 - \alpha_j \lambda \tau + O(\tau^2)) D(\Phi^{\alpha_{j+1} \tau}_{j+1} ... \Phi^{\alpha_J \tau}_J(\hat{\rho}) \| \E_{j *} (\Phi^{\alpha_{j+1} \tau}_{j+1} ... \Phi^{\alpha_J \tau}_J(\hat{\rho}))) \\ & \pla
	+ D(\E_{j *}(\Phi^{\alpha_{j+1} \tau}_{j+1} ... \Phi^{\alpha_J \tau}_J(\hat{\rho})) \| \E_{\sigma *}(\hat{\rho})) \pl ,
\end{split}
\end{equation*}
where the first equality follows from Lemma \ref{lem:chainexp}, the second equality from $\E_{j *}$ projecting to a fixed point of and thereby absorbing $\Phi_j^{\alpha_j \tau}$, and the first inequality from assumed (C)MLSI of $\Phi_j^{\alpha_j \tau}$. As relative entropies to conditional expectations, both of the final terms are bounded, so we may extract the $O(\tau^2)$ to an additive correction. Next, we recall Equation \eqref{eq:gammaassump} to substitute the subtracted term, yielding
\begin{equation*}
\begin{split}
... \leq & D(\E_{j *}(\Phi^{\alpha_{j+1} \tau}_{j+1} ... \Phi^{\alpha_J \tau}_J(\hat{\rho})) \| \E_{\sigma *}(\hat{\rho}))
	+ D(\Phi^{\alpha_{j+1} \tau}_{j+1} ... \Phi^{\alpha_J \tau}_J(\hat{\rho}) \| \E_{j, \sigma *} (\Phi^{\alpha_{j+1} \tau}_{j+1} ... \Phi^{\alpha_J \tau}_J(\hat{\rho})))
	\\ & - \alpha_j \lambda \tau \gamma_j D(\Phi^\tau(\hat{\rho}) \| \E_{\sigma *}(\hat{\rho})) + O(\tau^2) \pl.
\end{split}
\end{equation*}
Via the chain rule of relative entropy (Lemma \ref{lem:chainexp}),
\begin{equation} \label{eq:jthstep}
 ... = D(\Phi^{\alpha_{j+1} \tau}_{j+1} ... \Phi^{\alpha_J \tau}_J(\hat{\rho}) \| \E_{\sigma *}(\hat{\rho}))
	- \alpha_j \lambda \tau \gamma_j D(\Phi^\tau(\hat{\rho}) \| \E_{\sigma *}(\hat{\rho})) + O(\tau^2) \pl.
\end{equation}
Note that $\kappa$ as in Equation \eqref{eq:indexboundlast} is finite for each $\N_j \subseteq \M$, and $J$ is assumed finite. Iterating Equation \eqref{eq:jthstep},
\begin{equation} \label{eq:iterated}
\begin{split}
& D(\Phi^{\alpha_{1} \tau}_{1} ... \Phi^{\alpha_J \tau}_J(\hat{\rho}) \| \E_{\sigma *}(\hat{\rho})) \leq
\bigg (1 - \lambda \tau \sum_{j=1}^J \alpha_j \gamma_j  \bigg ) D(\hat{\rho} \| \E_{\sigma *}(\rho) ) + O(\tau^2) \pl.
\end{split}
\end{equation}

Now we must lower bound $\sum_j \alpha_j \gamma_j$. Via Equations \eqref{eq:wsbbound} and \eqref{eq:trotter},
\begin{equation*}
\begin{split}
& D(\Phi^{\alpha_{j+1} \tau}_{j+1} ... \Phi^{\alpha_J \tau}_J(\hat{\rho}) \| \E_{j *}(\Phi^{\alpha_{j+1} \tau}_{j+1} ... \Phi^{\alpha_J \tau}_J(\hat{\rho}))) 
 \\ & \pla \geq D(\hat{\rho} \| \E_{j *}(\hat{\rho})) - \text{wsb}((J-j + 1) O(\tau), \N \subseteq \M, \sigma) \pl.
\end{split}
\end{equation*}
Hence via assumed $\{ \alpha_j \}$-(C)SQF and Equation \eqref{eq:gammaassump},
\begin{equation*}
\begin{split}
&  \sum_j \alpha_j \gamma_j D(\hat{\rho} \| \E_{\sigma *} (\hat{\rho})) = \sum_j \alpha_j D(\Phi^{\alpha_{j+1} \tau}_{j+1} ... \Phi^{\alpha_J \tau}_J(\hat{\rho}) \| \E_{j *} \Phi^{\alpha_{j+1} \tau}_{j+1} ... \Phi^{\alpha_J \tau}_J(\hat{\rho})) 
 \\ & \geq \sum_j \alpha_j (D(\hat{\rho} \| \E_{j *}(\hat{\rho})) - \text{wsb}((J-j + 1) O(\tau), \N \subseteq \M, \sigma))
\\ & \geq D(\hat{\rho} \| \E_{\sigma *} (\hat{\rho}))
		- \sum_{j=1}^J \alpha_j \text{wsb}((J-j + 1) O(\tau), \N \subseteq \M, \sigma) \pl.
\end{split}
\end{equation*}
By data processing, $D(\hat{\rho} \| \E_{\sigma *}(\hat{\rho})) \geq D(\Phi^t(\hat{\rho}) \| \E_{\sigma *}(\hat{\rho}))$, so either the right hand side is zero, or the left hand side is greater than zero. If $D(\Phi^t(\hat{\rho}) \| \E_{\sigma *}(\hat{\rho})) = 0$, then the Theorem is trivially complete. Otherwise, returning to equations \eqref{eq:iterated} and \eqref{eq:trotter},
\begin{equation} \label{eq:taulsi}
\begin{split}
& D(\Phi^\tau(\hat{\rho}) \| \E_{\sigma *}(\hat{\rho}))
	\leq \Big (1 - \tau \lambda  + O(\tau^2 \ln \tau) \Big ) D(\hat{\rho} \| \E_{\sigma *}(\hat{\rho})) + O(\tau^2).
\end{split}
\end{equation}
Since $\Phi^t = (\Phi^\tau)^{t / \tau}$,
\begin{equation}
D(\Phi^t (\rho) \| \E_{\sigma *}(\rho)) \leq (1 - \lambda \tau + O(\tau^2 \ln \tau))^{t/\tau} D(\rho \| \E_{\sigma *}(\rho)) + O(\tau) \pl.
\end{equation}
Finally, we observe that
\begin{equation*}
\lim_{\tau \rightarrow 0} ( 1 - \lambda \tau + O(\tau^2 \ln \tau) )^{t/\tau} = e^{- \lambda t} \pl.
\end{equation*}
This limit follows from the fact that any increase in the linear order term in $\tau$ fully absorbs the superlinear order term. Hence this absorbing increase is taken to zero in the limit. This limit completes the proof when $\L_0 = 0$.

Finally, we show that adding $\L_0$ does not reduce the (C)MLSI constant. Let $\tilde{\L} = \sum_{j=1}^J \alpha_j \L_j$. Here we assume that $\tilde{\L}$ has $\lambda$-(C)MLSI. Via the Suzuki-Trotter formula,
\[ \| \Phi^t - ( \exp(- \L_0 t / k) \exp(- \tilde{\L} t / k) )^k \|_1 \leq O(1/k) \]
for any $k \in \NN$. Hence via Equations \eqref{eq:wsbbound} and \eqref{eq:wsb},
\[ | D(\Phi^t(\rho) \| \E_{\sigma *}(\rho)) - D(( \exp(- \L_0 t / k) \exp(- \tilde{\L} t / k) )^k \| \E_{\sigma *}(\rho)) | \leq O(\ln(1/k) / k) \pl. \]
The $O(\ln(1/k) / k)$ dependence comes from the binary entropy term in Equation \eqref{eq:wsb}. Since we will take $k \rightarrow \infty$, we henceforth work with the Trotter expanded term. For any $l \in 1...k$,
\begin{equation*}
\begin{split}
& D(( \exp(- \L_0 t / k) \exp(- \tilde{\L} t / k) )^l (\rho) \| \E_{\sigma *}(\rho))
	\\ & \pla \leq  D( \exp(- \tilde{\L} t / k) (\exp(- \L_0 t / k) \exp(- \tilde{\L} t / k) )^{l-1} (\rho) \| \E_{\sigma *}(\rho))
	\\ & \pla \leq e^{- \lambda t / k} D( (\exp(- \L_0 t / k) \exp(- \tilde{\L} t / k) )^{l-1} (\rho) \| \E_{\sigma *}(\rho)) \pl,
\end{split}
\end{equation*}
where the 1st inequality follows from data processing and the assumption that $\exp(- \L_0 \tau) \circ \E_{\sigma *} = \E_{\sigma *}$ for all $\tau \in \RR^+$ and the second inequality from the assumed (C)MLSI of $\tilde{\L}$.
\end{proof}

\section{Data Availability Statement}
Data sharing is not applicable to this article as no new data were created or analyzed in this study.